\documentclass[12pt, twoside, reqno]{amsart}
\usepackage[thm_section]{macros}
\usepackage{nicefrac}

\begin{document}
\begin{titlepage}
\begin{spacing}{1}
\title{\textbf{\Large Reproducible Aggregation of Sample-Split Statistics$^*$}}
\author{
\begin{tabular}[t]{c@{\extracolsep{4em}}c} 
\large{David M. Ritzwoller} &  \large{Joseph P. Romano}\vspace{-0.7em}\\ \vspace{-1em}
\small{Stanford University} & \small{Stanford University} \\ \vspace{-0.7em}
\end{tabular}%
\\
}
\date{%
\today\\ $^*$Email: ritzwoll@stanford.edu, romano@stanford.edu. We thank Sourav Chatterjee, Jiafeng Chen, Han Hong, Guido Imbens, Lihua Lei, Evan Munro, Aaditya Ramdas, Brad Ross and audiences at the Econometric Society North American and European summer meetings for helpful comments. We thank Bhaskar Chakravorty for sharing his data and code. Ritzwoller gratefully acknowledges support from the National Science Foundation under the Graduate Research Fellowship. Computational support was provided by the Data, Analytics, and Research Computing (DARC) group at the Stanford Graduate School of Business (RRI:SCR\_022938).}
                      
\begin{abstract}
\smalltonormalsize{Statistical inference is often simplified by sample-splitting. This simplification comes at the cost of the introduction of randomness not native to the data. We propose a simple procedure for sequentially aggregating statistics constructed with multiple splits of the same sample. The user specifies a bound and a nominal error rate. If the procedure is implemented twice on the same data, the nominal error rate approximates the chance that the results differ by more than the bound. We illustrate the application of the procedure to several widely applied econometric methods.}
\\
\\
\textbf{Keywords:} Sample-splitting, Cross-Fitting, Cross-Validation, Reproducibility
\\
\textbf{JEL:} C01, C13, C52
\end{abstract}
\end{spacing}
\end{titlepage}
\maketitle
\thispagestyle{empty}
\setcounter{page}{1}
\clearpage
\begin{spacing}{1.3}

\section{Introduction\label{eq: introduction}}
Sample-splitting is ubiquitous in modern econometric theory. Routine statistical tasks---model selection, dimension reduction, nuisance parameter estimation---can be implemented on a randomly selected subsample of a data set, without contaminating the validity of a statistical inference produced on its complement. This principle underlies the widely applied practices of cross-validation for predictive risk estimation \citep{stone1974cross,arlot2010survey} and cross-fitting for adaptive estimation of semiparametric models \citep{bickel1982adaptive,schick1986asymptotically,chernozhukov2018double}, among a growing set of additional applications. 

For a fixed data set, statistics constructed with sample-splitting are not deterministic. Two researchers can compute the same statistic on the same data and obtain different values. Researchers are incentivized to report significant results. If there is scope to materially alter the statistics that they report through the choice of the split of their sample, should this choice be left to chance?

This paper makes two contributions. First, we show that many widely applied sample-split econometric methods exhibit significant residual randomness. We give examples from the applied economics literature where the randomness induced by sample-splitting determines the statistical significance of a treatment effect estimated with cross-fitting, the interpretation of a model selected with cross-validation, and the qualitative features of treatment targeting rules learned from cross-fit nuisance parameter estimates.

Second, and accordingly, we propose an efficient method for removing the residual randomness from sample-split statistics. The procedure takes as input a bound and an error rate. The statistic of interest is sequentially aggregated over randomly drawn splits of the sample. The procedure is stopped after an estimate of the residual variation of the aggregate statistic falls below a pre-determined threshold. If the procedure were run twice, we show that the chance that the outputs differ by more than the bound is well-approximated by the error rate. That is, by setting the bound and error rate to be sufficiently small, sample-split statistics aggregated with the procedure are reproducible. 

We begin, in \cref{sec: Reproducible}, by discussing several widely applied sample-split econometric methods and demonstrating that, for each, the randomness induced by sample-splitting can substantively effect results. In \cref{sec: Reproducible Aggregation}, we propose an efficient method for sequentially aggregating sample-split statistics that ensures that the residual randomness is small. We illustrate that, in each of our examples, the proposed method stabilizes results---ensuring reproducibility at a minimal computational expense. We establish that the procedure is valid, in a particular asymptotic sense, without imposing any restrictions on the data generating process or statistic of interest. Similarly, we show that, for a large class of applications, sample-split statistics aggregated with the procedure maintain (or, improve upon) unconditional statistical guarantees. That is, reproducibly aggregated sample-split statistics are still consistent and associated inferences are still valid.

To implement the procedure, a user must make several choices that may affect performance, including the specification of a suitable bound and error rate. Additionally, in most applications, the procedure is applied to stabilize a statistic that is itself constructed with cross-splitting. In these cases, users must also specify how many folds to use for cross-splitting. 

To shed light on these choices, in \cref{sec: Guarantees}, we give an analysis of the performance of the procedure under a set of simplifying conditions. We give two main results. First, the computation needed to achieve a given bound on residual randomness is very sensitive to the desired error tolerance, but is insensitive to the number of folds that are used for cross-splitting. Second, and on the other hand, the accuracy of the nominal error rate of the procedure deteriorates as the number of folds used for cross-splitting increases. We conclude, in \cref{sec: conclusion}, by synthesizing these results into a set of concrete recommendations for practice. We emphasize simple rules-of-thumb for choosing suitable error tolerances and approaches to cross-splitting.

The main theoretical challenge posed by this analysis is the accommodation of the dependence between statistics computed on cross-splits of a sample. We address this through an application of the method of exchangeable pairs \citep{stein1986approximate, ross2011fundamentals,chen2011normal}. To construct an appropriate exchangeable pair for our problem, we develop a novel application of a coupling argument due to \cite{chatterjee2005concentration}, that may be of independent interest.

\subsection{Related Literature\label{sec: related}} 
The procedure studied in this paper is applicable to a large variety of sample-split statistical methods. In \cref{sec: Reproducible}, we give a selective review of various sample-split methods that are frequently used in applied economics. There are many, additional, sample-split methods, proposed in the statistics literature, that have the potential to be useful in econometric applications. Generic sample-split methods for testing statistical hypotheses are studied in \cite{guo2017analysis}, \cite{diciccio2020exact}, and \cite{wasserman2020universal}. Additional applications include sample-split procedures for selective inference \citep{rinaldo2019bootstrapping}, inference on high-dimensional linear models \citep{meinshausen2010stability}, conformal and predictive inference  \citep{lei2018distribution}, and knockoff tests of conditional independence \citep{barber2015controlling}. 

\cite{chernozhukov2018double} and \cite{chernozhukov2018generic} advocate for the aggregation of estimators and $p$-values computed with sample splitting, over a pre-determined number of splits, in the context of applications related to semiparametric estimation and characterization of treatment effect heterogeneity, respectively. We second these recommendations and contribute a general purpose method that provides a statistical guarantee that residual randomness has been controlled up to a specified level of error, at a minimal computational cost.

Our setting is related to a large literature that studies methods for constructing confidence intervals for cross-validated estimates of generalization error. Several examples include \cite{dietterich1998approximate}, \cite{nadeau1999inference}, \cite{lei2020cross}, \cite{bayle2020cross}, \cite{austern2020asymptotics}, and \cite{bates2023cross}. By contrast, we are interested in the randomness conditional on the data.   Formally, our non-asymptotic results are most similar to the Berry-Esseen bounds given in \cite{austern2020asymptotics}, who study the unconditional normal approximation of statistics similar to the those considered in \cref{sec: Guarantees}. Our bounds apply under weaker conditions and are substantially simpler.

Some of our results build on a literature studying the role of algorithmic stability in the accuracy of cross-validation \citep{kale2011cross,kumar2013near}. These papers are related to a broader literature that derives generalization bounds for stable algorithms, originating with \cite{bousquet2002stability}. Some of the concentration inequalities that we derive can be compared to the results of \cite{cornec2010concentration} and \cite{abou2019exponential}. Again, the setting we study is different and our conditions, and resultant bounds, are substantially simpler.

Although our emphasis is on statistics constructed with sample-splitting, the algorithmic and formal methods studied in this paper are potentially applicable to randomized algorithms more generally \citep{motwani1995randomized}. See \cite{beran1987stochastic} for a classical analysis of the asymptotics of randomized tests and estimators.

\section{The Residual Randomness of Sample-Split Statistics\label{sec: Reproducible}}

Sample-splitting has proliferated as a useful subroutine for simplifying various tasks associated with modern statistical inference. The high-level situation is as follows. Consider a researcher who observes the independent data $D = (D_i)_{i=1}^n$ and wishes to report the statistic
\begin{equation}
 \Psi(D, \eta)~,\label{eq: phi infeasible}
\end{equation} 
where $\eta$ is some unknown nuisance parameter.

For example, each observation $D_i$ could contain a measurement of an outcome $Y_i$, a treatment $W_i$, and a vector collecting a large number of covariates $X_i$. The researcher could be interested in measuring the effect of $W_i$ on $Y_i$. To ensure that their estimate is not unnecessarily imprecise, they might like to include controls for only the subset of covariates that are correlated with the outcome. Here, formally, the nuisance parameter $\eta$ collects the indices of the subset of ``relevant'' controls and the statistic $\Psi(D, \eta)$ denotes a treatment effect estimate associated with, say, a regression of the outcome on the treatment with controls for covariates with indices in $\eta$.

In practice, the researcher cannot compute the statistic \eqref{eq: phi infeasible}, as the nuisance parameter $\eta$ is unknown. Often, however, an estimator $\hat{\eta}(D)$ is available and so it may be tempting to report the feasible statistic 
\begin{equation}
 \Psi(D, \hat{\eta}(D))~.\label{eq: phi feasible unsplit}
\end{equation} 
In the example, the estimator $\hat{\eta}(D)$ could collect the indices of the covariates whose absolute sample correlation with the outcome is greater than some pre-determined threshold. 

Statistical inferences that, counterfactually, would be appropriate if the infeasible statistic \eqref{eq: phi infeasible} were available will not necessarily be valid if they were instead based on the feasible statistic \eqref{eq: phi feasible unsplit}. In particular, the process of computing the nuisance parameter estimate $\hat{\eta}(D)$ might produce a confounding effect. In the example, screening control variables according to their correlation with the outcome produces a bias in the resultant treatment effect estimate.

Sample-splitting solves this problem.\footnote{Of course, under certain conditions involving the sparsity of the covariance between the treatment, outcome, and controls, ``double post selection'' of control variables with a Lasso regression \citep{tibshirani1996regression} can produce consistent treatment effect estimates \citep{belloni2014inference}. However, consistent estimates can be obtained under weaker conditions through a closely related approach based on sample-splitting \citep{chernozhukov2018double,belloni2012sparse}.} Let $\mathsf{s}$ denote a randomly drawn subset of the numbers $[n]=\{1,\ldots,n\}$ of size $b$ and let $\tilde{\mathsf{s}}$ denote its complement. Sample-split statistics take the form
\begin{equation}\label{eq: T def intro}
T(\mathsf{s}, D) = \Psi(D_{\mathsf{s}}, \hat{\eta}(D_{\tilde{\mathsf{s}}}))~,
\end{equation}
where the quantities $D_{\mathsf{s}} = (D_i)_{i\in\mathsf{s}}$ and $D_{\tilde{\mathsf{s}}} = (D_i)_{i\in\tilde{\mathsf{s}}}$ collect the data with indices in $\mathsf{s}$ and $\tilde{\mathsf{s}}$, respectively. Splitting the data $D$ into two independent subsamples $(D_{\mathsf{s}}, D_{\tilde{\mathsf{s}}})$ prevents the construction of nuisance parameter estimates from contaminating statistical inferences that the researcher may wish to make with the feasible statistic $\eqref{eq: T def intro}$.

There are two, well-known, practical issues with this approach. First, by splitting the data, statistical precision may be meaningfully reduced. Second, the statistic \eqref{eq: T def intro} is random through both the data $D$ and the choice of the random subset $\mathsf{s}$. That is, if the same sample-split statistic were computed on the same data by two different researchers, and the statistic was sensitive to the choice of the random subset $\mathsf{s}$, then the researchers could report meaningfully different results.

To address these concerns, researchers often aggregate several replications of sample-split statistics through cross-splitting \citep{stone1974cross,schick1986asymptotically}. In particular, researchers typically report aggregate statistics of the form
\begin{equation}\label{eq: cross-split def intro}
a(\mathsf{r}, D) = \frac{1}{k} \sum_{j=1}^k T(\mathsf{s}_j, D)~,
\end{equation}
where the quantity $\mathsf{r} = (\mathsf{s}_j)^k_{j=1}$ denotes a random $k$-fold partition of $[n]$, i.e., a collection of $k$ mutually exclusive sets whose union is equal to $[n]$. By reusing subsamples of the data for computation of both the statistic of interest and the estimation of nuisance parameters, cross-split statistics mitigate potential losses in statistical precision.

In this section, we demonstrate, in several real applications from applied economics, that the second concern---the residual randomness generated by sample-splitting---can substantively affect results. This residual variability is, often, not resolved by cross-splitting. Consequently, in \cref{sec: Reproducible Aggregation}, we propose an approach to aggregating sample-split statistics that controls the scope of residual randomness at a minimal computational cost. Revisiting the applications, we demonstrate that sample-split statistics aggregated through this procedure are reproducible. 

\subsection{Cross-Validation\label{sec: cross-validation}} 

The most prevalent instance of sample-splitting in applied economics is the use of cross-validation for model selection. Often, in this setting, the data $D_i$ consist of a measurement of an outcome $Y_i$ and a vector $X_i$ collecting measurements of $p$ covariates. Interest is in choosing a parsimonious subset of the covariates that, together, best predicts the outcome. 

Lasso regression is a standard approach to this problem \citep{tibshirani1996regression,hastie2015statistical}. The Lasso coefficient is the solution to the regularized least-squares regression
\begin{equation}
\hat{\eta}_{\lambda}(D) = \underset{\eta \in \mathbb{R}^p}{\arg\min} \left\{ \frac{1}{n} \sum_{i=1}^n (Y_i - \eta^\top X_i)^2 + \lambda \sum_{j=1}^p | \eta_j \vert \right\}~,\label{eq: lasso reg}
\end{equation}
where $\lambda$ is some tuning parameter chosen by the user. Penalization of the $\ell_1$-norm of the coefficient $\eta$ encourages sparsity. That is, often, many elements of $\hat{\eta}_{\lambda}(D)$ are exactly equal to zero. The set of covariates associated with non-zero coefficients are referred to as the model ``selected'' by the Lasso. The selected model is sensitive to the choice of the tuning parameter $\lambda$. If $\lambda$ is sufficiently large, no covariates are selected. If $\lambda$ is sufficiently small, all covariates are selected.\footnote{For sufficiently small $\lambda$, a covariate is selected so long as it is not collinear with other covariates and its sample correlation with the outcome is not exactly equal to zero.}

Cross-validation is a widely applied and, perhaps, uncontroversial approach for choosing tuning parameters. Here, cross-validation entails assigning $\lambda$ the value that minimizes the cross-split estimate of the out-of-sample mean-squared error
\begin{equation}
a_{\lambda}(\mathsf{r}, D) = \frac{1}{k} \sum_{j=1}^k T_{\lambda}(\mathsf{s}_j, D)~,
\quad
\text{where}
\quad
T_{\lambda}(\mathsf{s},D)=\frac{1}{\vert \mathsf{s} \vert} \sum_{i\in\mathsf{s}} (Y_i - \hat{\eta}_{\lambda}(D_{\tilde{\mathsf{s}}})^\top X_i)^2\label{eq: k fold cross val}
\end{equation}
and, as before, $\mathsf{r}=(\mathsf{s}_j)^k_{j=1}$ is a random $k$-fold partition of $[n]$. By measuring error out-of-sample, i.e., in the independent subsample $D_{\mathsf{s}}$, the cross-validated risk estimate \eqref{eq: k fold cross val} avoids any ``over-fitting'' bias induced by the estimation of the coefficient $\hat{\eta}_{\lambda}(D_{\tilde{\mathsf{s}}})$. This approach is widely applied throughout various subfields of applied economics.\footnote{See, for example, applications to Crime \citep{mastrobuoni2020crime, arnold2020measuring}, Development \citep{casey2021experiment,blattman2024gang,sadka2024information}, Economic Theory \citep{fudenberg2019predicting}, Education \citep{ellison2021efficiency}, Environment \citep{deryugina2019mortality,cicala2022imperfect}, Finance \citep{koijen2024investors},  Health \citep{abaluck2016determinants,cooper2020surprise}, Industrial Organization \citep{kelly2023mechanics,dube2023personalized}, Innovation \citep{chen2021notching,myers2022estimating}, Labor \citep{muendler2010margins,card2020referees,adermon2021dynastic,derenoncourt2022can}, Market Design \citep{agarwal2019market}, Macroeconomics \citep{hansen2018transparency}, and Political Economy \citep{gentzkow2019measuring,cantoni2022does}.\label{fn: cross val lit}}

The risk estimate \eqref{eq: k fold cross val} is random both through the data $D$ and through the choice of the $k$-fold partition $\mathsf{r}$. That is, the choice of $\lambda$, and thereby, the selected model, have the potential to change for different choices of the random collection $\mathsf{r}$. To evaluate the scope of this sensitivity, we consider data from \cite{casey2021experiment}. \cite{casey2021experiment} study a large-scale experiment, implemented in Sierra Leone, in which a randomly selected subset of parliamentary elections were preceded by direct vote, party-specific primaries. They select covariates to include in various, downstream, econometric analyses with the cross-validated Lasso. In their setting, the outcome $Y_i$ is the vote share, in a poll of party officials, for each of 390 candidates. The vector $X_i$ collects measurements of 48 characteristics for each candidate. 

\begin{figure}[t]
\begin{centering}
\caption{Cross-Validation}
\vspace{-20pt}
\label{fig: lasso demo}
\begin{tabular}{c}
\textit{Panel A: Mean-Squared Error Quantiles}\tabularnewline
\includegraphics[scale=0.4]{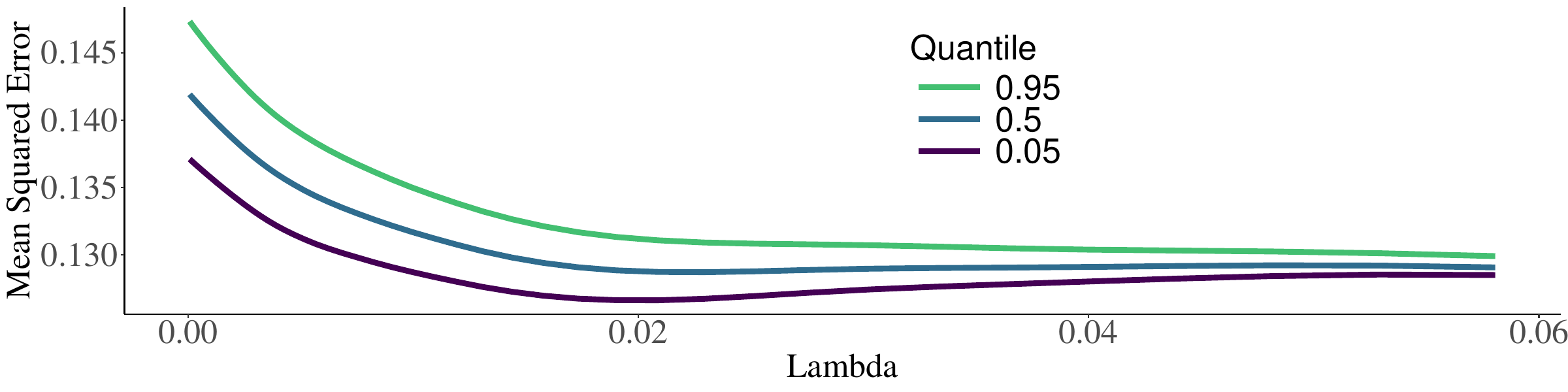}\tabularnewline
\textit{Panel B: Model Variability}\tabularnewline
\includegraphics[scale=0.4]{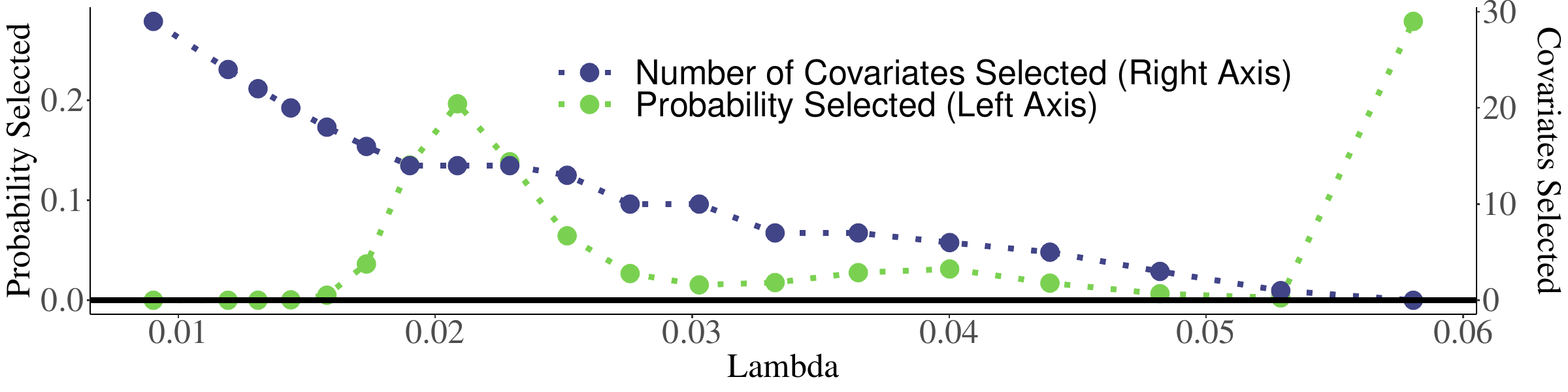}\tabularnewline
\end{tabular}
\par\end{centering}
\medskip{}
\justifying
\noindent{\footnotesize{}Notes: \cref{fig: lasso demo} measures the residual randomness of the cross-validated Lasso, implemented in data from \cite{casey2021experiment}. Panel A displays quantiles of the $10$-fold cross-validated estimate of the mean-squared error \eqref{eq: k fold cross val} of the lasso regression \eqref{eq: lasso reg} over a grid of values of $\lambda$. In Panel B, the probabilities that each value of $\lambda$ minimize the cross-validated risk estimate are displayed with light green dots, relative to the left $y$-axis. The number of covariates selected at each value of $\lambda$ are displayed with dark blue dots, relative to the right $y$-axis. See \cref{app: Casey app} for further details.}{\footnotesize\par}
\noindent\hrulefill
\end{figure}

Panel A of \cref{fig: lasso demo} displays quantiles, across random draws of the $10$-fold partition $\mathsf{r}$, of the cross-validated estimate of the mean-squared error \eqref{eq: k fold cross val} over a grid of values of $\lambda$. The curve associated with the 5th quantile has a minimum around $\lambda = 0.02$. By contrast, the curve associated with the 95th quantile is monotonically decreasing. This induces instability in the value of $\lambda$ chosen by cross-validation. Panel B of \cref{fig: lasso demo} displays, in light green, the probability---again, across random $10$-fold partitions $\mathsf{r}$---that each value of $\lambda$ is selected, i.e., minimizes the cross-validated estimate of the mean-squared error. The number of covariates in the model associated with each value of $\lambda$ is displayed in dark blue. The distribution of the selection probabilities has two maxima---both are associated with probabilities greater than 0.2. One entails selecting a model with 14 covariates. The other entails selecting a model with zero covariates. The random choice of the partition $\mathsf{r}$ substantively affects the size of the selected model.\footnote{\cite{casey2021experiment} take measures to stabilize their results. In particular, they select covariates that are selected in more than 200 of 400 randomly drawn 10-fold partitions $\mathsf{r}$. We use this setting as an example, in part, because a serious attempt was made to address residual randomness. This is not commonplace in the literature surveyed in \cref{fn: cross val lit} (a similar strategy is used in, e.g., \cite{hansen2018transparency}, however).}

\subsection{Cross-Fitting\label{sec: cross-fitting}}Meaningful residual randomness is not particular to cross-validated risk estimation. A second class of sample-split methods, seeing increasing use in applied economics, are characterized by the application of ``cross-fitting''  to accommodate---or characterize---treatment effect heterogeneity. Often, in this case, the data $D_i$ consist of an outcome $Y_i$, a binary treatment $W_i$, and a vector of covariates $X_i$. Let $Y_i(1)$ and $Y_i(0)$ denote the potential outcomes induced by the treatment $W_i$. As before, let $\mathsf{s}$ denote a random a subset of $[n]$, with complement $\tilde{\mathsf{s}}$.

The basic idea underlying this class of methods is to use the data from the units $i$ in $\tilde{\mathsf{s}}$ to predict the treatment effects $Y_i(1)-Y_i(0)$ for the units $i$ in $\mathsf{s}$. For example, a machine learning algorithm can be applied to the data $D_{\tilde{\mathsf{s}}}$ to construct an estimate of the conditional expectation
\begin{equation} \label{eq: conditional expectations}
\mu_w(x) = \mathbb{E}[Y_i \mid X_i = x, W_i = w]
\end{equation}
for each $w$ in $\{0,1\}$. Collect these estimates into $\hat{\eta}(D_{\tilde{\mathsf{s}}}) = (\hat{\mu}_1, \hat{\mu}_0)$. Predictions of the treatment effects for the units $i$ in $\mathsf{s}$ can be constructed through the sample-split statistic 
\begin{equation}\label{eq: cate prediction}
\psi(D_i, \hat{\eta}(D_{\tilde{\mathsf{s}}})) = \hat{\mu}_1(X_i) - \hat{\mu}_0(X_i)~.
\end{equation}
These predictions can then be used to estimate average treatment effects and characterize treatment effect heterogeneity, among other related problems. The point is, by splitting the data, these ``second-stage'' analyses are not confounded by the construction of the ``first-stage'' estimators.  

``Double Machine Learning'' (DML) estimates of average treatment effects are a leading example of a method that takes this structure \citep{chernozhukov2018double}. Here, the sample-split statistics \eqref{eq: cate prediction} are aggregated with cross-splitting, through
\begin{equation}\label{eq: dml definition}
a(\mathsf{r},D)= \frac{1}{k}\sum_{j=1}^k T(\mathsf{s}_j, D)~,
\quad\text{where}\quad 
T(\mathsf{s}, D) = \frac{1}{\vert \mathsf{s}\vert}\sum_{i\in\mathsf{s}} \psi(D_i, \hat{\eta}(D_{\tilde{\mathsf{s}}}))
\end{equation}
and, again, $\mathsf{r}$ is a random $k$-fold partition of $[n]$. An appropriate standard error for \eqref{eq: dml definition} 
is itself given by the cross-split statistic
 \begin{equation}
\label{eq: k fold dml se}
\mathsf{se}(\mathsf{r},D) =   \frac{1}{n} \sqrt{ \sum_{j=1}^k \sum_{i \in  \mathsf{s}_j} \left(\psi(D_i, \hat{\eta}(D_{\tilde{\mathsf{s}}_j})) - a(\mathsf{r},D) \right)^2 }~.
\end{equation}
\cite{chernozhukov2018double} show that an asymptotically efficient test of the null hypothesis that an average effect treatment effect is less than zero can be constructed by comparing \eqref{eq: dml definition} to the critical value $\mathsf{cv}_{\alpha}(\mathsf{r},D) = z_{1-\alpha} \cdot \mathsf{se}(\mathsf{r},D)$, where $z_{1-\alpha}$ is the $1-\alpha$ quantile of the standard normal distribution.\footnote{This result does not apply to estimates of the form \eqref{eq: cate prediction}, but rather to estimators constructed with a ``Neyman Orthogonal'' moment, which, in this case, additionally require a non-parametric estimate of the propensity score.} DML estimators have been increasingly used in applied economics, as they place weaker restrictions on treatment effect heterogeneity than, say, methods based on linear regression.\footnote{See, for example, applications in \cite{okunogbe2022technology}, \cite{beraja2023ai}, \cite{covert2023relinquishing}, \cite{delfino2024breaking}, \cite{farronato2024consumer}.\label{fn: dml list}}

The DML estimator \eqref{eq: dml definition} and standard error \eqref{eq: k fold dml se} are, again, random through both the data $D$ and the $k$-fold partition $\mathsf{r}$. To evaluate the extent of resultant variability, we consider data from \cite{chakravorty2024can}. \cite{chakravorty2024can} use DML to study the effect of a program, implemented in two Indian states, involving the provision of information concerning prospective jobs to vocational trainees, on employment outcomes. Here, the binary outcome $Y_i$ indicates employment five months after completing training for each of 890 trainees placed into jobs, $W_i$ denotes assignment to the program, and $X_i$ collects measurements of 77 pre-treatment covariates. 

Panel A of \cref{fig: cross-fitting demo} displays a heat map of the joint distribution of the $5$-fold cross-fit estimator \eqref{eq: dml definition} and the associated critical value $\mathsf{cv}_{\alpha}(\mathsf{r},D)$  over random draws of the partition $\mathsf{r}$.\footnote{We follow the replication package associated with \cite{chakravorty2024can}. Nuisance parameters are estimated with random forests using the ``Ranger'' R package \citep{wright2017ranger}. Estimates and standard errors are constructed using the ``DoubleML'' R package \citep{bach2024double}. See \cref{app: Chakravorty app} for further details.} A black line is placed at the threshold where the estimate is equal to the critical value. The residual variability in the estimate is large relative to estimates of the sampling variability, and is sufficient to determine the purported statistical significance. Concretely, the difference between the 5th and 95th quantiles of the distribution of the estimator (0.094 and 0.131, respectively) is equal to 68\% of the median of the distribution of the standard error (0.054).\footnote{Some papers take measures to address analogous residual randomness (see e.g., \cite{covert2023relinquishing}). In fact, \cite{chernozhukov2018double} suggest taking the median over several partitions $\mathsf{r}$. Most implementations of DML in standard statistical software allow the user to aggregate estimates, although this aggregation does not occur by default (e.g., \citealp{bach2024double} in R or \citealp{ahrens2024ddml} in STATA).}

\begin{figure}[p]
\begin{centering}
\caption{Cross-Fitting}
\label{fig: cross-fitting demo}
\vspace{-20pt}
\begin{tabular}{cc}
\multicolumn{2}{c}{\textit{Panel A: Treatment Effect Estimation}}\tabularnewline
\multicolumn{2}{c}{\includegraphics[scale=0.4]{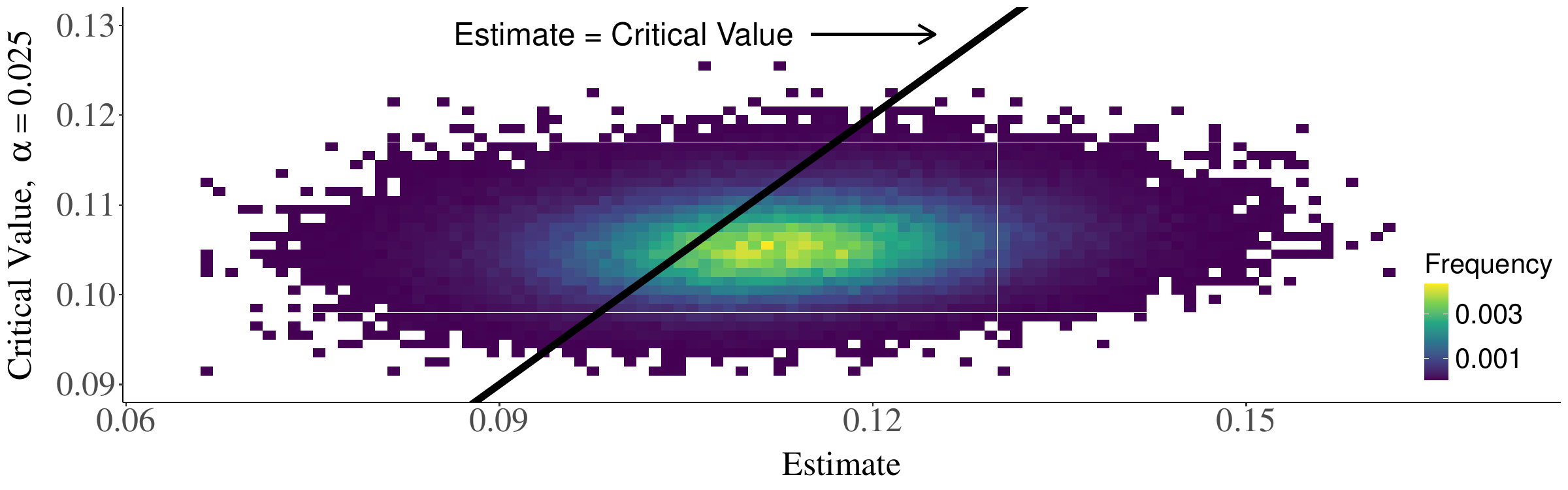}}\tabularnewline
\multicolumn{2}{c}{\textit{Panel B: Policy Evaluation}}\tabularnewline
\multicolumn{2}{c}{\includegraphics[scale=0.4]{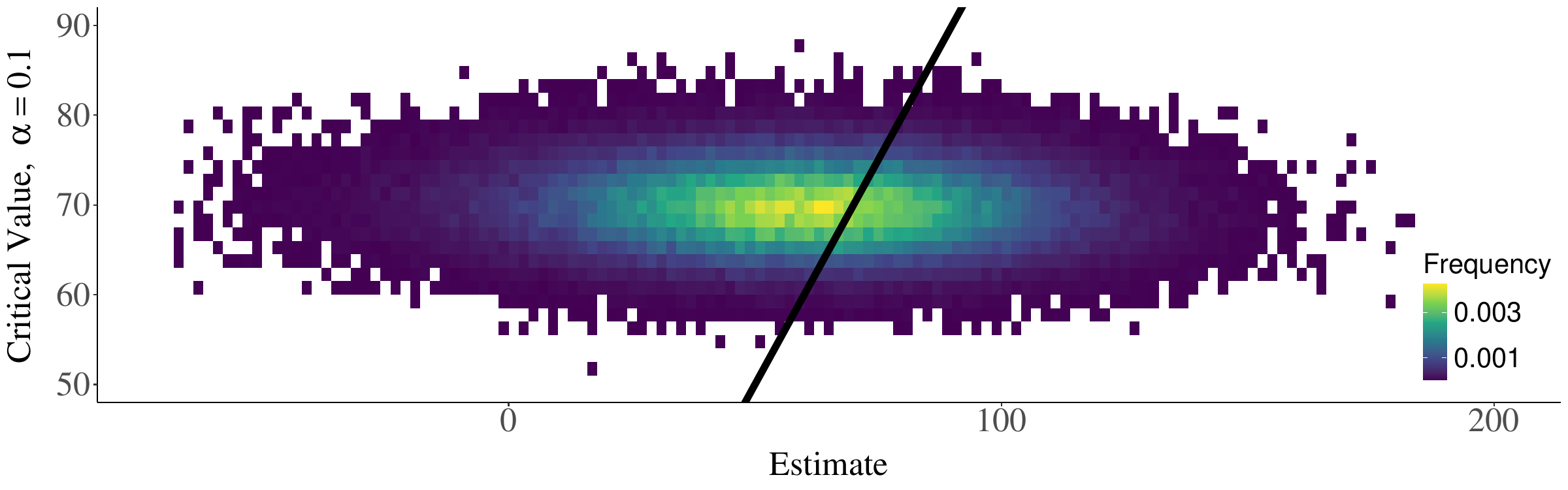}}\tabularnewline
\multicolumn{2}{c}{\textit{Panel C: Testing for Treatment Effect Heterogeneity}}\tabularnewline
\multicolumn{2}{c}{\includegraphics[scale=0.4]{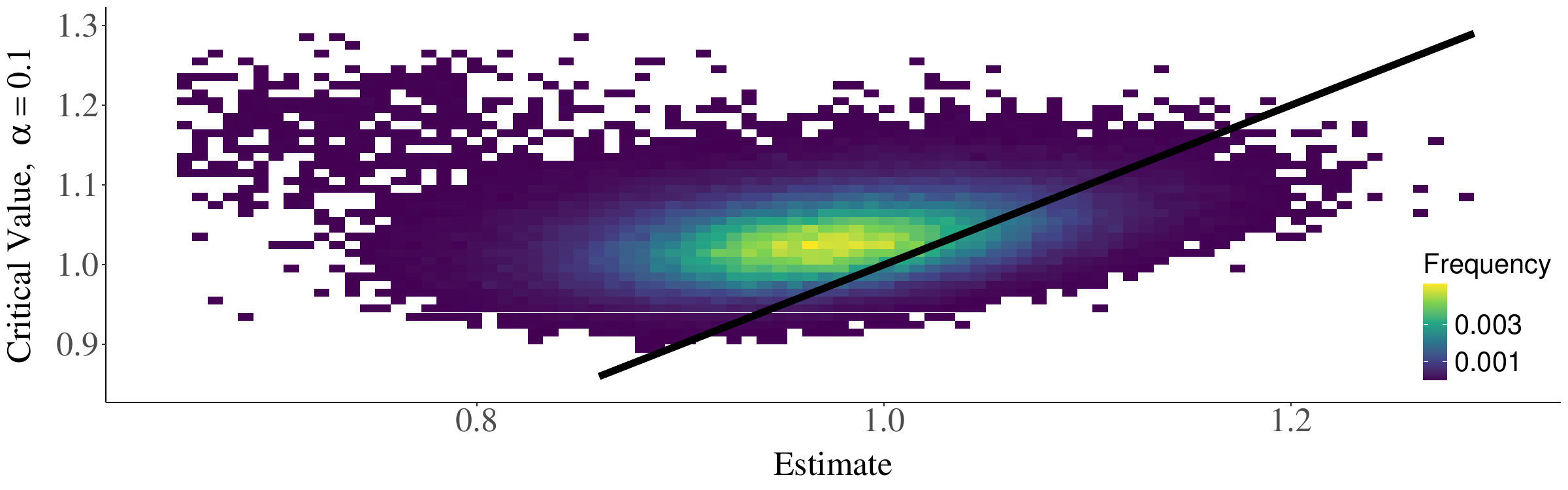}}\tabularnewline
\end{tabular}
\par\end{centering}
\medskip{}
\justifying
\noindent{\footnotesize{}Notes: \cref{fig: cross-fitting demo} displays discretized  heat maps quantifying the residual randomness of three estimators constructed with cross-fitting. All three panels give the joint distribution of an estimator and an associated critical value. In each case, a black line has been placed at the threshold where the estimator is equal to the critical value. Panel A displays the distribution of the DML estimate \eqref{eq: dml definition} cross $5$-fold cross-splits using data from \cite{chakravorty2024can}. Panel B displays the difference between treatment effects estimates for the ``most impacted'' and ``most deprived'' individuals, following \cite{haushofer2022targeting},  using data from \cite{egger2022general}, again across $5$-fold cross-splits. Panel C displays the distribution of the estimate associated with a test of treatment effect heterogeneity using data from \cite{beaman2023selection}. Here, each estimate and critical value is computed by averaging over 250 independently drawn half-samples $\mathsf{s}$. See \cref{app: simulation appendix} for further details on the construction of each panel.}{\footnotesize\par}
\end{figure}

The same behavior is exhibited in related applications that use cross-fitting in more complicated ways. In these cases, measures to stabilize results are more common, although there is little guidance on how to best operationalize this stabilization. To illustrate this, we consider examples from \cite{haushofer2022targeting} and \cite{beaman2023selection}.\footnote{Additional, related, methods that uses cross-fitting to estimate and evaluate treatment targeting rules are proposed by \cite{athey2021policy} and \cite{yadlowsky2024evaluating}.\label{fn: policy design}} \cite{haushofer2022targeting} use data from a randomized cash transfer implemented in Kenya. These data were originally considered in \cite{egger2022general}. \cite{beaman2023selection} use data from an experiment concerning agricultural lending in Mali. Both papers use sample-splitting to construct estimates, of the form \eqref{eq: cate prediction}, of the treatment effects $Y_i(1) - Y_i(0)$ for each of the units $i$ in $\mathsf{s}$. \cite{haushofer2022targeting} additionally construct sample-split estimates of the untreated outcome $Y_i(0)$ for each of the units in $\mathsf{s}$. 

\cite{haushofer2022targeting} identify the 50\% of units in $\mathsf{s}$ that have the largest predicted treatment effect as well as the 50\% of units that have the smallest predicted untreated outcome. They refer to these groups as the ``most impacted'' and ``most deprived,'' respectively. They estimate the difference in the average treatment effect for the two groups, and construct a standard error, and associated critical value, for this difference with the bootstrap. See \cref{app: Egger app} for details. Panel B of \cref{fig: cross-fitting demo} displays a heat map of the joint distribution of the $5$-fold cross-fit estimate of the difference between the treatment effect estimates for the two groups, and the associated critical value, over random draws of the partition $\mathsf{r}$.\footnote{We were not able to access a replication package associated with \cite{haushofer2022targeting}, and so implement a simplified version of the exercise considered in that paper using data from the replication package associated with \cite{egger2022general}. Treatment effects, and untreated outcomes, are estimated with random forests using the ``GRF'' R package \citep{athey2019generalized}. Analogous estimates reported in \cite{haushofer2022targeting} are statistically significant. We emphasize that, as we implement only a simplified version of the exercise conducted in \cite{haushofer2022targeting}, our estimates should only be interpreted an an illustration of the scope of residual randomness in analyses of this type, rather than as a substantive characterization of underlying treatment effect heterogeneity. In particular, we make no attempt to reweigh observations according to their sampling probabilities and appear to be using a different measure of the time between the administration of the experiment and the measurement of post-treatment outcomes.} The residual randomness is considerable. Here, the difference between the 5th and 95th quantiles of the distribution of the estimator (37.82 and 109.69, respectively) is equal to 190\% of the median of the distribution of the standard error (54.44). To stabilize these results, \cite{haushofer2022targeting} average over 400 draws of the $5$-fold partition $\mathsf{r}$. Due to this increased computation, they only report confidence intervals for their main results. The methods developed in this paper ensure that stabilization of sample-split statistics, in this way, occurs at a minimal computational expense.

Likewise, \cite{beaman2023selection} implement a test of treatment effect heterogeneity proposed by \cite{chernozhukov2018generic}. In particular, in data for units in $\mathsf{s}$, the outcome is regressed on the treatment, the treatment effect estimate, and an interaction between the treatment and the treatment effect estimate. The idea is that, if treatment effect estimates are well-calibrated, then the coefficient on the interaction should be statistically greater than zero. \cite{chernozhukov2018generic} recommend aggregating $p$-values associated with the coefficients on the interaction over 250 replications of this sample-split test. Panel C of \cref{fig: cross-fitting demo} displays a heat map measuring the joint distribution of the average coefficient and critical value associated with this test.\footnote{We follow the details of the replication package associated with \cite{beaman2023selection}. Nuisance parameters are estimated with random forests using the ``GRF'' R package \citep{athey2019generalized}. Further details are given in \cref{app: Beaman app}.} That is, each coefficient and critical value is computed by taking the average over 250 sample-splits. Despite this aggregation, the residual randomness remains meaningful.\footnote{\cite{chernozhukov2018generic} recommend aggregating $p$-values the median to ensure robustness to outliers. In practice, if aggregation of a sample-split statistic with a median, rather than a mean, makes a material difference, we strongly encourage researcher to investigate why some sample-splits generate extreme (or highly skewed) estimates (e.g., outliers in the underlying data or poor overlap in an intervention).} The difference between the 5th and 95th quantiles of the distribution of the estimator (0.86 and 1.08, respectively) is equal to 28\% of the median of the distribution of the standard error (0.80). Perhaps motivated by this instability, \cite{beaman2023selection} aggregate over 1000 sample-splits. 

Two themes emerge from these examples. First, sample-splitting is a versatile tool for simplifying various tasks associated with modern statistical inference. Second, the residual randomness induced by sample splitting is often large and can affect the substantive interpretation of results. In the next section, we provide a general-purpose method for aggregating sample-split statistics that ensures that residual randomness is controlled at a minimal computational cost. 

\section{Reproducible Aggregation\label{sec: Reproducible Aggregation}}

We propose a sequential method for aggregating sample-split statistics. Our objective is to ensure that the auxiliary randomness induced by sample-splitting is small. To introduce the method, we require some additional notation. The set $\mathcal{S}_{n,b}$ consists of all subsets of $\left[n\right]=\{1,\ldots,n\}$ of size $b$. In turn, the set $\mathcal{R}_{n,k,b}$ contains all collections of $k$ mutually exclusive elements of $\mathcal{S}_{n,b}$. That is, if $n=k\cdot b$, then the set $\mathcal{R}_{n,k,b}$ collects all partitions of $[n]$ into $k$ mutually exclusive sets of size $b$. We refer to the elements of the set $\mathcal{R}_{n,k,b}$ as cross-splits. Throughout, the quantity $\mathsf{R}_{g,k}=(\mathsf{r}_i)_{i=1}^g$ collects $g$ cross-splits, each given by $\mathsf{r}_i = (\mathsf{s}_{i,j})_{j=1}^k$. Unless otherwise specified, the elements of $\mathsf{R}_{g,k}$ are random, sampled independently and uniformly from the collection $\mathcal{R}_{n,k,b}$.

Suppose that we are interested in some sample-split statistic $T(\mathsf{s},D)$. We study the construction of aggregate statistics of the form
\begin{equation}
a(\mathsf{R}_{g,k}, D)= \frac{1}{g} \sum^g_{i=1} a(\mathsf{r_i}, D)~,
\quad\text{where}\quad
a(\mathsf{r_i}, D) = \mathcal{A}(\{T(\mathsf{s}_{i,j},D)\}_{j=1}^k)
\label{eq: aggregate}
\end{equation}
and the function $\mathcal{A}(\cdot)$ aggregates the statistics $T(\mathsf{s}_{i,j},D)$ across the cross-split $\mathsf{r}_i$. To simplify exposition, for the time being, we will restrict attention to the case that the statistic $a(\mathsf{R}_{g,k}, D)$ is real-valued. This is natural, if, for example, the statistic $a(\mathsf{r}, D)$ is a treatment effect estimate or $p$-value constructed with cross-fitting. Later, in our application to cross-validated risk estimation, we treat vector-valued sample-split statistics, e.g., cross-validated risk estimates queried at a vector of values of a penalization parameter. 

Our task is to formulate a method for choosing the number of cross-splits $g$ to ensure that the residual variability of the aggregate statistic \eqref{eq: aggregate} is small. We formalize this objective as follows. 
\begin{defn}[Reproducible Aggregation]  \label{def: reproducibility} Let the sequences $\{\mathsf{r}_i\}_{i=1}^\infty$ and $\{\mathsf{r}^\prime_{i}\}_{i=1}^\infty$ be drawn independently and uniformly, conditional on the data $D$, from the collection of cross-splits $\mathcal{R}_{n,k,b}$. Define $\mathsf{R}_{g,k} = \{\mathsf{r}_g\}_{i=1}^g $ and $\mathsf{R}^{\prime}_{g,k} = \{\mathsf{r}^{\prime}_g\}_{i=1}^g$ for each integer $g$. Suppose that the integers $\hat{g}$ and $\hat{g}^\prime$ are independent and identically distributed, again conditional on the data $D$. We say that the aggregate statistic $a(\mathsf{R}_{\hat{g},k}, D)$ is $(\xi,\beta)$-reproducible if
\begin{equation}
P\bigg\{ \big\vert a(\mathsf{R}_{\hat{g},k}, D) - a(\mathsf{R}^\prime_{\hat{g}^\prime,k}, D) \big\vert  \geq \xi \mid D\bigg\} \leq \beta \label{eq: Reproducibility}
\end{equation}
almost surely. 
\end{defn}
\noindent \cref{def: reproducibility} is motivated by the following thought experiment. Suppose that the data $D$ are given to two researchers. Each researcher is tasked with producing an estimate of the form \eqref{eq: aggregate}. They generate the collections of splits, $\mathsf{R}_{\hat{g},k}$ and $\mathsf{R}^\prime_{\hat{g}^\prime,k}$, independently using the same, potentially data-dependent, procedure. That is, the two collections of splits are independent and identically distributed, conditional on the data $D$. If the estimate $a(\mathsf{R}_{\hat{g},k}, D)$ is $(\xi,\beta)$-reproducible, then the probability that the two researchers' estimates differ by more than $\xi$ is less than $\beta$.

Of course, ensuring that estimates are reproducible, in the sense of \cref{def: reproducibility}, does not preclude deceptive behavior. Researchers might compute a reproducibility aggregated statistic many times, until a desirable estimate is obtained. However, reporting reproducible statistics greatly increases the cost of searches of this form, as the scope of residual randomness has been reduced.

We propose a sequential method for constructing reproducible sample-split statistics. The proposal is based on the fixed-length sequential confidence intervals of \cite{anscombe1952large} and \cite{chow1965asymptotic}. The procedure works by repeatedly drawing a cross-split $\mathsf{r}_g$ uniformly from $\mathcal{R}_{n,k,b}$, appending the cross-split to the collection $\mathsf{R}_{g,k} = (\mathsf{R}_{g-1,k}, \mathsf{r}_g)$, and estimating the conditional variance
\begin{flalign}
\label{eq: cond var}
v_{g,k}\left(D\right) = \Var( a(\mathsf{R}_{g,k}, D) \mid D)
\end{flalign}
with the plug-in estimator
\begin{flalign}
\label{eq: var estimator}
\hat{v}\left(\mathsf{R}_{g,k}, D\right) & 
=\frac{1}{g}\frac{1}{g-1}\sum_{i=1}^{g} (a(\mathsf{r}_i,D) - a(\mathsf{R}_{g,k},D))^2
\end{flalign}
until the condition 
\begin{equation}
\label{eq: empirical stopping time}
\hat{v}\left(\mathsf{R}_{g,k}, D\right) \leq \mathsf{cv}(\xi,\beta) = \frac{1}{2}\left(\frac{\xi}{z_{1-\beta/2}}\right)^2
\end{equation}
is satisfied, where $z_{\alpha}$ denotes the $\alpha$ quantile of the standard normal distribution. In particular, let $g_{\mathsf{init}}\geq2$ denote some ``burn-in'' period chosen by the user. The number of cross-splits $\hat{g}$ chosen by the procedure is the smallest value of $g$, greater than $g_{\mathsf{init}}$, such that \eqref{eq: empirical stopping time} is satisfied. The procedure is summarized in \cref{alg: sequential aggregation}.

\begin{algorithm}[t]
\caption{Anscombe-Chow-Robbins Aggregation} \label{alg: sequential aggregation}
\KwIn{Data $D$, tolerance $\xi$, error rate $\beta$, collection size $k$, split size $b$, initialization $g_{\mathsf{init}}$}

Set $g \leftarrow g_{\mathsf{init}}$

Draw $\mathsf{r}_1$, ..., $\mathsf{r}_{g_{\mathsf{init}}}$ independently and uniformly from $\mathcal{R}_{n,k,b}$.  Collect $\mathsf{R}_{g_{\mathsf{init}},k} = ( \mathsf{r}_j)_{j=1}^{g_{\mathsf{init}}}$.

\While{$\hat{v}(\mathsf{R}_{g,k}, D) > \mathsf{cv}(\xi,\beta) $}{

Set $g \leftarrow g + 1$

Draw $\mathsf{r}_g$ uniformly from $\mathcal{R}_{n,k,b}$. Collect $\mathsf{R}_{g,k} = (\mathsf{R}_{g-1,k}, \mathsf{r}_g)$.

}
Set $\hat{g} \leftarrow g$

\Return $a(\mathsf{R}_{\hat{g},k}, D)$ 

\nonl \hrulefill

\nonl {\footnotesize{}Notes: \cref{alg: sequential aggregation} gives a method for sequentially aggregating sample-split statistics. The critical value $\mathsf{cv}(\xi,\beta)$ is defined in display \eqref{eq: empirical stopping time}.}{\footnotesize\par}
\end{algorithm}
\smallskip

In \cref{sec: asymptotic validity}, we show that \cref{alg: sequential aggregation} is applicable, off-the-shelf, to a large class of problems. In particular, without imposing any conditions on the data generating process or the statistic of interest, we show that sample-split statistics aggregated with \cref{alg: sequential aggregation} are reproducible, in a particular asymptotic sense. Moreover, we show that, in many applications, statistics aggregated with \cref{alg: sequential aggregation} maintain, or improve on, various unconditional statistical guarantees.  We then revisit, in \cref{sec: practice}, the examples considered in \cref{sec: Reproducible}. We show that, in each case, an application of \cref{alg: sequential aggregation} produces a reproducible, stabilized, estimate. 

\subsection{Asymptotic Validity\label{sec: asymptotic validity}}

The following theorem establishes that statistics aggregated with \cref{alg: sequential aggregation} are reproducible, in a particular asymptotic sense. The proof is closely related to the arguments of \cite{anscombe1952large} and \cite{chow1965asymptotic} and is given in \cref{sec: proof asymptotic}. See e.g., Theorem 3.1 of \citet{gut2009stopped} for a textbook treatment.
\begin{theorem}
\label{eq: general reproducibility}
Suppose that the conditional variance $\Var(a(\mathsf{r}, D)\mid D)$ is strictly positive, almost surely, where $\mathsf{r}$ denotes a random collection drawn uniformly from $\mathcal{R}_{n,k,b}$. If the collections $\mathsf{R}_{\hat{g},k}$ and $\mathsf{R}^\prime_{\hat{g}^\prime,k}$ are independently obtained using \cref{alg: sequential aggregation}, then
\begin{equation}
\label{eq: almost sure reproduce}
P\bigg\{\big\vert a(\mathsf{R}_{\hat{g},k}, D) - a(\mathsf{R}^\prime_{\hat{g}^\prime,k}, D) \big\vert \geq \xi \mid D\bigg\} \overset{\text{a.s}}{\to} \beta
\quad
\text{as}\quad\xi\to0~.
\end{equation}
\end{theorem}
\noindent Two aspects of \cref{eq: general reproducibility} are worthy of emphasis. First, no restrictions on the data generating process or the statistic of interest are imposed. For instance, the data $D$ do not necessarily need to be i.i.d.\ and the statistic $a(\mathsf{r},D)$ can be random though quantities other than just $\mathsf{r}$ and $D$.  The result, thus, provides a strong assurance that \cref{alg: sequential aggregation} can be applied widely. 

Second, in the statement \eqref{eq: almost sure reproduce}, asymptotics are taken as $\xi \to 0$, with all other quantities, including the sample size $n$, fixed. This asymptotic framework is somewhat opaque, or at least nonstandard, as the parameter $\xi$ is a choice variable. The operational interpretation is that, if $\xi$ is chosen to be sufficiently small---at a point negligible relative to, say, the conditional variance $v_{1,k}(D)$---then the nominal reproducibility error $\beta$ is accurate. In \cref{sec: Guarantees}, by imposing some simplifying restrictions, we give a set of non-asymptotic results that make the dependence of the performance of \cref{alg: sequential aggregation} on the choices of $\xi$, $\beta$, and the statistic $a(\mathsf{r}, D) $ more transparent. 

The intuition underlying \cref{eq: general reproducibility} is straightforward. The result follows directly from the observation that the cross-splits $\mathsf{r}_i$ are independent and identically distributed conditional on the data $D$. Thus, for large values of $g$, the variance estimator $\hat{v}\left(\mathsf{R}_{g,k}, D\right)$ will be close to the conditional variance $v_{g,k}(D)$. As a consequence, if $\xi$ is sufficiently small, then the number of cross-splits $\hat{g}$ chosen by the procedure will be close to the ``oracle'' stopping time
\begin{align}
g^{\star}
& =
\underset{g \geq g_{\mathsf{init}} }{\arg\min}
\left\{ 
\mathsf{cv}(\xi,\beta) = \frac{1}{2}\left(\frac{\xi}{z_{1-\beta/2}}\right)^{2}
\geq
v_{g,k}(D) \right\} \label{eq: oracle splits main}\\
& = 
\underset{g \geq g_{\mathsf{init}} }{\arg\min}
\left\{ 
 \xi
 \geq 
z_{1-\beta/2}
\sqrt{\frac{2 v_{1,k}(D)}{g} }
 \right\}~,\label{eq: rewrite g star}
\end{align}
where we have used the fact that $v_{g,k}(D) = g^{-1}v_{1,k}(D)$ in writing \eqref{eq: rewrite g star}. This, then, ensures that the aggregate statistic is approximately $(\xi,\beta)$-reproducible, as 
\begin{flalign*}
 & P\left\{\big\vert a(\mathsf{R}_{g^\star,k},D)-a(\mathsf{R}_{g^\star,k}^{\prime},D)\big\vert 
 \geq 
 \xi\mid D\right\} \\
 & =P\left\{ \bigg\vert\sqrt{\frac{g^{\star}}{2 v_{1,k}(D)}}\left(\frac{1}{g^{\star}}\sum_{i=1}^{g^{\star}}a(\mathsf{r}_{i},D)-a(\mathsf{r}_{i}^{\prime},D)\right)\bigg\vert
 \geq
 z_{1-\beta/2}\mid D\right\} \overset{\text{a.s}}{\to} \beta 
\end{flalign*}
as $\xi\to0$, where the limit follows from the central limit theorem (as $g^{\star} \to \infty$ as $\xi\to0$).\footnote{The discrepancy between $a(\mathsf{R}_{\hat{g},k}, D)$ and $a(\mathsf{R}_{g^\star,k}, D)$ can be shown to be negligible by applying an appropriate maximal inequality---an idea due to \citet{renyi1957asymptotic}.}

In most cases, sample-split statistics are used because they have some desirable property. For example, a sample-split estimator $a(\mathsf{r}, D)$ might be asymptotically normal or might perform well in terms of some loss function. The following result---roughly, a sequential version of Jensen's inequality---can be applied to show that, in many situations, statistics aggregated with \cref{alg: sequential aggregation} inherit these properties. The proof is given in \cref{sec: proof asymptotic} and follows from an application of the martingale stopping theorem.
\begin{theorem}\label{thm: sequential Jensen}
If the collection $\mathsf{R}_{\hat{g},k}$ is obtained with \cref{alg: sequential aggregation}, the cross-split $\mathsf{r}$ is a random element of the set $\mathcal{R}_{n,k,b}$, and $f(\cdot)$ is any convex function, then
\begin{equation}
\mathbb{E}\left[ f( a(\mathsf{R}_{\hat{g},k}, D) ) \right] \leq  \mathbb{E}\left[ f( a(\mathsf{r}, D) )\right]~.
\end{equation}
\end{theorem}
\noindent Suppose that the estimator $a(\mathsf{r}, D)$ has good performance in terms of some convex loss. \cref{thm: sequential Jensen} demonstrates that the aggregate estimator $a(\mathsf{R}_{\hat{g},k}, D)$ will also perform well (and, in fact, may perform better). For example, \cite{chetverikov2021cross} demonstrate that the cross-validated lasso has nearly optimal rates of convergence in sparse regression problems. \cref{thm: sequential Jensen} shows that all higher-order moments of risk estimates obtained with our procedure are smaller than those of risk estimates obtained from a single cross-split. Thus, we expect the same result to hold for tuning parameters chosen with risk estimates aggregated with \cref{alg: sequential aggregation}.

Alternatively, in many cases, sample-split statistics are shown to be asymptotically normal by demonstrating that the variance of the discrepancy
\begin{equation}\label{eq: cheb diff main text}
\sqrt{n}\left(a(\mathsf{r}, D) - \frac{1}{n} \sum_{i=1}^b\psi(D_i, \eta)\right)
\end{equation}
is small, where $\psi(D_i, \eta)$ is some score, or influence, function and $\eta$ is some unknown nuisance parameter.\footnote{That is, if second term in \eqref{eq: cheb diff main text} is asymptotically normal and the variance of \eqref{eq: cheb diff main text} converges to zero, then the statistic $a(\mathsf{r}, D)$ is also asymptotically normal by Chebychev's inequality and the continuous mapping theorem.} This is how \cite{chernozhukov2018double} demonstrate that DML estimates of average treatment effects are asymptotically normal, for example. \cref{thm: sequential Jensen} demonstrates that the same argument will apply to statistics aggregated with \cref{alg: sequential aggregation}.

In some cases, substantial residual randomness in a sample-split statistic might cause concern for the validity of associated inferences. For example, for DML estimators, the variance of the term \eqref{eq: cheb diff main text} is bounded from below by the expectation of its variance conditional on the data. So, if the conditional variance is large, then the unconditional variance of \eqref{eq: cheb diff main text} is likely not small. On the other hand, the conditional variance can be reduced through the application of our procedure. In other cases, substantial residual randomness may not be a reason for concern, because inferences are based on estimating a nuisance parameter in one split of a data set, and then conducting inference in the other split by conditioning on this estimate. This is the case for the applications considered in  \cite{beaman2023selection} and \cite{haushofer2022targeting}.

There is a large statistical literature on the use and interpretation of average $p$-values. In general, a level $\alpha$ test can be constructed by comparing an average $p$-value to $2\cdot\alpha$ \citep[see e.g. ][]{ruger1978maximale,vovk2020combining,diciccio2020exact}.\footnote{A variety of methodological papers use this observation to construct sample-split hypothesis tests that control the Type I error rate under very general conditions.  In \cref{app: test}, we show that the sample-split hypothesis tests considered in e.g., \cite{diciccio2020exact},  \cite{meinshausen2009p}, and \cite{wasserman2020universal} continue to be valid when they are constructed sequentially with \cref{alg: sequential aggregation}.} Using a result analogous to \cref{thm: sequential Jensen}, for a pre-determined number of sample-splits $g$, \cite{chernozhukov2018generic} show that, under some high-level conditions, valid $p$-values aggregated by averaging over sample-splits are also valid $p$-values. Under the same conditions, \cref{thm: sequential Jensen} can be applied to give an analogous result for $p$-values aggregated with \cref{alg: sequential aggregation}. 

\subsection{Reproducible Aggregation in Practice\label{sec: practice}} Equipped with \cref{alg: sequential aggregation}, we return to the examples considered in \cref{sec: Reproducible}. We show that, in each case, an appropriate application of the procedure produces a stabilized, reproducible estimate. 

\subsubsection{Cross-Fitting\label{sec: cross-fitting 2}} 

We begin by treating the three examples that use cross-fitting. The most important choice to make when implementing \cref{alg: sequential aggregation} is the specification of the statistic $a(\mathsf{r}, D)$. There are many reasonable choices that one might make here. For example, one could choose to ensure that estimates or standard errors are stable up to a desired level of precision. For the sake of comparability across settings, we opt to stabilize the $p$-value
\begin{equation}\label{eq: average p val exhibit}
a(\mathsf{r_i}, D) = 1 - \Phi\left(\frac{\mathsf{est}(\mathsf{r}_i, D) }{\mathsf{se}(\mathsf{r}_i, D) }\right)~,
\end{equation}
where the function $\Phi(\cdot)$ is the standard normal c.d.f.\ and the quantities $\mathsf{est}(\mathsf{r}_i, D)$ and $\mathsf{se}(\mathsf{r}_i, D)$ denote the cross-split estimate and standard error that corresponding to the axes of \cref{fig: cross-fitting demo}. In effect, controlling the residual randomness of the average $p$-value ensures that the residual randomness in the aggregate estimate is small relative to the sampling error.

\cref{tab: practice} summarizes the results. We choose $\xi$ equal to either 0.001 and 0.01, as these are the levels of precision relevant for the determination of statistical significance at levels $\alpha = 0.025$ and $\alpha = 0.10$, respectively. For now, we follow the approaches to cross-splitting taken by all three papers. In the applications to \cite{chakravorty2024can} and \cite{haushofer2022targeting}, we aggregate over $5$-fold cross splits. In the application to \cite{beaman2023selection}, the split $\mathsf{r}_i$ contains a single half-sample, i.e., subset of $[n]$ of size $n/2$.\footnote{If it is desirable to aggregate the median $p$-value, \cref{alg: sequential aggregation} will continue to apply with a small modification. In particular, the variance estimator \eqref{eq: var estimator} can be replaced by an estimator constructed with the bootstrap. In this case, a result analogous to \cref{eq: general reproducibility} will continue to hold.} The estimates displayed in \cref{tab: practice} indicate that all three applications require aggregation over hundreds of cross-splits to ensure reproducibility, in the sense of \cref{def: reproducibility}, at these levels of error tolerance. Moreover, we find that nominal error rate $\beta$ associated with \cref{alg: sequential aggregation}, here set to $0.05$, is very accurate.

\begin{table}
\begin{centering}
\caption{Reproducible Aggregation in Practice: Cross-Fitting\label{tab: practice}}
\begin{tabular}{lcccc}
\toprule 
Application & $\xi$ & $p$-value & Average $\hat{g}$ & Reproducibility\tabularnewline
\midrule
\midrule 
\cite{chakravorty2024can} & 0.001 & 0.979 & 860.7 & 0.945 \tabularnewline
\cite{haushofer2022targeting}  & 0.01 & 0.80 & 1360.6  & 0.947 \tabularnewline
\cite{beaman2023selection} & 0.01 & 0.83  & 2794.8 & 0.950 \tabularnewline
\bottomrule
\end{tabular}
\par\end{centering}
\medskip{}\medskip{}
\justifying
\noindent{\footnotesize{}Notes: \cref{tab: practice} summarizes the application of \cref{alg: sequential aggregation} to three of the examples considered in \cref{sec: Reproducible}. The first column indicates the application. The second column specifies choices for the error tolerance $\xi$. In each case, we set the reproducibility error rate to $\beta = 0.05$ and the burn-in sample size $g_{\mathsf{init}}$ to 10. The third column gives the $p$-value produced by one-application of the procedure. The fourth column gives the average number of cross-splits $\hat{g}$ drawn in each implementation, taken across 2,000 replications of the aggregation procedure. The fifth column displays an estimate of the true reproducibility probability, i.e., the probability that two independent implementations of \cref{alg: sequential aggregation} produce estimates that differ by less than $\xi$, computed with these replicates.}{\footnotesize\par}
\noindent\hrulefill
\end{table}
 
\subsubsection{Cross-Validation\label{sec: cross val in practive}}

We now turn to our application to cross-validated risk estimation. In this setting, there is more ambiguity in how to best apply \cref{alg: sequential aggregation}. We find that the following approach works well in the data from \cite{casey2021experiment}. Recall that, in this application, the sample-split statistic $a_{\lambda}(\mathsf{r}, D)$ denotes the cross-validated estimate of the mean-squared error, queried at a specified value of the regularization parameter $\lambda$. We are interested in stabilizing these estimates for each value of $\lambda$ in an increasing sequence $(\lambda_l)_{l=1}^p$, where $\lambda_p$ is the smallest value such that no covariates have non-zero coefficients when estimated using the full data.\footnote{Throughout, we use the grid $(\lambda_l)_{l=1}^p$ of values of the regularization parameter $\lambda$ chosen by default by the ``glmnet'' R package \citep{friedman2021package}.} 

To do this, we apply \cref{alg: sequential aggregation} to each component of the vector
\begin{align}
(b_{\lambda_1}(\mathsf{r}, D), \ldots, b_{\lambda_{p-1}}(\mathsf{r}, D))~,
\quad\text{where}\quad
b_{\lambda_i}(\mathsf{r},D) = ( a_{\lambda_i}(\mathsf{r}, D) - a_{\lambda_p}(\mathsf{r}, D) ) ~.\label{eq: mse difference}
\end{align}
The rationale for this is that the relative, rather than absolute, values of the risk estimates are what are relevant for determining where the estimates takes their minimum value. We find that, in practice, consideration of the differences \eqref{eq: mse difference} can substantially reduce the amount of aggregation needed for stabilization.

We implement \cref{alg: sequential aggregation}, independently, for each component of the vector \eqref{eq: mse difference}. In other words, we ensure that each component of the risk estimate is marginally reproducible.\footnote{In principle, it is straightforward to ensure that the components of the vector \eqref{eq: mse difference} are simultaneously reproducible. For example, simultaneous reproducibility at level $\beta$ can be obtained by ensuring that each component is reproducible at level $\beta / (p-1)$, mimicking the standard Bonferonni adjustment. More sophisticated schemes, based on the bootstrap, say, are also applicable. In practice, ensuring simultaneous reproducibility can substantially increase the required computation and tends not to lead to materially different risk estimates.} We let $\xi_i$ denote the error tolerance used for the $i$th component of the vector \eqref{eq: mse difference}. We use a simple approach for determining suitable values for these tolerances. Roughly speaking, we compute the statistic \eqref{eq: mse difference} for each element of a small, initial sample of cross-splits (e.g., $g=20$). Using these estimates, we find that the level of precision needed to distinguish between the close-to-optimal values of $\lambda$ is approximately $\xi_i = 10^{-4}$. The error tolerances specified at values of $\lambda$ that can be determined to be sub-optimal are set to larger values. Further details are given in \cref{app: lasso xi determine}. 

\cref{fig: lasso reproduce} gives measurements of the performance of \cref{alg: sequential aggregation} for aggregation of the statistic \eqref{eq: mse difference}, implemented in the data from \cite{casey2021experiment}. Panel A displays quantiles of the reproducibly aggregated statistic \eqref{eq: mse difference} across replications of the procedure. The $y$-axis has been truncated to focus attention on the close-to-optimal values of the regularization parameter. There is little residual randomness. Panel B displays quantiles of the number of $10$-fold  cross-splits $\hat{g}$ chosen by the procedure at each value of the regularization parameter. Precise estimates at the close-to-optimal values require aggregation over thousands of cross-splits. Observe that less aggregation is used at small, sub-optimal values of $\lambda$, as we have used larger error tolerances for these values. We find that, if cross-validated model selection is implemented twice, and, in each case, aggregated at this level of error tolerance, then the same value of $\lambda$ is selected with probability 0.84 and the same 14 covariates are selected with probability 0.999. In \cref{app: lasso xi determine}, we show that the nominal reproducability error, again $\beta = 0.05$, is very accurate over the full range of the regularization parameter.

\begin{figure}[t]
\begin{centering}
\caption{Reproducible Aggregation in Practice: Cross-Validation}
\vspace{-20pt}
\label{fig: lasso reproduce}
\begin{tabular}{c}
\textit{Panel A: Mean-Squared Error Differences}\tabularnewline
\includegraphics[scale=0.4]{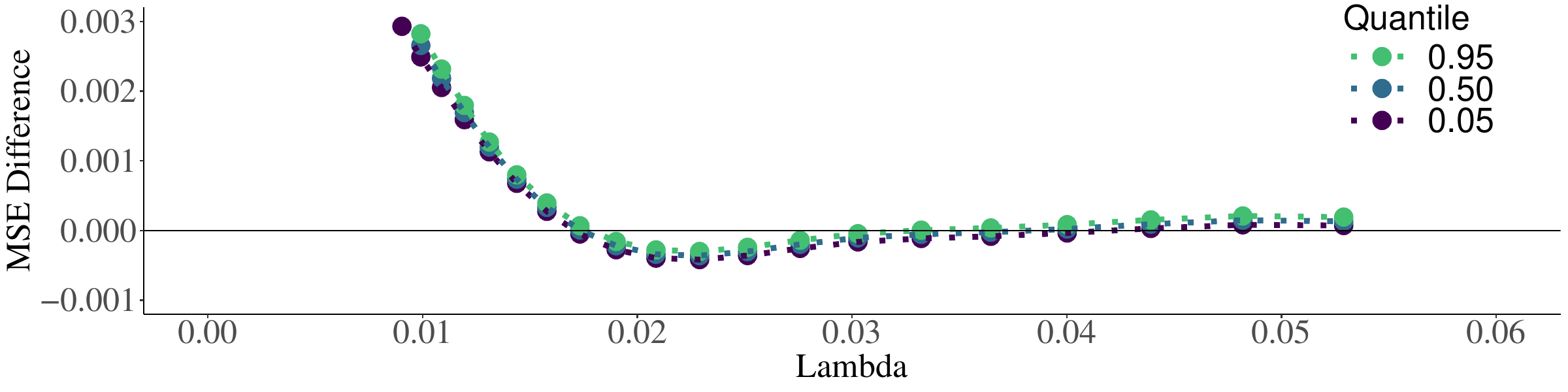}\tabularnewline
\textit{Panel B: Computation}\tabularnewline
\includegraphics[scale=0.4]{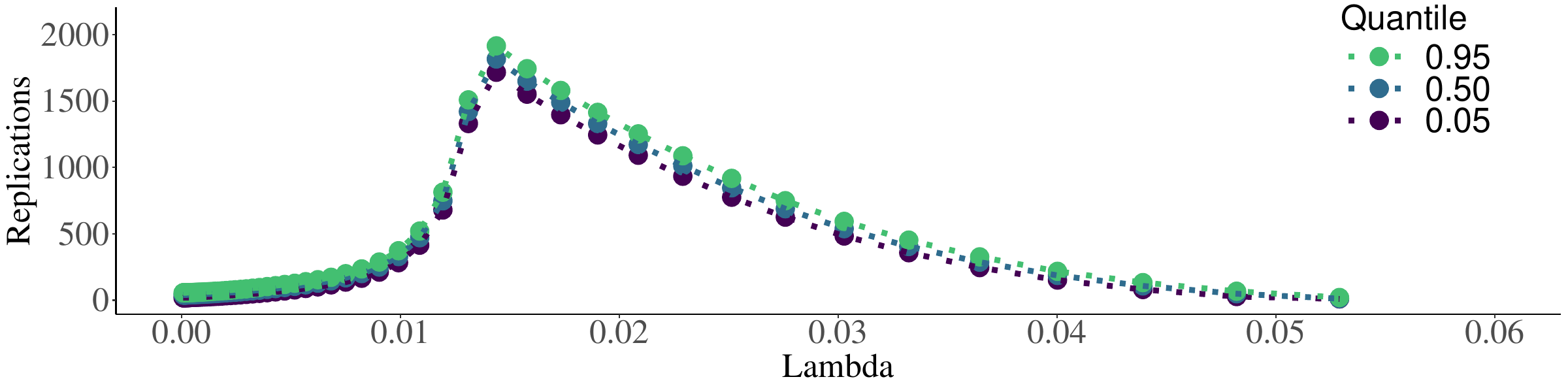}\tabularnewline
\end{tabular}
\par\end{centering}
\medskip{}
\justifying
\noindent{\footnotesize{}Notes: \cref{fig: lasso reproduce} displays the performance of \cref{alg: sequential aggregation} in the application to cross-validated Lasso, implemented in data from \cite{casey2021experiment}. \cref{alg: sequential aggregation} is applied independently to each component of the vector \eqref{eq: mse difference}. The error tolerances $\xi_i$ are specified in \cref{app: lasso xi determine}. We set the nominal reproducibility error to $\beta = 0.05$. Panel A displays quantiles of the statistic \eqref{eq: mse difference} at each value of a grid of values of the regularization parameter $\lambda$. The $y$-axes is truncated to focus attention on large values of $\lambda$. We give an un-truncated version in \cref{app: lasso xi determine}. Panel B displays quantiles of the number of $10$-fold  cross-splits $\hat{g}$ chosen by the procedure at each value of the regularization parameter.}{\footnotesize\par}
\noindent\hrulefill
\end{figure}

\section{Theoretical Analysis\label{sec: Guarantees}}

The asymptotic results given in \cref{sec: Reproducible Aggregation} are quite general. In particular, \cref{eq: general reproducibility} holds in the absence of any restrictions on the data generating process or statistic under consideration. It is worth asking, however, whether these results confer a clear statistical understanding. At least two issues arise. First, \cref{eq: general reproducibility} relies entirely on the fact that, in \cref{alg: sequential aggregation}, successive cross-splits are sampled independently. That is, we have said nothing, yet, about the role of cross-splitting. Second, the interpretation of the asymptotic approximation with $\xi\to0$ is somewhat opaque. What would be a reasonable value of $\xi$ to choose to ensure that \cref{alg: sequential aggregation} is accurate? And how do these choices impact the amount of computation required to implement the procedure?

In this section, we give a non-asymptotic description of the performance of \cref{alg: sequential aggregation}. This is accomplished by placing some simplifying restrictions on the statistic of interest. The more specialized analysis that follows is aimed at providing qualitative and quantitative intuition for how the computational cost and reproducibility error of statistics aggregated with \cref{alg: sequential aggregation} depend on the choices of $k$ and $\xi$. Proofs for results stated in this section are given in \cref{sec: general results}. 

\subsection{Symmetry, Linearity, and Stability\label{sec: sls}}

We impose a set of simplifying restrictions.  Recall that, in general, we are considering the aggregation of cross-split statistics of the form
\begin{equation}
a(\mathsf{r}, D) = \mathcal{A}(\{T(\mathsf{s}_{j},D)\}_{j=1}^k)~,
\quad\text{where}\quad 
T\left(\mathsf{s}, D \right)=
\Psi \left(D_\mathsf{s}, \hat{\eta} \left(D_{\tilde{\mathsf{s}}}\right)\right)
\end{equation}
is a sample-split statistic and the function $\hat{\eta}(\cdot)$ is an estimator of an unknown nuisance parameter $\eta$. In the main text, we restrict attention to the case that $n = k\cdot b$, i.e., where each element $\mathsf{r}$ in $\mathcal{R}_{n,k,b}$ is a complete partition of $[n]$ into $k$ sets of size $b$. Each of the results given here will follow directly from more general results stated in \cref{sec: general results}, where this restriction is not imposed.

First, we assume that the sample-split statistic under consideration is symmetric and deterministic in each part of a split sample.
\begin{assumption}[Symmetry and Determinism]
\label{assu: invariance}For all sets $\mathsf{s}$ in $\mathcal{S}_{n,b}$
and data $D$, the statistic $T\left(\mathsf{s},D\right)$ is deterministic and invariant
to permutations of the data with indices in $\mathsf{s}$ and $\tilde{\mathsf{s}}$, respectively. 
\end{assumption}
\noindent The intention of \cref{assu: invariance} is to restrict the residual randomness under consideration to the randomness introduced by sample-splitting. This holds in cases where $\hat{\eta}(\cdot)$ is deterministic, e.g., when $\hat{\eta}(\cdot)$ is a coefficient vector determined by a regularized regression or in the applications to hypothesis testing considered by \cite{diciccio2020exact} or \cite{wasserman2020universal}. \cref{assu: invariance} rules out procedures where the estimator $\hat{\eta}(\cdot)$ is random conditional on the data. This excludes settings where, e.g., $\hat{\eta}(\cdot)$ is estimated with stochastic gradient descent, bagging or subsampling, or is itself constructed with data splitting. 

Second, we assume that the aggregation function $\mathcal{A}(\cdot)$ is an average and that the statistic $T\left(\mathsf{s}, D \right)$ is linearly separable in the first part of the split sample.
\begin{assumption}[Linearity]
\label{assu: linearity}For all cross-splits $\mathsf{r} = (\mathsf{s}_j)_{j=1}^k$ in $\mathsf{R}_{n,k,b}$, the cross-split statistic $a(\mathsf{r}, D)$ can be represented by
\begin{equation}
\label{eq: linear agg}
a(\mathsf{r}, D)
= \frac{1}{k} \sum_{i\in\mathsf{s}} T\left(\mathsf{s}_j, D \right)~.
\end{equation}
Moreover, for all sets $\mathsf{s}$ in $\mathcal{S}_{n,b}$, the sample-split statistic $T\left(\mathsf{s}, D \right)$ can be represented by
\begin{equation}
\label{eq: linear sep}
T\left(\mathsf{s},D\right) 
= \frac{1}{b} \sum_{i\in\mathsf{s}} \psi(D_i, \hat{\eta} \left(D_{\tilde{\mathsf{s}}}\right))
\end{equation}
for some function $\psi(\cdot,\cdot)$.
\end{assumption}
\noindent We make these restrictions to ease exposition. \cref{assu: linearity} is satisfied if, for example, the statistic under consideration is a cross-fit treatment effect or cross-validated mean-squared error estimate. In principle, \cref{assu: linearity} rules out some applications of interest. In practice, so long as the linear representations \eqref{eq: linear agg} and \eqref{eq: linear sep} hold up to a suitable degree of approximation, the qualitative and quantitative predictions of the results that follow will continue to hold.\footnote{Extensions of our results to cases where the function $\mathcal{A}(\cdot)$ or the statistic $T\left(\mathsf{s}, D \right)$ satisfy component-wise Lipschitz or bounded differences conditions, say, are feasible, and will exhibit the same qualitative behavior.} In particular, we show below that the predictions of our results play out in the data from \cite{chakravorty2024can}, where the statistic $T\left(\mathsf{s},D\right) $ is a $p$-value of the form \eqref{eq: average p val exhibit}.

The remaining assumptions, and our ensuing results, are expressed in terms of two objects that measure the sensitivity of the statistic under consideration to perturbations of the data and of the splits, respectively. We refer to these objects as stabilities. They are defined as follows.
\begin{defn}[Sample Stability] \label{def: sample stability}
Fix a set $\mathsf{s}\subseteq\mathcal{S}_{n,b}$ and let $i$ be an arbitrary element of $\mathsf{s}$.  Let $D^\prime$ denote an independent and identical copy of the data $D$. For each $\mathsf{q}\subseteq [n]$, let $\tilde{D}^{(\mathsf{q})}$ be constructed by replacing $D_j$ with $D^\prime_j$ in $D$ for each $j$ in $\mathsf{q}$. Let $\mathsf{q}$ be a randomly selected subset of $\tilde{\mathsf{s}}$ of cardinality $q$.  We refer to the quantity
\begin{align}
\sigma^{(r,q)}
&=\mathbb{E}\left[
\big\vert
\psi(D_i, \hat{\eta} \left(D_{\tilde{\mathsf{s}}}\right))
-
\psi(D_i, \hat{\eta} (\tilde{D}^{(\mathsf{q})}_{\tilde{\mathsf{s}}}))
\big\vert^{r}\right]
\end{align}
as the $(r,q)$-order sample stability.
\end{defn}
\begin{defn}[Split Stability] \label{def: split stability}
We refer to the quantity
\begin{equation}
\zeta^{(r)}=\mathbb{E}\left[\max_{\mathsf{s},\mathsf{s}^{\prime}\in\mathsf{S}_{n,b}}\left(T\left(\mathsf{s},D\right)-T\left(\mathsf{s}^{\prime},D\right)\right)^{r}\right].
\end{equation}
as the $r$th-order split stability.
\end{defn} 

We restrict attention to statistics whose sample stabilities decay in a suitable way with the sample size $n$. Throughout, we say $x\lesssim y$ if there exists a universal constant $C$ such that $x\leq Cy$.
\begin{assumption}[Sample Stability Decay]\label{def: split stability}
The $(r,q)$-order sample stability satisfies the bound
\begin{equation} \label{eq: training bound}
\sigma^{(r,q)} \lesssim \left(\frac{\sqrt{q}}{n-b}\right)^r
\end{equation}
uniformly for each $q$ in $[b]$ and $r$ in $\{2,4\}$.
\end{assumption} 
\noindent We call a statistic sample stable if it satisfies \cref{def: split stability}.\footnote{The $(2,1)$-order sample stability $\sigma^{(2,1)}$ is a widely studied object in the statistical learning literature, where it is referred to as mean-square stability \citep[see e.g.,][]{bousquet2002stability,kale2011cross,kumar2013near}.} Many sample-split statistics of interest are sample stable. In \cref{app: M estimation}, we show that statistics satisfying \cref{assu: linearity} are sample stable if the nuisance parameter estimator $\hat{\eta}(\cdot)$ is an empirical risk minimizer of a, potentially regularized, strictly convex loss. There is a large literature that gives analogous bounds for other standard machine learning estimators, including bagged or subsampled estimators, like random forests \citep{chen2022debiased,ritzwoller2024uniform}, ensemble estimators \citep{elisseeff2005stability}, and estimators computed with stochastic gradient descent \citep{hardt2016train}.\footnote{\cite{chen2022debiased} show that, under some regularity conditions, sample-splitting is unnecessary for the consistency and asymptotic normality of DML estimators, if a condition related to, but partially stronger than, \cref{def: split stability} is satisfied. We comment on the relationship between \cref{def: split stability} and the conditions considered in \cite{chen2022debiased} in \cref{app: M estimation}.\label{fn: chen compare}} 

Nevertheless, sample-stability should be viewed as a strong assumption, that is only applicable to highly regular estimators. Below, we show that the predictions that follow from the imposition of sample-stability, concerning the qualitative behavior of the residual randomness, play out in the applications to \cite{casey2021experiment} and \cite{chakravorty2024can}. These predications may not have the same quality in settings that use less well-behaved nuisance parameter estimators. 

The split stability $\zeta^{(r)}$ is a less frequently studied object. We will only require that it is finite for $r$ equal to 4 or 8, depending on the setting. This is a weak restriction that will hold, for example, if  the statistic $T\left(\mathsf{s}, D \right)$ is bounded.

\subsection{Computation and Concentration\label{sec: concentration and normal approx}}

\cref{alg: sequential aggregation} entails sequentially computing the statistic $a(\mathsf{R}_{g,k}, D)$, as $g$ increases, until a stopping criteria is satisfied. How much computation should we expect to do? And how does this quantity depend on the parameters $k$ and $\xi$?  

Recall from \cref{sec: asymptotic validity} that we should expect the total number of splits $\hat{m}=\hat{g} \cdot k$ used by \cref{alg: sequential aggregation} to be close to the ``oracle'' quantity
\begin{equation}\label{eq: m star def}
m^\star = g^\star \cdot k \approx 2 k \cdot v_{1,k}(D)\left(\frac{z_{1-\beta/2}}{\xi}\right)^2 ~,
\end{equation}
where $v_{1,k}(D) = \Var(a(\mathsf{r}, D)\mid D)$ denotes the conditional variance of the statistic computed using a single cross-split. Two aspects of the expression \eqref{eq: m star def} are worth highlighting. First, the total number of splits depends on the error tolerance $\xi$ through the factor $\xi^{-2}$. In other words, in order to reduce the reproducibility error by a factor of $10$, e.g., to move from $\xi = 0.1$ to $\xi = 0.01$, the total number of splits must be increased by a factor of $100$. Second, the total number of splits depends on the parameter $k$ through the factor $k \cdot v_{1,k}(D)$. How should we expect this quantity to scale with $k$?

We answer this question with the following result, which characterizes the rate of convergence of the statistic $a(\mathsf{R}_{g,k}, D)$ around its conditional mean, given by
\begin{equation} \label{eq: a bar}
\bar{a}\left(D\right) = \mathbb{E}\left[a(\mathsf{R}_{g,k}, D) \mid D\right] = \mathbb{E}\left[T(\mathsf{s}_{i,j},D) \mid D\right]
\end{equation}
under \cref{assu: linearity}. The nonstandard aspect of this result is that we account for the dependence in the summands in $a(\mathsf{R}_{g,k}, D)$ across cross-splits, i.e., the dependence induced by cross-fitting. This is accomplished by applying a coupling argument due to  \cite{chatterjee2005concentration,chatterjee2007stein}.\footnote{\label{fn: concentration}\cref{thm: cross split moment} is closely related to the unconditional variance bounds given in \cite{kale2011cross} and \cite{kumar2013near}, who give bounds with the same dependence on $k$ for cross-validated risk estimation. Our result follows from a different method of argument. In particular, \cref{thm: cross split moment} is a corollary of a more general result, presented in \cref{sec: general results}, that gives an analogous large deviations bound. In particular, we show that, for each $\varepsilon>1$, the bound 
\begin{equation*}
P\left\{ \vert a(\mathsf{R}_{g,k}, D)- \bar{a}(D) \vert \leq \sqrt{\frac{b-1}{n^2} \frac{1}{g} \frac{\log(\varepsilon^{-1})}{\delta}} \mid D \right\} \geq 1 - \varepsilon
\end{equation*}
holds with probability greater than $1-\delta$ as $D$ varies. That is, the rate of convergence suggested by \cref{thm: cross split moment} holds for all higher-order moments as well. This large deviations bound is applied repeatedly to establish the result given in the following subsection.}
\begin{theorem}
\label{thm: cross split moment}Suppose that \cref{assu: invariance,assu: linearity,def: split stability} hold, the data $D$ are independently and identically distributed, and $n = k\cdot b$. If the 4th-order split stability $\zeta^{(4)}$ is finite, then for each $\delta > 0$, the inequality 
\begin{equation} \label{eq: variance bound display}
v_{g,k}(D) = \mathbb{E}\left[\left(a\left(\mathsf{R}_{g,k}, D\right)-\bar{a}\left(D\right)\right)^{2} \mid D \right] \lesssim \frac{1}{\delta} \frac{b-1}{n^2} \frac{1}{g}~.
\end{equation}
holds with probability greater than $1-\delta$ as $D$ varies. 
\end{theorem}
\noindent The left-hand side of the inequality \eqref{eq: variance bound display} is random through the data $D$. \cref{thm: cross split moment} says that, if $\mathcal{F}$ is the event that the inequality \eqref{eq: variance bound display} holds, then $P\{\mathcal{F}\} > 1-\delta$ unconditionally. This bound results from an application of Markov's inequality, at one point in the proof, to bound a complicated, data-dependent term with a term that depends on the sample stability. This strategy---bounding the conditional quantities of interest with unconditional quantities---is helpful because the resultant unconditional object, the sample stability, is tractable and has been characterized in many settings of interest.

\cref{thm: cross split moment} demonstrates that the variance $v_{g,k}(D)$ converges to zero at the rate
\begin{equation} \label{eq: bound simple display}
\frac{1}{n} \frac{b-1}{n} \frac{1}{g} \leq  \frac{1}{n} \frac{1}{k} \frac{1}{g}~.
\end{equation}
\noindent Consequently, for a fixed total number of splits $m=g\cdot k$, the rate of convergence is proportional to $(m n)^{-1}$. That is, the conditional randomness of aggregate statistics constructed with $m$ sample-splits concentrates like averages of $m$ i.i.d.\ random variables, despite the dependence across cross-splits. Observe, also, that if $b = 1$, corresponding to ``jackknife'' or ``leave-one-out'' sample-splitting, then there is no residual randomness and the left-hand-side of \eqref{eq: bound simple display} collapses to zero.

Plugging the bound \eqref{eq: variance bound display} into the expression \eqref{eq: m star def}, we find that
\begin{equation}\label{eq: m constant}
m^\star \approx \frac{1}{n} \left(\frac{z_{1-\beta/2}}{\xi}\right)^2 ~,
\end{equation}
In other words, \cref{thm: cross split moment} implies that---in contrast to the error tolerance $\xi$---the number of cross-folds $k$ does not affect the required computation. On the other hand, all else equal, less aggregation is needed to remove residual randomness from settings with larger sample sizes $n$.

\cref{thm: cross split moment} plays out empirically. \cref{fig: non-sequential} displays measurements of the conditional variance $v_{1,k}(D)$ in dark blue, as $k$ varies, for our applications to \cite{casey2021experiment} and \cite{chakravorty2024can}. For the application to \cite{casey2021experiment}, we use the cross-validated estimate of the mean-squared error at the value of $\lambda$ that minimizes the curves displayed in Panel A of \cref{fig: lasso reproduce}. For the application to \cite{chakravorty2024can}, we use the $p$-value \eqref{eq: average p val exhibit}, associated with the 
cross-fit estimate of the average treatment effect. If the variance bound \eqref{eq: variance bound display} is accurate, then the approximation
\begin{equation}\label{eq: var prediction}
v_{1,k}(D) \approx \left(\frac{k^\prime}{k}\right) v_{1,k^\prime}(D)
\end{equation}
should hold for each pair $k,k^\prime$. To test this, we display estimates of right-hand-side of \eqref{eq: var prediction} in light green, where $k^\prime$ is the largest value of $k$ considered in each sub-figure. In each case, the approximation is remarkably accurate.\footnote{It is worth emphasizing that neither \cref{assu: invariance} nor \cref{assu: linearity} holds for the application to \cite{chakravorty2024can}. That is, in this case, nuisance parameters are estimated with random forests, which are random conditional on the data, and the $p$-value \eqref{eq: average p val exhibit} does not admit an exact representation of the form \eqref{eq: linear sep}. Rather, in this setting, both of these assumptions are good approximations (so long as the number of trees used to construct the random forests is sufficiently large), and so the predictions of \cref{thm: cross split moment} hold.}

\begin{figure}[t]
\begin{centering}
\caption{Concentration}
\vspace{-20pt}
\label{fig: non-sequential}
\begin{tabular}{c}
\textit{Panel A: \cite{casey2021experiment}}\tabularnewline
\includegraphics[scale=0.4]{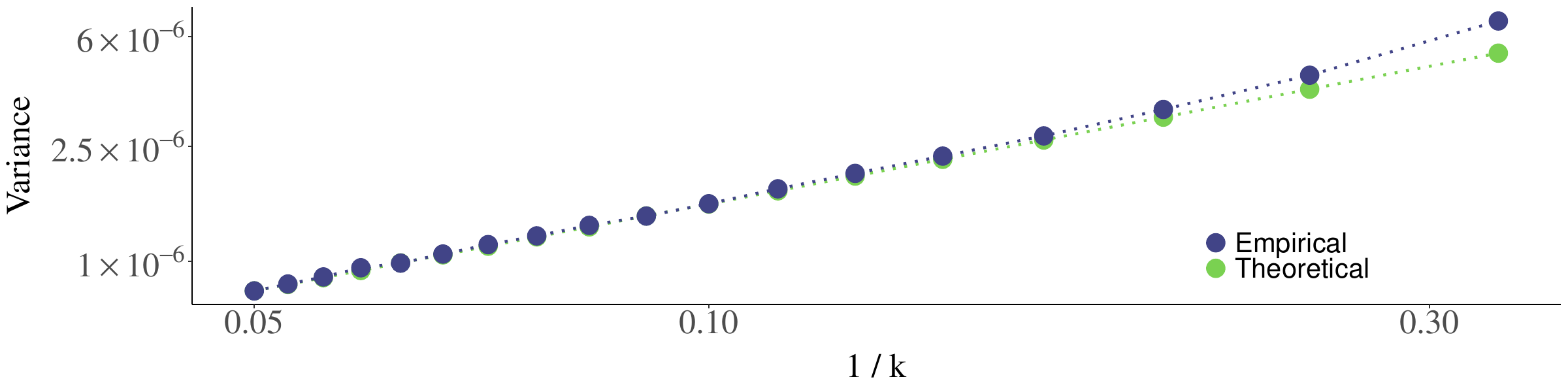}\tabularnewline
\textit{Panel B: \cite{chakravorty2024can}}\tabularnewline
\includegraphics[scale=0.4]{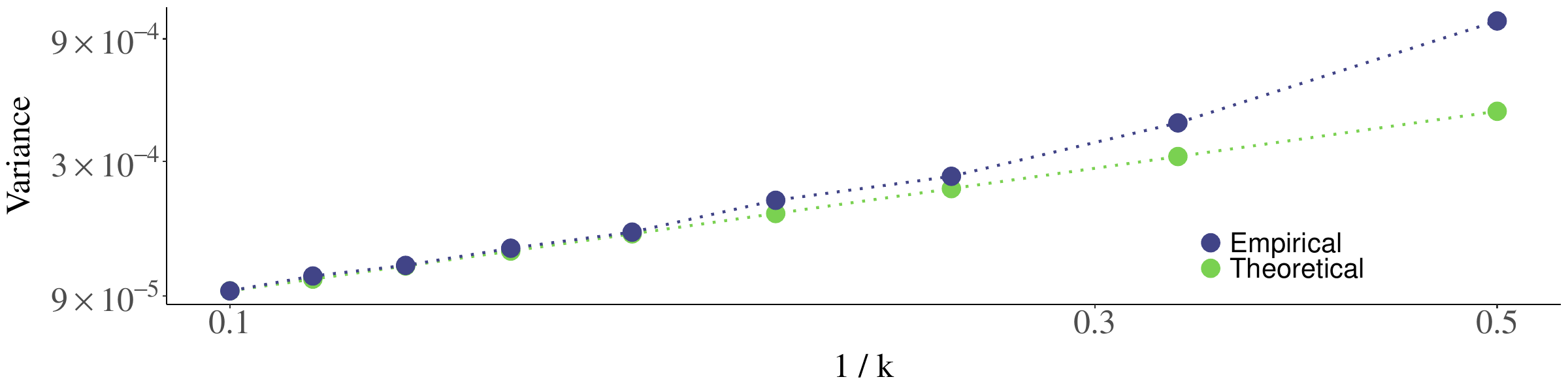}\tabularnewline
\end{tabular}
\par\end{centering}
\medskip{}
\justifying
\noindent{\footnotesize{}Notes: \cref{fig: non-sequential} illustrates the concentration of the residual randomness of various cross-split statistics with the number of cross-splits $k$, using data from \cite{casey2021experiment} and \cite{chakravorty2024can}. For the application to \cite{casey2021experiment}, we use the cross-validated estimate of the mean-squared error at the value of $\lambda$ that minimizes the curves displayed in Panel A of \cref{fig: lasso reproduce}. For the application to \cite{chakravorty2024can}, we use the $p$-value of the form \eqref{eq: average p val exhibit}, associated with the 
cross-fit estimate of the average treatment effect. The $x$-axes give $1/k$. The $y$-axes give measurements of the conditional variance of each statistic. Both the axes are displayed on a logarithmic scale, base 10. The theoretical prediction \eqref{eq: var prediction}, based on \cref{thm: cross split moment}, is given in light green.}{\footnotesize\par}
\noindent\hrulefill
\end{figure}

\subsection{Reproducibility\label{sec: nonasymp reproduce}}

\cref{alg: sequential aggregation} confers a guarantee. In particular, statistics aggregated with \cref{alg: sequential aggregation} should be interpreted as being reproducible, up to an error tolerance $\xi$, with probability greater than $\beta$. How accurate is this guarantee? In particular, how does the accuracy of the nominal reproducibility error $\beta$ depend on the choice parameters $\xi$ and $k$? 
These questions are answered by the following Berry-Esseen type bound on the accuracy of the nominal reproducibility of \cref{alg: sequential aggregation}.
\begin{theorem}
\label{thm: stable symmetric reproduce}
Suppose that the collections $\mathsf{R}_{\hat{g},k}$ and $\mathsf{R}^\prime_{\hat{g}^\prime,k}$ are independently obtained using \cref{alg: sequential aggregation}.  If \cref{assu: invariance,assu: linearity,def: split stability} hold, the conditional variance $v_{1,k}(D) = \Var(a(\mathsf{r}, D)\mid D)$ is strictly positive, almost surely, the data $D$ are independent and identically distributed, and the eighth-order split stability $\zeta^{(8)}$ is finite, then for all sufficiently small $\xi$, the inequality
\begin{align}
& \bigg\vert P\bigg\{\big\vert a(\mathsf{R}_{\hat{g},k}, D) - a(\mathsf{R}^\prime_{\hat{g}^\prime,k}, D) \big\vert \geq \xi \mid D\bigg\} 
 - \beta \bigg\vert \nonumber \\
&\quad\quad\quad\quad
\lesssim 
\frac{1}{\delta^{3/4}}
\frac{1}{k n}
\left(\frac{1}{v_{1,k}(D)}\right)^{5/4}
\left(
\frac{\xi}{z_{1-\beta/2} }
\right)^{1/2}~,\label{eq: reproduce BE bound main}
\end{align}
holds with probability greater than $1-\delta$ as $D$ varies, where in writing \eqref{eq: reproduce BE bound main}, we have omitted a multiplicative term that converges to zero logarithmically as $\xi$ decreases to zero.
\end{theorem}
\noindent  \cref{thm: stable symmetric reproduce} describes the settings under which \cref{alg: sequential aggregation} is accurate. To unpack this result, observe that
the variance bound \eqref{eq: variance bound display} gives
\begin{equation}
\frac{1}{k n}
\left(\frac{1}{v_{1,k}(D)}\right)^{5/4}
\left(
\frac{\xi}{z_{1-\beta/2} }
\right)^{1/2} 
\gtrsim
\frac{k^{1/4}}{n}
\left(
\frac{\xi}{z_{1-\beta/2} }
\right)^{1/2} ~.
\end{equation}
This suggests that the performance of \cref{alg: sequential aggregation} improves as $k$ and $\xi$ decrease and as $n$ increases.\footnote{The dependence of the bound \eqref{eq: reproduce BE bound main} on $\xi$ is sharp, at least up to the logarithmic factor. This follows from general results concerning randomly stopped sums given in \cite{landers1976exact, landers1988sharp}.} These predictions hold empirically. Panel A of \cref{fig: performance casey} displays estimates of the reproducibility error in the application to \cite{casey2021experiment} as $k$ and $\xi$ vary. An analogous figure for the application to  \cite{chakravorty2024can} is displayed in \cref{sec: simulation app}. As predicted, over most of the range of $\xi$, the reproducibility error is increasing as $k$ increases.\footnote{Note, however, that for large values of $\xi$, the reproducibility error decreases as $k$ increases. This is due to early stopping. That is, if the variance $v_{1,k}(D)$ is large relative to $\xi$, then there is an increased chance that $\hat{g}$ stops immediately after the burn-in period, i.e., $\hat{g}=g_{\mathsf{init}}$.}
 
\begin{figure}[t]
\begin{centering}
\caption{Performance in Application to \cite{casey2021experiment}}
\label{fig: performance casey}
\vspace{-20pt}
\begin{tabular}{cc}
\multicolumn{2}{c}{\textit{Panel A: Reproducibility Error}}\tabularnewline
\multicolumn{2}{c}{\includegraphics[scale=0.4]{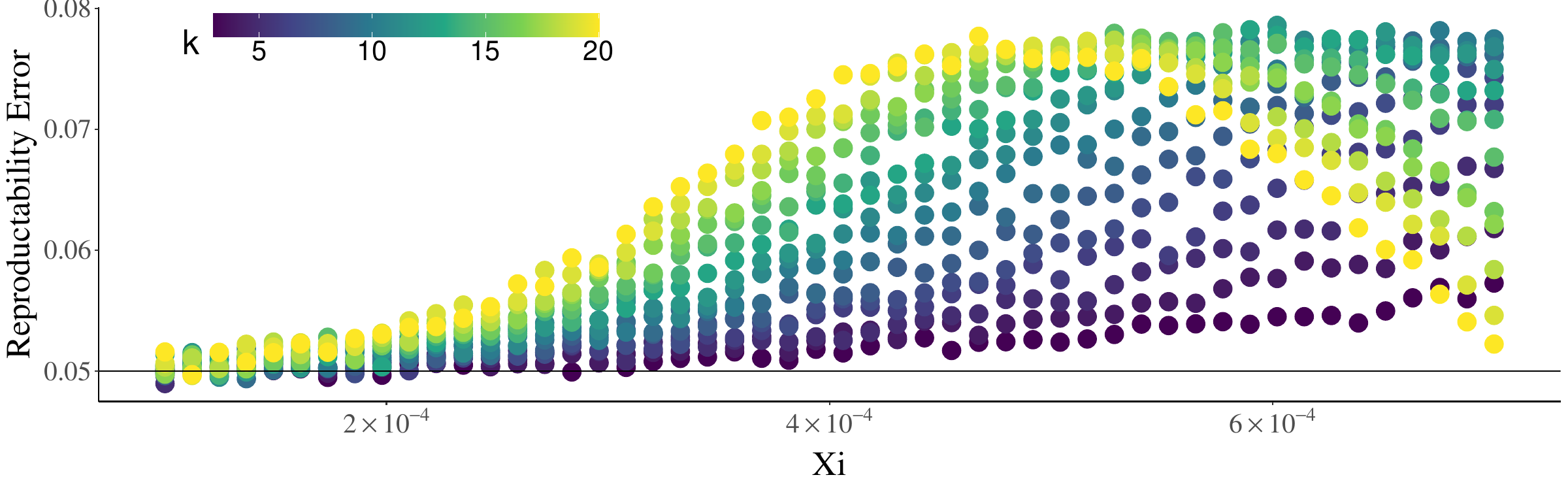}}\tabularnewline
\multicolumn{2}{c}{\textit{Panel B: Average Replications, $\hat{g}$}}\tabularnewline
\multicolumn{2}{c}{\includegraphics[scale=0.4]{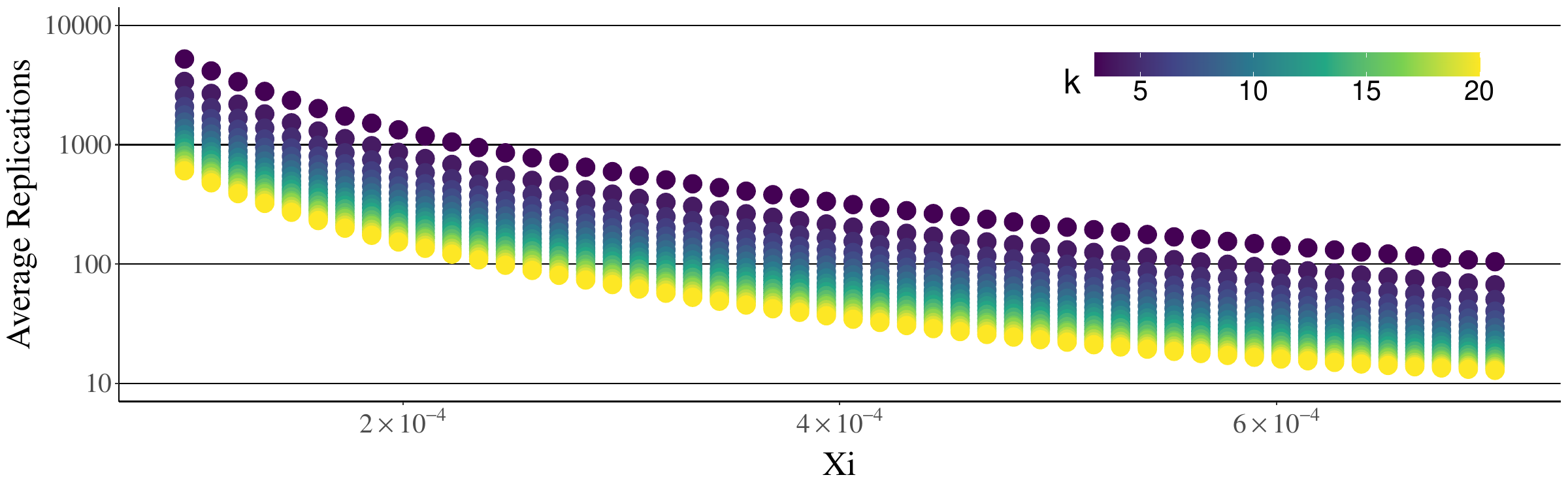}}\tabularnewline
\multicolumn{2}{c}{\textit{Panel C: Discrepancy from Oracle Stopping Time,  $\hat{g}/g^{\star} - 1$}}\tabularnewline
\multicolumn{2}{c}{\includegraphics[scale=0.4]{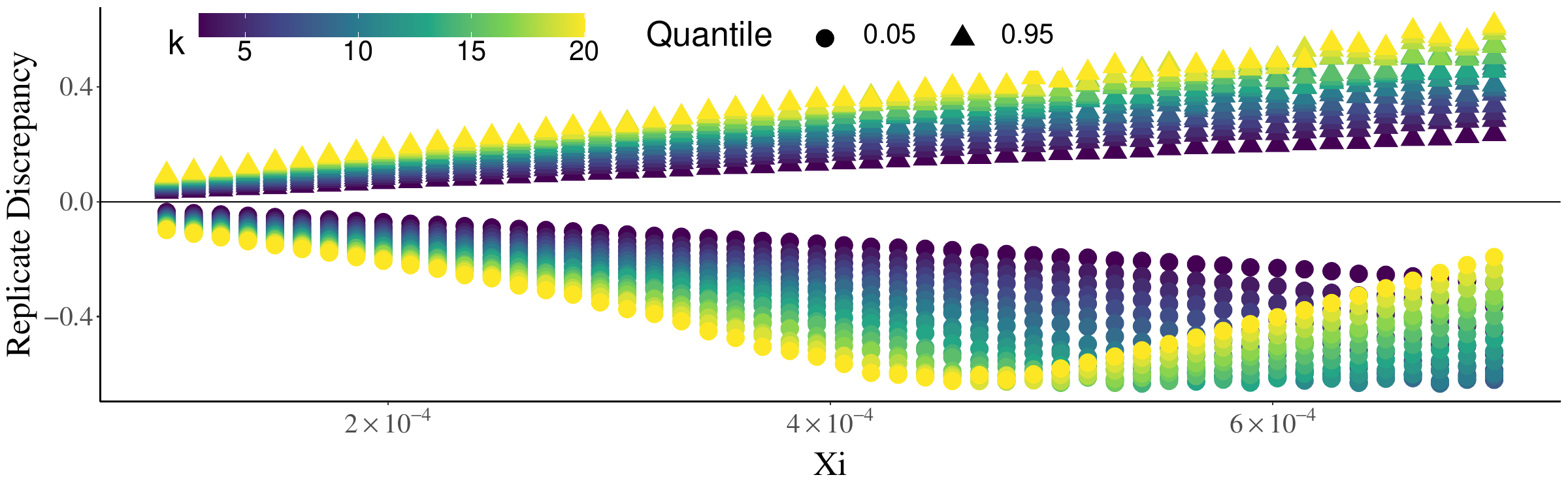}}\tabularnewline
\end{tabular}
\par\end{centering}
\medskip{}
\justifying
\noindent{\footnotesize{}Notes: \cref{fig: performance casey} displays measurements of the performance of \cref{alg: sequential aggregation} on the data from \cite{casey2021experiment}. Panel A displays measurements of the reproducibility error, $P\{\vert a(\mathsf{R}_{\hat{g},k}, D) - a(\mathsf{R}^\prime_{\hat{g}^\prime,k}, D) \vert \geq \xi \mid D\}$, as $\xi$ and $k$ vary. A solid horizontal line is displayed at the nominal error rate $\beta = 0.05$. Panel B displays measurements of the average number of replications $\hat{g}$ s $\xi$ and $k$ vary. The $y$-axis is displayed with a log scale, base 10. Solid horizontal lines are placed at each exponential factor of 10. Panel C displays measurements of the 5th and 95th quantiles of the discrepancy $\hat{g}/g^{\star} - 1$ as $k$ and $\xi$ vary. Further details on the construction of this figure are given in \cref{sec: simulation app}.}{\footnotesize\par}
\end{figure}

Likewise, Panel B displays estimates of the average number of cross-splits $\hat{g}$ used in \cref{alg: sequential aggregation}, at each value of $\xi$ and $k$.  Again, as predicted, the total number of splits used by the procedure stays constant as $k$ varies. To see this, observe that the average value of $\hat{g}$ is roughly 10 times smaller for $k = 20$ than for $k=2$, as predicted by the approximation \eqref{eq: m constant}. Comparing Panels A and B, observe that, once the average value of $\hat{g}$ is larger than approximately 500, the reproducibility error is close to the nominal error rate $\beta$.

If the total number of splits required to ensure reproducibility at a given error tolerance $\xi$ is constant as $k$ varies, then why does the performance of \cref{alg: sequential aggregation} decrease with $k$? As part of the proof of \cref{thm: stable symmetric reproduce}, we show that, for each $\varepsilon > 0$, the bound
\begin{equation}
P\left\{ \bigg\vert \frac{\hat{g}}{g^\star} - 1\bigg\vert \lesssim \frac{1}{n} \frac{1}{k} \left(\frac{1}{v_{1,k}(D)}\right)^{3/2} \frac{\xi}{z_{1-\beta/2}} \sqrt{\frac{\log(\varepsilon^{-1})}{\delta}} \mid D \right\} \geq 1 - \varepsilon
\end{equation}
holds with probability greater than $1-\delta$, as $D$ varies. Again the variance bound \eqref{eq: variance bound display} gives
\begin{equation}
 \frac{1}{n} \frac{1}{k} \left(\frac{1}{v_{1,k}(D)}\right)^{3/2} \frac{\xi}{z_{1-\beta/2}}  \gtrsim \frac{k^{1/2}}{n} \frac{\xi}{z_{1-\beta/2}}~.
\end{equation}
Written differently, as $k$ increases the discrepancy between the realized and oracle number of cross-splits, $\vert\hat{g}/g^{\star} - 1\vert$, increases, reducing the accuracy of the nominal error rate. Roughly speaking, this happens because the quality of the estimator $\hat{v}_{g,k}(D)$ for the conditional variance $v_{g,k}(D)$ depends only on the number of cross-splits $g$.\footnote{\label{fn: normal approx}Similarly, in \cref{sec: general results}, we show that the quality of a normal approximation to $a(\mathsf{R}_{g,k},D)$ depends only on $g$, although this discrepancy is not the leading term in the reproducibility error. The close approximation exhibited in \cref{fig: non-sequential} suggests that is may be reasonable to estimate $v_{1,k}(D)$ by taking the sample variance both across and within cross-splits, i.e., computing the sample variance across all $m$ sample-splits. Although this worth further consideration, asymptotic validity, i.e., \cref{eq: general reproducibility}, would not hold at the same level of generality.} That is, for a given number of sample-splits $\hat{m}=k\cdot \hat{g} $, \cref{alg: sequential aggregation} is most accurate when $\hat{g}$ is large, as the variance estimate $\hat{v}_{\hat{g},k}(D)$ is more precise. 

As before, these predictions play out in practice. Panel C of \cref{fig: performance casey}  displays estimates of the 5th and 95th quantiles of the distribution of the discrepancy $\hat{g}/g^{\star} - 1$ as $k$ and $\xi$ vary. Over most of the range of $\xi$ the discrepancy between $\hat{g}$ and $g^\star$ is increasing in $k$. At small values of $k$ and large values of $\xi$, there is an increased chance that $\hat{g}$ stops immediately after the burn-in period, i.e., $\hat{g}=g_{\mathsf{init}}$.

It is worth pausing to note that these results do not support eschewing cross-splitting altogether, i.e., setting $k$ equal to one and aggregating over independent splits, as we have restricted attention to the case that $n=k\cdot b$. In \cref{sec: general results}, we give analogous results that relax this assumption and show that cross-splitting, i.e., setting $n=k\cdot b$, reduces residual randomness at a faster rate than independent splitting, i.e., setting $k$ equal to one. In other words, all else equal, \cref{alg: sequential aggregation} performs best, in the sense that the nominal reproducibility error is most accurate, when $k$ is equal to 2 and $b$ is equal to $b/2$.

To summarize, the computation needed to achieve a desired bound on residual randomness is highly sensitive to the error tolerance $\xi$, but is not affected by the number of cross-folds $k$. On the other hand, the accuracy of the nominal error rate of \cref{alg: sequential aggregation} decreases as the number of cross-folds $k$ increases. If the error bound $\xi$ is chosen to be suitably small, such that the realized number of cross-splits $\hat{g}$ is greater than roughly 500, then the nominal reproducibility probability tends to be quite accurate, and insensitive to changes in the number of cross-splits. 

\section{Recommendations for Practice\label{sec: conclusion}}

Sample-splitting is a helpful tool for simplifying many widely encountered problems in applied econometrics. This simplification comes at the cost of the introduction of residual randomness. We have shown, in several applications, that this residual randomness is large enough to substantively affect results. To address this, we have proposed a simple procedure, summarized in \cref{alg: sequential aggregation}, for removing the auxiliary randomness from sample-split statistics. The procedure takes as input a bound and an error rate. We have shown that, if the procedure were run twice, the chance that the results differ by more than the bound is well-approximated by the error rate. 

We conclude, in this section, by detailing several recommendations for how to best implement \cref{alg: sequential aggregation} in practice. The most important choice to make is the specification of the sample-split statistic of interest. For example, suppose that we are interested in estimating an average treatment effect, or other causal contrast, using an estimator based on sample-splitting. Denote this quantity by $\mathsf{est}(\mathsf{r}, D)$ and let $\mathsf{se}(\mathsf{r}, D)$ denote an associated, potentially sample-split, standard error estimate. There are several reasonable choices that one might make. In this case, we recommend applying \cref{alg: sequential aggregation} to sequentially aggregate the $p$-value 
\begin{equation}\label{eq: con p val}
a(\mathsf{r}, D) = 1 - \Phi\left(\frac{\mathsf{est}(\mathsf{r}, D)}{\mathsf{se}(\mathsf{r}, D)}\right)
\end{equation}
associated with a test that the contrast of interest is greater than zero, where $\Phi(\cdot)$ denotes the standard normal c.d.f. This approach has the benefit of benchmarking residual randomness with an estimate of the sampling error.

In some settings, it may be of interest to report the scope of residual randomness for alternative quantities associated with statistics that have aggregated with \cref{alg: sequential aggregation}. For example, suppose that we have applied \cref{alg: sequential aggregation} to stabilize the $p$-value \eqref{eq: con p val}, and obtain $a(\mathsf{R}_{\hat{g}, k}, D)$.  We may also wish to report the associated estimate $\mathsf{est}(\mathsf{R}_{\hat{g}, k}, D)$. In this case, we recommend reporting and interpreting the standard error
\begin{equation}\label{eq: con est se}
\sqrt{ \frac{1}{\hat{g}} \frac{1}{\hat{g}-1} \sum^{\hat{g}}_{i=1} (\mathsf{est}(\mathsf{r}_i, D) - \mathsf{est}(\mathsf{R}_{\hat{g}, k}, D))^2 }
\end{equation}
in the usual way. That is, the standard error \eqref{eq: con est se} estimates the residual randomness of the estimate $\mathsf{est}(\mathsf{R}_{\hat{g}, k}, D)$, conditional on the data. If this is too large, the error tolerance $\xi$ should be decreased. 

The optimal choice of statistic in settings where cross-validation is used for model selection and risk estimation is less immediate. In this case, we recommend applying \cref{alg: sequential aggregation} to each element of the statistic \eqref{eq: mse difference}, i.e., the relative values of the risk estimates.

The second most important choice to make is the choice of the error tolerance $\xi$. In some cases, an appropriate choice is clear. For example, it may be desirable to ensure that a $p$-value is reproducible at the level relevant for the determination of statistical significance at level 0.10 or 0.025, e.g., setting $\xi = 0.01$ or $\xi = 0.001$. In other cases, like in applications to cross-validated model selection, this choice is less clear. We propose a procedure for making this choice in that setting in \cref{app: lasso xi determine}.

The error tolerance $\xi$ should be set to a value that is sufficiently small so that the nominal reproducibility error is accurate. On the other hand, $\xi$ should not be set so small that the computation associated with implementing \cref{alg: sequential aggregation} becomes infeasible. To ensure that the chosen value of $\xi$ satisfies these constraints, we recommend computing the statistic of interest for each of a small, initial sample of cross-splits (e.g., $g$ equal to 20 or 30). Let $\hat{v}_{1,k}(D)$ denote an estimate of the conditional variance $v_{1,k}(D)$ computed using this sample. An estimate of the total number of cross-splits needed to implement \cref{alg: sequential aggregation} at a specified level of $\xi$ can then be obtained by 
\begin{equation}\label{eq: g star re con}
 g(\xi) = 2 \hat{v}_{1,k}(D)\left(\frac{z_{1-\beta/2}}{\xi}\right)^2 ~.
\end{equation}
Motivated by the numerical results reported in \cref{sec: Guarantees}, we recommend choosing a value of $\xi$ such that this estimate is greater than 500. 

In some cases, the estimate \eqref{eq: g star re con} may be too large at practically relevant values of $\xi$ to be computationally feasible. Unfortunately, as we have shown in \cref{sec: Guarantees}, changing the number of cross-folds $k$ will not address this issue, and it may be best to consider alternative choices of nuisance parameter estimators. In this situation, one might consider using ``leave-one-out'' or ``jackknife'' sample-splitting, i.e., setting $k = n$, which exhibits no residual randomness. This is not an omnibus fix, however. For example, proofs of the asymptotic normality of DML estimates of average treatment effects require that $k$ is small relative to $n$ \citep{chernozhukov2018double}. Similarly, leave-one-out cross-validation is not necessarily optimal for model selection (see e.g., \cite{shao1993linear} for early discussion of this point). \cite{chetverikov2024tuning} gives a review of various, alternative, approaches to tuning parameter selection.

If the number of cross-folds $k$ is set too high, the realized number of cross-splits $\hat{g}$ may be too small and \cref{alg: sequential aggregation} may perform poorly. But, so long as the user ensures that the error tolerance $\xi$ is sufficiently small such that the number of cross-splits $\hat{g}$ tends to be large (e.g., greater than 500, say), small changes in the number of folds $k$ should not have adverse effects. Thus, this choice should be made on substantive grounds, that will depend on the application. In practice, the conventional choices of $k$ equal to 2, 5, or 10 should work well in most settings.

We have had relatively little to say about the choice of the nominal reproducibility error $\beta$. We have found that setting $\beta = 0.05$ is suitable for most applications. Finally, although it has not played a central role in our theoretical analysis, it is important to choose the burn-in period $g_{\mathsf{init}}$ to be suitably large to ensure that the variance estimate $\hat{v}_{g_{\mathsf{init},k}}(D)$ is reasonably accurate, otherwise \cref{alg: sequential aggregation} will tend to stop too early. We recommend setting $g_{\mathsf{init}}$ to be equal to at least 10.

\clearpage
\end{spacing}
\begin{spacing}{1.2}
\bibliographystyle{apalike}
\bibliography{references.bib}
\end{spacing}
\newpage

\begin{appendix}
\renewcommand\thefigure{\thesection.\arabic{figure}}    

\begin{center}
{\large\it Supplemental Appendix to:}
\vspace{0.2em}
\begin{spacing}{1}
\Large{\textbf{Reproducible Aggregation of Sample-Split Statistics\daggerfootnote{\textit{Date}: \today}}}
\end{spacing}
\end{center}
\begin{center}
\begin{tabular}[t]{c@{\extracolsep{4em}}c} 
\large{David M. Ritzwoller} & \large{Joseph P. Romano}\vspace{-0.7em}\\ \vspace{-1em}
\small{Stanford University} & \small{Stanford University} \\ \vspace{-0.7em}
\end{tabular}%
\end{center}
\begin{spacing}{1.13}
\DoToC
\end{spacing}
\thispagestyle{empty}
\setcounter{page}{0}
\setcounter{figure}{0}   
\newpage

\begin{spacing}{1.3}

\section{Data and Simulations\label{app: simulation appendix}}

\subsection{\cite{casey2021experiment}\label{app: Casey app}} 

We obtain the data associated with \cite{casey2021experiment} from the replication package posted on OpenICPSR.\footnote{The replication package associated with \cite{casey2021experiment} is posted at the URL \url{https://www.openicpsr.org/openicpsr/project/124501/version/V1/view}.} The units of observation are the aspirant political candidates in Sierra Leone's 2018 parliamentary elections. The data that we consider consist of an outcome and 48 covariates. The outcome is the vote share for each aspirant candidate in a poll of party officials. The covariates collect various measurements relating to the candidate's personal, political, and financial characteristics. All data cleaning steps match \cite{casey2021experiment}.

We consider the exercise summarized in Appendix Table A.4 of \cite{casey2021experiment}. These authors repeatedly compute 10-fold cross-validated estimates of the risk associated with a regularized regression of the outcome on the collection of covariates, at each value of a grid of penalization parameters. They include both Ridge and Lasso type penalties. For each replication, they compute the covariates that are selected in the regression associated with the penalty parameter having the minimum risk estimate. They display the covariates that are selected in more than 200 of the 400 replications. These covariates are included as controls in various downstream analyses. In our replication of this exercise, to ease exposition, we only include a Lasso type penalty. 

\subsection{\cite{chakravorty2024can}\label{app: Chakravorty app}} We obtain the data and replication package associated with \cite{chakravorty2024can} directly from the authors.\footnote{We thank Bhaskar Chakravorty for sharing this material.} The units of observation are trainees in a job training program implemented in the Indian states of Bihar and Jharkhand. There are 2,488 total trainees, who were divided into batches of approximately 30 trainees. Batches of trainees were randomly assigned to treatment with a program involving the provision of information concerning potential placement jobs. All data cleaning steps match those taken in \cite{chakravorty2024can} .

We replicate the analysis reported in Column 4 of Table 1 of \cite{chakravorty2024can}. This estimate restricts attention to the 890 trainees that completed the training and were subsequently placed into a job. The outcome variable is an indicator for whether the trainee is still in the job five months after training completion. The covariates collect measurements of 77 attributes related to the demographics, human capital, and expectations of the trainees. These covariates are listed in Table 3.1 of \cite{chakravorty2024can}. 

Following \cite{chakravorty2024can}, we estimate the effect of the treatment using the implementation of Double Machine Learning available through the ``DoubleML'' R package \citep{bach2024double}. We use 5-fold cross-fitting. Nuisance parameters (i.e., outcome regressions and propensity scores) are estimated with Random Forests using the default tuning parameters associated with the  ``Ranger'' R package \citep{wright2017ranger}. Standard errors are clustered at the batch level. 

\subsection{\cite{haushofer2022targeting}\label{app: Egger app}} 

\cite{haushofer2022targeting} consider data originally studied in \cite{egger2022general}. At the time of the preparation of this paper, no replication materials for \cite{haushofer2022targeting} were publicly available. We replicate a simplified version of the exercise considered in \cite{haushofer2022targeting}, using data obtained from the replication package associated with \cite{egger2022general}.\footnote{The replication package for \cite{egger2022general} is available at \url{https://www.econometricsociety. org/publications/econometrica/2022/11/01/General-Equilibrium-Effects-of-Cash-Transfers-Experimental-Evidence-From-Kenya}.} The units of observation are households in three sub-counties of Siaya County, western Kenya. Households were randomly allocated large cash transfers. See \cite{egger2022general} and \cite{haushofer2022targeting} for further details on the design of this experiment. 

Using the replication data from \cite{egger2022general}, we begin with a sample of 5,423 treatment eligible households. We then restrict attention to 4,758 households who were surveyed in both the first and second round of the experiment. We drop any households that are missing measurements of post-treatment assets, consumption, income, or food security, leaving a sample of 4,755 households. These data are merged with pre-treatment measurements  of the demographics, assets, food security, and labor market participation for each household, leaving a sample of 4,754 households. 

We obtain a cleaned dataset that records 24 measurements associated with each household. The dataset identifies the village that contains each household, whether each household was assigned to treatment, and post-treatment measurements of indices of assets, consumption, income, and food security. The covariates are measurements of the household size, the number of meals eaten the day before the survey, the number of high-protein meals eaten the day before the survey, the time between the administration of the program and the measurement of post-treatment outcomes,\footnote{We use the measurement of the time between the administration of the program and the measurement of post-treatment outcomes available as ``exptoend'' in the dataset ``GE\_Survey\_and\_Transfer\_Dates.dta'' in the replication package for \cite{egger2022general}. This variable appears to differ from the analogous quantity used in \cite{haushofer2022targeting}, whose density is displayed in their Figure A.2.} and indicators for whether the household contains a widow, contains a female, has children, has children under 3, has children under 6, contains an elderly resident, has livestock, has land, has more than a quarter acre of land, has a radio and tv, has a self-employed resident, and has an employed resident. 

We implement a simplified version of the exercise studied in \cite{haushofer2022targeting}. Consider a subset $\mathsf{s}$. Let the complement of $\mathsf{s}$ in $[n]$ be denoted by $\tilde{\mathsf{s}}$. Let $Y_i$ denote a measurement of an index quantifying post-treatment consumption, $X_i$ denote a vector of measurements of pre-treatment covariates, and $W_i$ denote assignment to treatment. Let $H_i$ denote the number of occupants of household $i$. Let $D_i$ collect these observations. Let $Y_i(1)$ and $Y_i(0)$ denote the potential outcomes induced by the treatment $W_i$. For each household $i$ in $\mathsf{s}$, let 
\begin{equation}\label{eq: te and uto}
\mathsf{te}(D_i, \hat{\eta}(D_{\tilde{\mathsf{s}}}))
\quad\text{and}\quad
\mathsf{uto}(D_i, \hat{\eta}(D_{\tilde{\mathsf{s}}}))
\end{equation} 
denote sample-split estimates of the treatment effect $Y_i(1) - Y_i(0)$ and \emph{per-capita} untreated outcome $Y_i(0)/H_i$, respectively. In this case, the nuisance parameters $\hat{\eta}(D_{\tilde{\mathsf{s}}})$ consist of nonparametric estimates of outcome regressions, i.e., the conditional expectations \eqref{eq: conditional expectations}, computed using the data for units $i$ in $\tilde{\mathsf{s}}$. 

Following \cite{haushofer2022targeting}, we compute these estimates with Random Forests using the ``GRF'' R package \citep{athey2019generalized}.\footnote{Following \cite{haushofer2022targeting}, we use the default tuning parameters, but set ``sample.fraction'' equal to  0.1 and ``min.node.size'' equal to 10.} As best as we can tell, there are two differences between our implementation and the results presented in \cite{haushofer2022targeting}. First, we make no attempt to include sample weights that reflect the sampling probabilities. Second, we include the time between treatment and the measurement of post-treatment outcomes as a covariate, rather than implement some form of de-meaning and re-weighing of estimates across time. 

Let $\mathsf{Impacted}(\mathsf{s})$ collect the 50\% of units in $\mathsf{s}$ associated with the largest values of $\mathsf{te}(D_i, \hat{\eta}(D_{\tilde{\mathsf{s}}}))$. Similarly, let $\mathsf{Deprived}(\mathsf{s})$ collect the 50\% of units $i$ in $\mathsf{s}$ associated with the smallest values of $\mathsf{uto}(D_i, \hat{\eta}(D_{\tilde{\mathsf{s}}}))$. Let $\mathsf{Impacted}_w(\mathsf{s})$ collect the subset of $\mathsf{Impacted}(\mathsf{s})$ with $W_i = w$. Define $\mathsf{Deprived}_w(\mathsf{s})$ analogously. Consider the sample-split statistic
\begin{align}\label{eq: def imp dep diff}
T(\mathsf{s},D) 
& = \left( \frac{1}{\vert \mathsf{Impacted}_1(\mathsf{s}) \vert } \sum_{i \in \mathsf{Impacted}_1(\mathsf{s}) } Y_i -  \frac{1}{\vert \mathsf{Impacted}_0(\mathsf{s}) \vert } \sum_{i \in \mathsf{Impacted}_0(\mathsf{s}) } Y_i \right) \nonumber \\
& - \left( \frac{1}{\vert \mathsf{Deprived}_1(\mathsf{s}) \vert } \sum_{i \in \mathsf{Deprived}_1(\mathsf{s}) } Y_i -  \frac{1}{\vert \mathsf{Deprived}_0(\mathsf{s}) \vert } \sum_{i \in \mathsf{Deprived}_0(\mathsf{s}) } Y_i \right)~.
\end{align}
That is, the sample-split statistic \eqref{eq: def imp dep diff} is an estimator of the difference in the average treatment effect of the most impacted and most deprived units in $\mathsf{s}$. The statistic \eqref{eq: def imp dep diff} can be aggregated with cross-splitting, through
\begin{equation}\label{eq: appendix haus cross}
a(\mathsf{r},D) = \frac{1}{k} \sum_{i=1}^k T(\mathsf{s}_i,D)~,
\end{equation}
where $\mathsf{r} = (\mathsf{s}_i)_{i=1}^k$ is a $k$-fold partition of $[n]$. Following \cite{haushofer2022targeting}, in \cref{sec: Reproducible}, we use 5-fold cross-fitting. 

\cite{haushofer2022targeting} say that they construct a standard error for \eqref{eq: appendix haus cross} with the bootstrap. There are a variety of ways that one might do this. In our implementation, we keep the values of the treatment effect and untreated outcome estimates \eqref{eq: te and uto} fixed for each household, and compute a bootstrap standard error for the statistic \eqref{eq: appendix haus cross}, re-computing the groups $\mathsf{Impacted}(\mathsf{s})$ and $\mathsf{Deprived}(\mathsf{s})$  in each bootstrap replicate. 

\subsection{\cite{beaman2023selection}\label{app: Beaman app}} We obtain the data associated with \cite{beaman2023selection} from the associated replication package posted on the Econometric Society's webpage.\footnote{The replication package for \cite{beaman2023selection} is available at  \url{https://www.econometricsociety.org/publications/econometrica/2023/09/01/Selection-into-Credit-Markets-Evidence-from-Agriculture-in-Mali}} The units of observations are low-income farmers in Mali. The cleaned data consist of observations of 5210 farmers across 198 villages. All data cleaning steps match \cite{beaman2023selection}.

\cite{beaman2023selection} consider a two-stage experiment. In the first stage, a microcredit organization randomly offered group-liability loans to women in 88 of the 198 villages. In the second stage, households were randomly offered cash grants. We replicate the analysis reported in Column (4) of Appendix Table VI of \cite{beaman2023selection}. In this setting, attention is restricted to the 2,142 households who received a loan. They are interested in testing whether the effects of cash-grants on farm profits are heterogeneous. For this problem, the data $D_i$ consist of a measurement $Y_i$ of the post-transfer profit for farm $i$, the variable $W_i$ denotes assignment to the cash transfer, and the vector $X_i$ collects a vector of pre-treatment measurements of the physical and financial attributes of each farm. Let $Y_i(1)$ and $Y_i(0)$ denote the potential outcomes induced by the treatment $W_i$.

\cite{beaman2023selection} implement a sample-split test for treatment effect heterogeneity proposed by \cite{chernozhukov2018generic}. Let $\mathsf{s}$ denote a random half-sample of $[n]$, i.e., a random subset of $[n]$ of size $n/2$. As before $\tilde{\mathsf{s}}$ denotes the complement of $\mathsf{s}$ in $[n]$. For each farm $i$ in $\mathsf{s}$, let 
\begin{equation}\label{eq: te Beaman}
\mathsf{te}(D_i, \hat{\eta}(D_{\tilde{\mathsf{s}}}))
\end{equation} 
denote sample-split estimates of the treatment effect $Y_i(1) - Y_i(0)$. The nuisance parameters $\hat{\eta}(D_{\tilde{\mathsf{s}}})$ consist of nonparametric estimates of outcome regressions, i.e., the conditional expectations \eqref{eq: conditional expectations}, computed using the data for units $i$ in $\tilde{\mathsf{s}}$. We compute these estimates with Random Forests using the ``GRF'' R package \citep{athey2019generalized}. Tuning parameters are chosen with cross-validation in the same manner implemented in \cite{beaman2023selection}.

To test for treatment effect heterogeneity, we estimate the coefficients of the linear regression
\begin{equation}
Y_i = \alpha + \tau_1 \cdot W_i + \tau_2 \cdot \mathsf{te}(D_i, \hat{\eta}(D_{\tilde{\mathsf{s}}}))
+ \tau_3 \cdot W_i \cdot \mathsf{te}(D_i, \hat{\eta}(D_{\tilde{\mathsf{s}}})) + \varepsilon_i~.
\end{equation}
The idea is that, if the treatment effect estimates are heterogeneous, then the ``true'' value of the interaction coefficient $\tau_3$ should be positive. Let $\hat{\tau}_3(\mathsf{s},D)$ denote the estimate of $\tau_3$ computed with linear regression in the sample $\mathsf{s}$. Let $\mathsf{se}(\mathsf{s},D)$ denote the associated standard error, clustered at the village level. \cite{chernozhukov2018generic} advocate for computing the median estimate and standard error over 250 half-splits $\mathsf{s}$. \cite{beaman2023selection} use 1000 half-splits. 

\subsection{Simulations\label{sec: simulation app}}

In this section, we give further details concerning the simulations whose results are reported in \cref{sec: Guarantees}. For both the applications to \cite{casey2021experiment} and \cite{chakravorty2024can}, we sample 100,000 replications of the cross-split statistic of over a range of values of $k$. In the application to \cite{chakravorty2024can}, we have increased the number of trees used to compute the random forest nuisance parameter estimates, from 100 to 3000, in order to ensure that the residual randomness induced by cross-splitting is greater than the residual randomness induced by the nuisance parameter estimate. We use these replicates to estimate the variance $v_{1,k}(D)$ for each value of $k$, in addition to the oracle stopping time
\begin{equation}\label{eq: m star re def}
g^\star = 2  v_{1,k}(D)\left(\frac{z_{1-\beta/2}}{\xi}\right)^2 ~,
\end{equation}
at each value of $k$ and $\xi$

We implement \cref{alg: sequential aggregation} 100,000 times for each value of $k$ and $\xi$, for each application, by sampling with replacement from the replicates of 100,000 draws. We use these estimates to estimate the reproducibility error 
\begin{equation}
P\bigg\{\big\vert a(\mathsf{R}_{\hat{g},k}, D) - a(\mathsf{R}^\prime_{\hat{g}^\prime,k}, D) \big\vert \geq \xi \mid D\bigg\} 
\end{equation}
as well as the distribution of $\hat{g}$ for each $k$ and $\xi$. These estimates are used to construct \cref{fig: performance casey}.

\cref{fig: performance chakravorty} displays measurements analogous to those displayed in \cref{fig: performance casey} for the application to \cite{chakravorty2024can}. The accuracy of the reproducibility error, relative to the application to \cite{casey2021experiment}, is somewhat worse. This is likely a consequence of the fact that the quality of a normal approximation to the conditional distribution of the $p$-value \eqref{eq: average p val exhibit} is worse than the quality of the normal approximation to the conditional distribution of the mean-squared error estimate. As before, the reproducibility error tends to be most accurate at values of $k$ and $\xi$ where the average number of cross-splits $\hat{g}$ is greater than 500.  

\begin{figure}[p]
\begin{centering}
\caption{Performance in Application to \cite{chakravorty2024can}}
\label{fig: performance chakravorty}
\vspace{-20pt}
\begin{tabular}{cc}
\multicolumn{2}{c}{\textit{Panel A: Reproducibility Error}}\tabularnewline
\multicolumn{2}{c}{\includegraphics[scale=0.4]{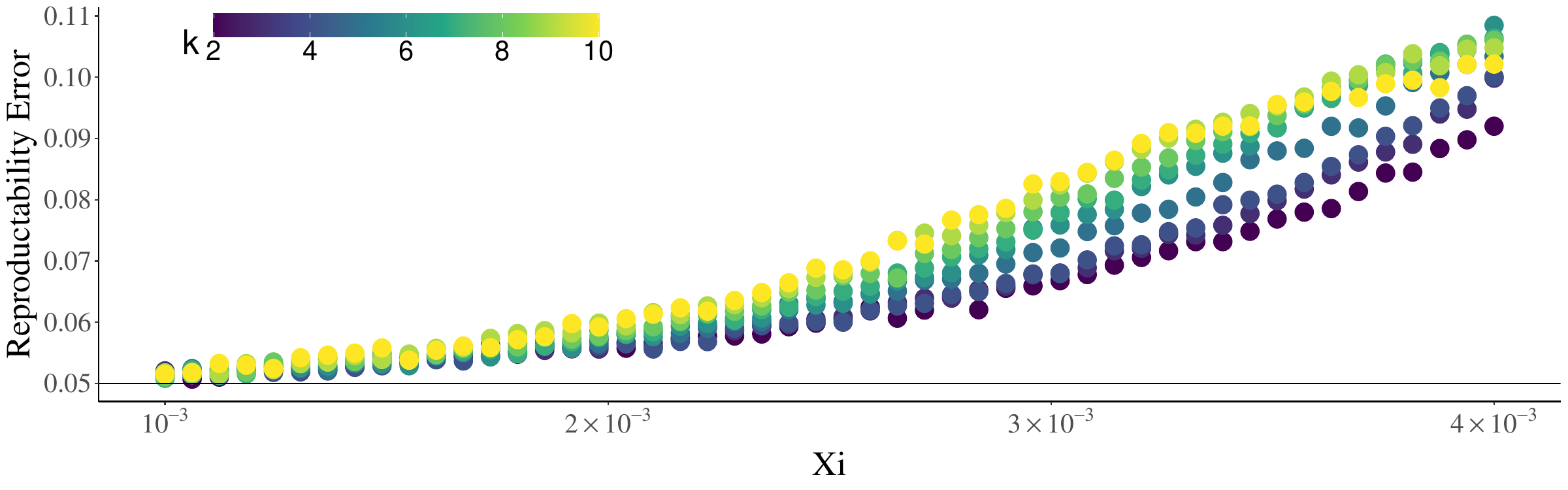}}\tabularnewline
\multicolumn{2}{c}{\textit{Panel B: Average Replications}, $\hat{g}$}\tabularnewline
\multicolumn{2}{c}{\includegraphics[scale=0.4]{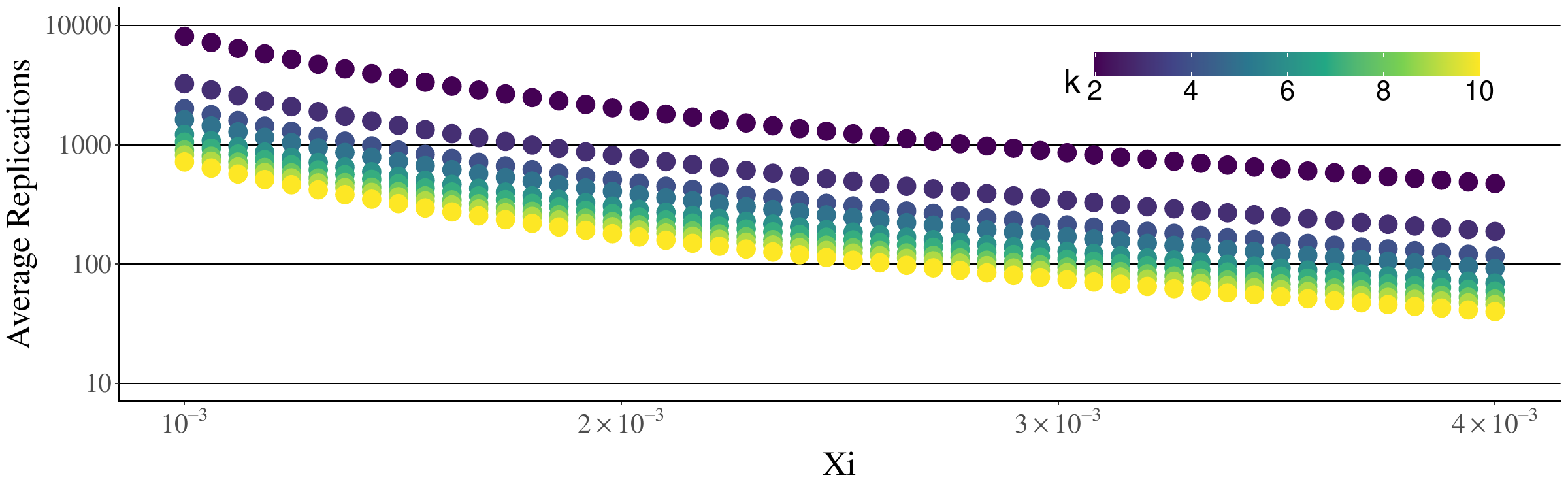}}\tabularnewline
\multicolumn{2}{c}{\textit{Panel C: Discrepancy from Oracle Stopping Time}, $\hat{g}/g^{\star} - 1$}\tabularnewline
\multicolumn{2}{c}{\includegraphics[scale=0.4]{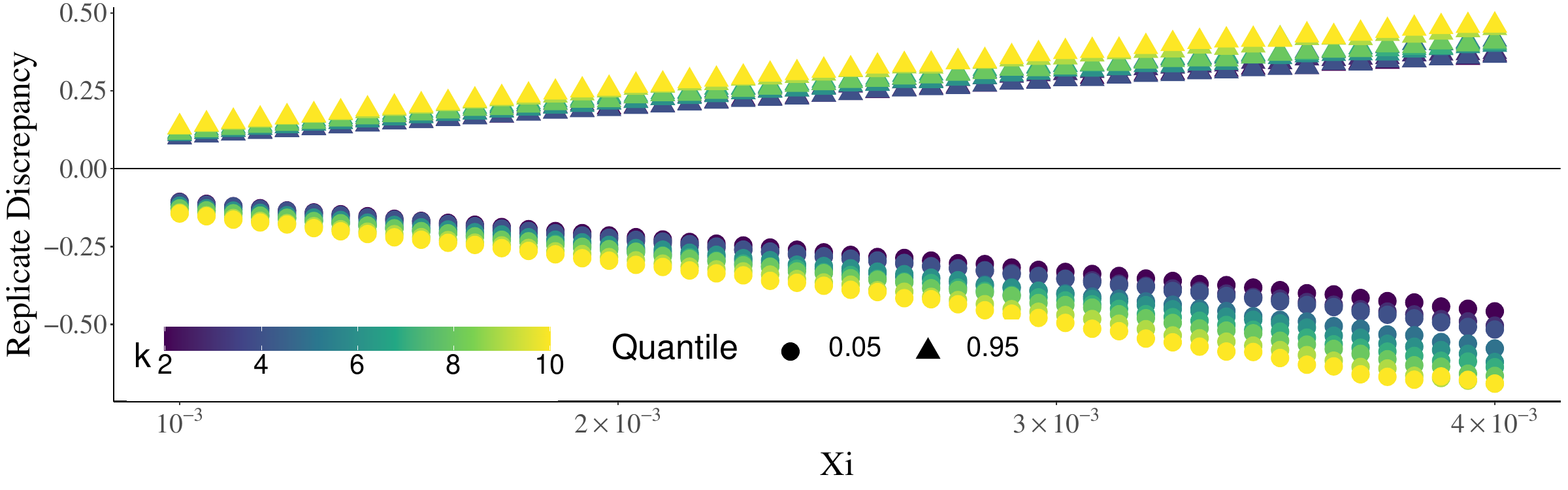}}\tabularnewline
\end{tabular}
\par\end{centering}
\medskip{}
\justifying
\noindent{\footnotesize{}Notes: \cref{fig: performance chakravorty} displays measurements of the performance of \cref{alg: sequential aggregation} on the data from \cite{chakravorty2024can}. Panel A displays measurements of the reproducibility error, $P\{\vert a(\mathsf{R}_{\hat{g},k}, D) - a(\mathsf{R}^\prime_{\hat{g}^\prime,k}, D) \vert \geq \xi \mid D\}$, as $\xi$ and $k$ vary. A solid horizontal line is displayed at the nominal error rate $\beta = 0.05$. Panel B displays measurements of the average number of replications $\hat{g}$ s $\xi$ and $k$ vary. The $y$-axis is displayed with a log scale, base 10. Solid horizontal lines are placed at each exponential factor of 10. Panel C displays measurements of the 5th and 95th quantiles of the discrepancy $\hat{g}/g^{\star} - 1$ as $k$ and $\xi$ vary. Further details on the construction of this figure are given in \cref{sec: simulation app}.}{\footnotesize\par}
\end{figure}

\clearpage

\section{Auxiliary Results and Discussion\label{app: additional}}

\subsection{Validity of Testing Procedures Based on Multiple Sample-Splitting\label{app: test}}

In this appendix, we discuss the application of \cref{alg: sequential aggregation} to testing procedures based on averaging over multiple splits of the same sample. We show that methods based on both $p$-values and $e$-values constructed with sample splitting continue to control the Type I error rate if they are aggregated sequentially with \cref{alg: sequential aggregation}. Both results follow from the ``Exchangeable Markov Inequality'' of \cite{ramdas2023randomized}. As before, let $D=(D_i)_{i=1}^{n}$ be independent and identically distributed according to a probability distribution $P$. Interest is in testing the null hypothesis $H_0:P\in\mathbf{P}$ for some collection of probability distributions $\mathbf{P}$.

\subsubsection{Methods Based on $p$-Values}  Suppose that we have access to a valid $p$-value $\hat{p}(\mathsf{s},D)$. That is, the statistic $\hat{p}(\mathsf{s},D)$ satisfies
\[
P\left\{ \hat{p}(\mathsf{s},D) \leq u \mid D_{\tilde{\mathsf{s}}} \right\} \leq u
\]
for all $u$ in $(0,1)$ and $P$ in $\mathbf{P}$. For example, a test-statistic could be chosen using the data in $D_{\tilde{\mathsf{s}}}$ and a $p$-value can be constructed based on this test statistic using the data in $D_{\mathsf{s}}$. For any collection $\mathsf{R}_{g,k}$ in $\mathcal{R}_{n,k,b}$, let
\begin{equation} \label{eq: agg p val}
a_\delta(\mathsf{R}_{g,k},D) = \frac{1}{g}\frac{1}{k} \sum_{i=1}^{g} \sum_{j=1}^k \mathbb{I}\left\{\hat{p}(\mathsf{s}_{i,j},D) \leq \delta \right\}
\end{equation}
denote the proportion of $p$-values that are less than or equal than $\delta$. \cite{ruger1978maximale}, \cite{meinshausen2009p}, and \cite{diciccio2020exact} observe that if $\mathsf{R}_{g,k}$ is constructed independently of the data $D$, then
\[
P\left\{ a_\delta(\mathsf{R}_{g,k},D) \geq c \right\} \leq \frac{\mathbb{E}\left[a_\delta(\mathsf{R}_{g,k},D)\right]}{c} \leq \frac{\delta}{c}
\]
by Markov's inequality, for all $P$ in $\mathbf{P}$. Thus, if $\delta$ and $c$ are chosen such that $\delta/c=\alpha$, then the test that rejects the null hypothesis $H_0$ if $a_\delta(\mathsf{R}_{g,k},D)$ is larger than $c$ has level $\alpha$.

The following theorem establishes that this test continues to be valid if the collection of sample-splits $\mathsf{R}_{g,k}$ is constructed sequentially with \cref{alg: sequential aggregation}. 
\begin{theorem}\label{thm: seq p val}
If the statistic $a_\delta(\mathsf{R}_{\hat{g},k},D)$ defined in \eqref{eq: agg p val} is constructed sequentially with \cref{alg: sequential aggregation}, then
\begin{equation} \label{eq: pval valid}
P\left\{ a_\delta(\mathsf{R}_{\hat{g},k},D) \geq c \right\} \leq \frac{\delta}{c}
\end{equation}
for all $P$ in $\mathbf{P}$.
\end{theorem}

\begin{proof}
We apply the following inequality, due to \cite{ramdas2023randomized}. 
\begin{theorem}[Theorem 1.1, \cite{ramdas2023randomized}] \label{eq: exchangeable Markov}
If $X_1,X_2,\ldots$ form an exchangeable sequence of integrable random variables, then
\begin{equation}
P\left\{ \exists t \geq1:  \frac{1}{t} \sum_{i=1}^t \vert X_i \vert \geq 1/a \right\} \leq a \mathbb{E}\left[\vert X_i\vert\right]
\end{equation}
for any $a>0$.
\end{theorem}
\noindent
Consequently, we have that
\begin{align}
P\left\{ a_\delta(\mathsf{R}_{\hat{g},k},D) \geq c \right\} 
& \leq P\left\{ \exists g \geq 1: a_\delta(\mathsf{R}_{g,k},D) \geq  c \right\} \\
& \leq  \frac{\mathbb{E}\left[\mathbb{I}\left\{\hat{p}(\mathsf{s}_{i,j},D) \leq \delta \right\}\right]}{c}  \leq\frac{\delta}{c}\nonumber
\end{align}
by \cref{eq: exchangeable Markov}, as required.\hfill
\end{proof}
 
\subsubsection{Methods Based on $e$-Values} Next, we consider settings where we have access to a valid $e$-value $\hat{e}(\mathsf{s},D)$ (see \cite{ramdas2023game} for a recent review). That is, the nonnegative statistic $\hat{e}(\mathsf{s},D)$ satisfies
\[
\mathbb{E}_P \left[ \hat{e}(\mathsf{s},D) \right] \leq 1
\]
for all $P$ in $\mathbf{P}$. For example, this setting applies to the ``Universal Inference'' procedure of \cite{wasserman2020universal}. Here, an estimator $\hat{P}(\tilde{\mathsf{s}})$ of $P$ is formed using the data $D_{\tilde{\mathsf{s}}}$ and is used in the split-likelihood ratio test statistic
\begin{equation} \label{eq: split ratio}
\hat{e}(\mathsf{s},D) = \inf_{P \in \mathbf{P}} \prod_{i \in \mathsf{s}} \frac{\text{d} \hat{P}(\tilde{\mathsf{s}})}{\text{d} P}(D_i)~.
\end{equation}
\cite{wasserman2020universal} prove that \eqref{eq: split ratio} is an $e$-value. For any collection $\mathsf{R}_{g,k}$ in $\mathcal{R}_{n,k,b}$, let
\begin{equation} \label{eq: agg e val}
a(\mathsf{R}_{g,k},D) = \frac{1}{g}\frac{1}{k} \sum_{i=1}^{g} \sum_{j=1}^k \hat{e}(\mathsf{s}_{i,j},D)
\end{equation} 
denote an aggregate $e$-value. See \cite{dunn2023gaussian} and \cite{tse2022note} for further discussion of aggregate $e$-values. Observe that
\begin{equation} \label{eq: e val valid}
P\left\{ a(\mathsf{R}_{g,k},D) \geq 1/\alpha \right\} \leq \alpha \mathbb{E}\left[\hat{e}(\mathsf{s},D)\right] \leq \alpha~,
\end{equation}
by Markov's inequality, for all $P$ in $\mathbf{P}$. Thus, the test that rejects the null hypothesis $H_0$ if $a(\mathsf{R}_{g,k},D)$ is larger than $1/\alpha$ has level $\alpha$. 

We again establish that this test continues to be valid if the collection of sample splits $\mathsf{R}_{g,k}$ is constructed sequentially with \cref{alg: sequential aggregation}.
\begin{theorem}
If the aggregate $e$-value $a_\delta(\mathsf{R}_{\hat{g},k},D)$ defined in \eqref{eq: agg e val} is constructed sequentially with \cref{alg: sequential aggregation}, then
\begin{equation} \label{eq: pval valid}
P\left\{ a(\mathsf{R}_{\hat{g},k},D) \geq \frac{1}{\alpha} \right\} \leq \alpha
\end{equation}
for all $P$ in $\mathbf{P}$.
\end{theorem}

\begin{proof}
The claim is established with an argument very similar to the proof of \cref{thm: seq p val}. Namely, the Markov inequality used to establish \eqref{eq: e val valid} can then be replaced by \cref{eq: exchangeable Markov}, as before.\hfill
\end{proof}

\subsection{Determining Suitable Levels of Precision for Cross-Validated Risk Estimation\label{app: lasso xi determine}}
In this section, we give the details supporting our application of \cref{alg: sequential aggregation} to cross-validated risk estimation, using data from \cite{casey2021experiment}. Recall from \cref{sec: cross val in practive}, that we apply the procedure, independently, to each component of the vector
\begin{align}
(b_{\lambda_1}(\mathsf{r}, D), \ldots, b_{\lambda_{p-1}}(\mathsf{r}, D))~,
\quad\text{where}\quad
b_{\lambda_i}(\mathsf{r},D) = ( a_{\lambda_i}(\mathsf{r}, D) - a_{\lambda_p}(\mathsf{r}, D) ) ~.\label{eq: mse difference app}
\end{align}
Let $\xi_i$ denote the error tolerance used for the $i$th component of the vector \eqref{eq: mse difference app}. We use a simple, data-driven procedure for choosing these tolerances. 

In particular, we draw a small number of cross-splits, i.e., $g=20$ and compute the aggregated difference
\begin{equation}\label{eq: initial estimates}
b_{\lambda_i}(\mathsf{R}_{g,k}, D) = \frac{1}{g} \sum^g_{i=1} b_{\lambda_i}(\mathsf{r}_i, D)
\end{equation}
for each $i$ in $1,\ldots,p-1$. Panel A of \cref{fig: lasso xi choice} displays the absolute value of these estimates in purple, where the $y$-axis has been transformed to a logarithmic scale. The absolute value of the differences \eqref{eq: initial estimates} are on the order $10^{-4}$ at the close-to-optimal values of $\lambda$, i.e., values of $\lambda$ above about $0.01$. The differences for smaller values of $\lambda$ are much larger, and, as indicated by \cref{fig: lasso demo}, have larger conditional variances. These qualitative features are not particularly sensitive to residual randomness. Panel B of \cref{fig: lasso xi choice} displays the quantiles of the difference \eqref{eq: initial estimates} across draws of $g=20$ cross-splits. 

\begin{figure}[t]
\begin{centering}
\caption{Determining Error Tolerances for Cross-Validated Risk Estimation}
\vspace{-20pt}
\label{fig: lasso xi choice}
\begin{tabular}{c}
\textit{Panel A: One Draw of Risk Estimates, The Choice of $\xi_i$}\tabularnewline
\includegraphics[scale=0.4]{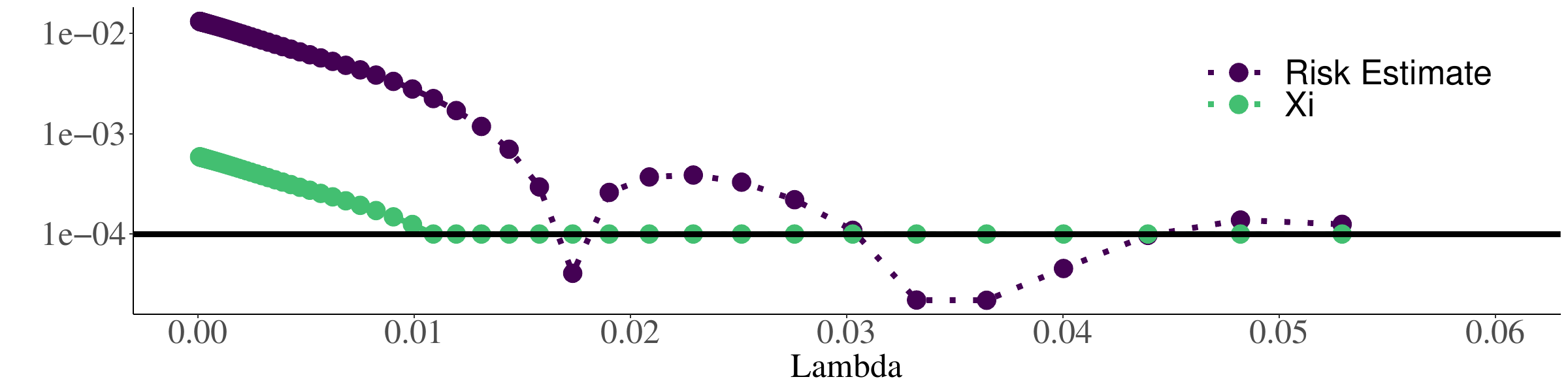}\tabularnewline
\textit{Panel B: Quantiles of Risk Estimates}\tabularnewline
\includegraphics[scale=0.4]{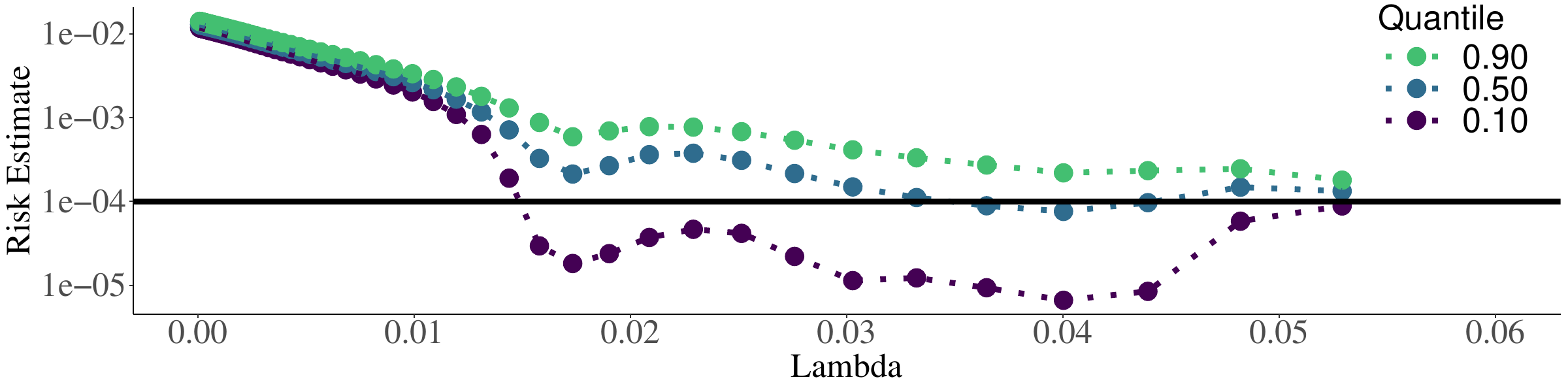}\tabularnewline
\end{tabular}
\par\end{centering}
\medskip{}
\justifying
\noindent{\footnotesize{}Notes: \cref{fig: lasso xi choice} displays auxiliary measurements relating to the the performance of \cref{alg: sequential aggregation} in the application to cross-validated Lasso, implemented in data from \cite{casey2021experiment}. Panel A displays the absolute value of the differences \eqref{eq: initial estimates} computed on a random sample of $g=20$ cross-splits, in purple, and the realized values of the error tolerances \eqref{eq: error tol lasso}, in green. Panel B displays the quantiles of the absolute value of the differences \eqref{eq: initial estimates}, over replications of random samples of $g=20$ cross-splits. In both cases, the $y$-axes is displayed in a logarithmic scale.}{\footnotesize\par}
\noindent\hrulefill
\end{figure}

These observations suggest the following approach for choosing the tolerances $\xi_i$. We would like to set $\xi_i = 10^{-4}$ for the large, close-to-optimal values of the regularization parameter, but would be willing to tolerate a larger error tolerance at smaller, sub-optimal values. This is practically important, as achieving the same level of precision for the small values of $\lambda$ requires aggregation over many more cross-splits (as the conditional variances are much larger). We operationalize this approach in the following way.  First, using the sample of $g=20$ cross-splits that we used to estimate \eqref{eq: initial estimates}, we compute the $p$-values
\begin{equation}
q_{\lambda_i}(\mathsf{R}_{g,k}, D) = 1 - \Phi\left(\frac{b_{\lambda_i}(\mathsf{R}_{g,k}, D) }{\sqrt{\hat{v}_{\lambda_i}(\mathsf{R}_{g,k}, D) }} \right)
\end{equation}
for each $i$ in $1,\ldots,p-1$, where $\hat{v}_{\lambda_i}(\mathsf{R}_{g,k}, D)$ denotes the sample-variance of the statistics $b_{\lambda_i}(\mathsf{r}_i, D)$ across the sample-splits. Let $\bar{\lambda}$ denote the largest value of the regularization parameter such that $q_{\bar{\lambda}}(\mathsf{R}_{g,k}, D) < 0.2$ and let the constant ``$\mathsf{scale}$'' solve the equality 
\begin{equation}
10^{-4} = \frac{b_{\bar{\lambda}}(\mathsf{R}_{g,k}, D)}{\mathsf{scale}}~.
\end{equation}
We set the error tolerances to
\begin{equation}\label{eq: error tol lasso}
\xi_i = \max\{ 10^{-4}, b_{\lambda_i}(\mathsf{R}_{g,k}, D)\times \mathsf{scale}  \}~.
\end{equation}
In other words,  we set $\xi_i=10^{-4}$ for all $i$ with $q_{\lambda_i}(\mathsf{R}_{g,k}, D) \geq 0.2$ and increase the error tolerance in proportion to $b_{\lambda_i}(\mathsf{R}_{g,k}, D)$ for all other values of $i$. The realized values of $\xi_i$ are displayed in Panel A of \cref{fig: lasso xi choice}. These error tolerances are used throughout the main text.

Panel A displays quantiles of the reproducibly aggregated statistic \eqref{eq: mse difference} across replications of the procedure, without truncation of the $y$-axis. Panel B of \cref{fig: lasso xi choice} displays estimates of the marginal reproducibility probability at each value of the regularization parameter, computed using $2,000$ replications of the procedure. Over the full range of the regularization parameter, the reproducibility error is very close to the nominal value $\beta = 0.05$. 

\begin{figure}[t]
\begin{centering}
\caption{Auxiliary Figures for Application to Cross-Validation}
\vspace{-20pt}
\label{fig: lasso xi auc}
\begin{tabular}{c}
\textit{Panel A: Mean-Squared Error Differences}\tabularnewline
\includegraphics[scale=0.4]{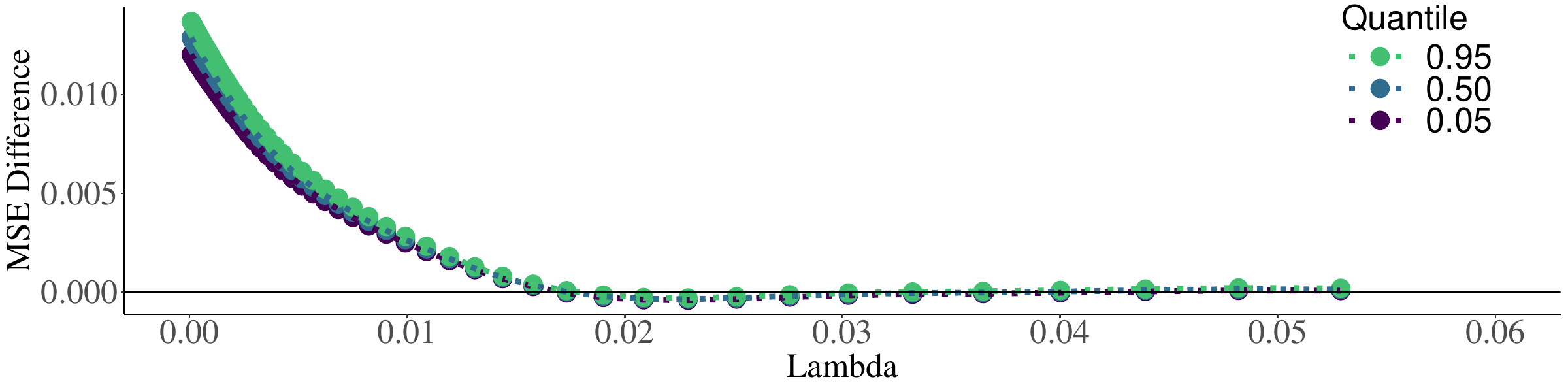}\tabularnewline
\textit{Panel B: Reproducibility}\tabularnewline
\includegraphics[scale=0.4]{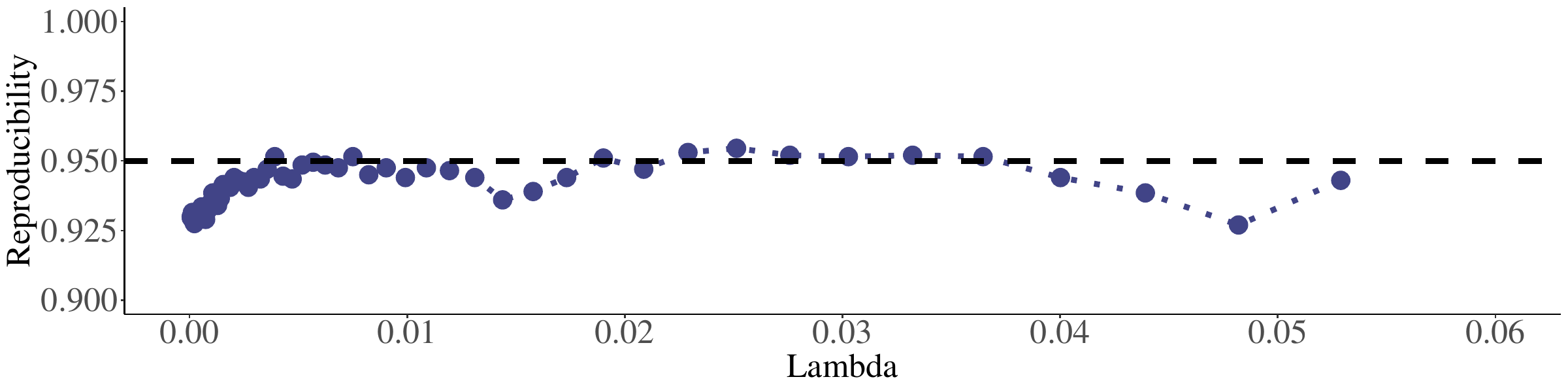}\tabularnewline
\end{tabular}
\par\end{centering}
\medskip{}
\justifying
\noindent{\footnotesize{}Notes: \cref{fig: lasso xi choice} displays auxiliary measurements relating to the the performance of \cref{alg: sequential aggregation} in the application to cross-validated Lasso, implemented in data from \cite{casey2021experiment}. Panel A is analogous to Panel A of \cref{fig: lasso reproduce}, but the $y$-axis is not truncated. Panel B displays estimates of the marginal reproducibility probability for each value of the regularization parameter, where the nominal reproducibility error is set to $\beta = 0.05$.}{\footnotesize\par}
\noindent\hrulefill
\end{figure}

\subsection{Sample Stability\label{app: M estimation}}

This appendix collects discussion and analysis concerning the $(r,q)$-sample stability $\sigma^{(r,q)}$ defined in \cref{def: sample stability}. First, in \cref{app: m est app}, we specialized our consideration to the case that the estimator $\hat{\eta}$ is a regularized empirical risk minimizer. Some closely related arguments are used in proof of Proposition 4 of \citet{austern2020asymptotics}. Second, in \cref{app: chen compare} we compare we compare the stable-stability decay condition \cref{def: split stability} to the conditions considered by \cite{chen2022debiased}. The proof of the Theorem stated in \cref{app: m est app} is given in \cref{app: train proof}

\subsubsection{Specialization to Regularized Empirical Risk Minimizers\label{app: m est app}}

Assume that the parameter $\eta$ is an element of some closed convex space $H\subseteq\mathbb{R}^{p}$. Consider the estimator
\begin{flalign}
\Psi\left(D_{\mathsf{s}},\eta\right) & =\frac{1}{b}\sum_{i\in\mathsf{s}}\psi\left(D_{i},\eta\right)\label{eq: psi average}\\
\hat{\eta} & =\underset{\eta\in H}{\arg\min}\left\{ \frac{1}{n-b}\sum_{i\in\tilde{\mathsf{s}}}\ell\left(D_{i},\eta\right)+\lambda_{1,n}\|\eta\|_{1}+\lambda_{2,n}\|\eta\|_{2}\right\} ,\label{eq: eta M estimator}
\end{flalign}
where $\psi\left(\cdot,\cdot\right)$ are $\ell\left(\cdot,\cdot\right)$ functions and $\lambda_{1,n},\lambda_{2,n}\geq0$ are penalty parameters.

Let $\nabla\ell\left(d,\eta\right)$ and $\nabla^{2}\ell\left(d,\eta\right)$
denote the gradient and Hessian of the function $\eta\mapsto\ell\left(d,\eta\right)$. Let $\nabla_j \ell\left(\cdot,\cdot\right)$ denote the $j$th component of $\nabla \ell\left(\cdot,\cdot\right)$.
Similarly, we write 
\begin{flalign*}
\bar{\nabla}_{\tilde{\mathsf{s}}}\ell\left(D,\eta\right) & =\frac{1}{n-b}\sum_{i\in\tilde{\mathsf{s}}}\nabla\ell\left(D_{i},\eta\right)\quad\text{and}\\
\bar{\nabla}_{\tilde{\mathsf{s}}}^{2}\ell\left(D,\eta\right) & =\frac{1}{n-b}\sum_{i\in\tilde{\mathsf{s}}}\nabla^{2}\ell\left(D_{i},\eta\right)
\end{flalign*}
for the empirical averages of the gradients and the Hessian evaluated on the data in the set $\tilde{\mathsf{s}}$.

We impose the following set of restrictions. The first two restrictions concern the curvature of the loss function. 
\begin{assumption}
\label{assu: invertibility}The minimum eigenvalue of $\bar{\nabla}_{\tilde{\mathsf{s}}}^{2}\ell\left(D,\eta\right)$
is bounded below by $\rho$ almost surely. 
\end{assumption}
\begin{assumption}
\label{assu: regularity} The loss function $\ell\left(\cdot,\cdot\right)$
is strictly convex and twice continuously differentiable in its second
argument. 
\end{assumption}
\begin{remark}
Under \cref{assu: invertibility}, if $\rho>0$, then the dimension $p$
of the parameter vector $\eta$ is less than the number of observations
$n$. To the best of our knowledge, analysis of the stability of $\ell_{1}$
regularized empirical risk minimization in the regime with $p$ larger
than $n$ is an open problem. Further analysis of the sample stability in
the high-dimensional regime is an interesting direction for further research.\hfill\qed
\end{remark}

The second two restrictions are statistical. The first concerns the expected curvature of the moment function $\psi\left(\cdot,\cdot\right)$. The second is a sub-exponential tail bound on the gradient $\nabla_j \ell\left(\cdot,\cdot\right)$.
\begin{assumption}
\label{assu: psi eta lipschitz}The function $\psi\left(\cdot,\cdot\right)$ satisfies
the inequality
\begin{equation}
\mathbb{E}\left[
  \left(
  \psi(D_i, \eta) - \psi(D_i, \eta^\prime)
  \right)^{r}
  \right] 
   \lesssim
  \mathbb{E} \left[\|\hat{\eta}-\hat{\eta}^{\prime}\|_{2}^r  \right]
\end{equation}
for each pair $\eta,\eta^\prime \in H$. 
\end{assumption}
\begin{assumption}
\label{assu: orlicz}The $\psi_1$-Orlicz norm bound 
\begin{equation}
\| \nabla_j \ell_j(D^\prime_i, \eta) - \nabla_j\ell(D_i, \eta) \| \lesssim \kappa^2
\end{equation}
holds for each pair $\eta \in H$. 
\end{assumption}
\begin{remark}
\cref{assu: psi eta lipschitz} is satisfied in many problems of interest. For example, \cite{chen2022debiased} show that similar conditions hold when $\psi(\cdot,\cdot)$ is given by various moment functions used for semiparametric causal inference. \cref{assu: orlicz} is equivalent to the assumption that $\nabla_j \ell_j(D^\prime_i, \eta) - \nabla\ell_j(D_i, \eta)$ is sub-exponential with parameter $\kappa^2$. \hfill\qed
\end{remark}

The following theorem gives a bound on the $(r,q)$-sample stability. 
\begin{theorem}
\label{thm: train}If \cref{assu: invertibility,assu: psi eta lipschitz,assu: regularity,assu: orlicz}, if $\rho+\lambda_{2,n}>0$, then 
\[
\sigma^{(r,q)}\lesssim 
C_r p \left(\frac{\sqrt{q}}{n-b}\frac{\kappa}{\rho+\lambda_{2,n}}\right)^{r}
     + C_r \left(\lambda_{1,n}\frac{p}{\rho+\lambda_{2,n}}\right)^{r}~,
\]
uniformly over all even integers $r$ and positive integers $q \leq n-b$, where $C_r$ is some positive constant that depends only on $r$.
\end{theorem}
\begin{remark}
A necessary and sufficient condition for $\hat{\eta}$ to be
consistent for the population risk minimizer associated with (\ref{eq: eta M estimator})
is that $\lambda_{1,n}=o\left((n-b)^{-1}\right)$. See e.g., \citet{fu2000asymptotics}.
Thus, so long as the penalty $\lambda_{1,n}$ is chosen in this regime,
the sample stability $\sigma^{(r,q)}$
will be
 \[
O\left(\left(\frac{\sqrt{q} \kappa }{n-b}\right)^{r}\right)
\]
up to terms that depend on the dimension of $\eta$, as required.\hfill\qed
\end{remark}

\subsubsection{Comparison to \cite{chen2022debiased}\label{app: chen compare}} 

In the main text, we restrict attention to sample stable estimators. In a recent article, \cite{chen2022debiased} show that, under some regularity conditions, if a condition related to, but partially stronger than, sample stability is satisfied, then sample-splitting is unnecessary for the consistency and asymptotic normality of DML estimators. In this section, we comment on the difference between the notion of sample stability used in this article and the related condition used in \cite{chen2022debiased}.

First, for ease of reference, we recall our definition of sample stability. We restrict attention to the case that $q=1$ and $r=2$, as this is what is relevant to make the comparison. Fix a set $\mathsf{s}\subseteq\mathcal{S}_{n,b}$ and let $i$ be an arbitrary element of $\mathsf{s}$.  Let $D^\prime$ denote an independent and identical copy of the data $D$. For each $i$ in $[n]$, let $\tilde{D}^{(i)}$ be constructed by replacing $D_i$ with $D^\prime_i$ in $D$. Let $I$ be a randomly selected element of $\tilde{\mathsf{s}}$.  In this paper, we refer to the quantity
\begin{align}
\sigma^{(2,1)}
&=\mathbb{E}\left[
\big(
\psi(D_i, \hat{\eta} \left(D_{\tilde{\mathsf{s}}}\right))
-
\psi(D_i, \hat{\eta} (\tilde{D}^{(I)}_{\tilde{\mathsf{s}}}))
\big)^2\right] \label{eq: sample stable app}
\end{align}
as the $(2,1)$-sample stability. In particular, in the main text, we restrict attention to settings where the bound
\begin{equation}\label{eq: stability specialize}
\sigma^{(2,1)} \lesssim \left(\frac{1}{n-b}\right)^2
\end{equation}
holds. 

\cite{chen2022debiased} consider settings where the nuisance parameter estimator $\hat{\eta}$ is not constructed with sample-splitting. Likewise, they consider moments of the form
\begin{align}
\tilde{\sigma}^{(2,1)}
&=\mathbb{E}\left[
\left(
\psi(D_i, \hat{\eta} \left(D\right))
-
\psi(D_i, \hat{\eta} (\tilde{D}^{(i)}))
\right)^2\right] ~, \label{eq: sample stable chen}
\end{align}
where $i$ is an arbitrary element of $[n]$. Observe that, in contrast to \eqref{eq: sample stable app}, the moment \eqref{eq: sample stable chen} is evaluated using the same observation $D_i$ used to perturb the nuisance parameter estimate---and that the nuisance parameter estimate is evaluated using the complete data. Roughly speaking, \cite{chen2022debiased} show that, under some regularity conditions, DML estimators are consistent and asymptotically normal if the bound
\begin{equation}\label{eq: stability chen}
\tilde{\sigma}^{(2,1)} \lesssim n^{-1}
\end{equation}
holds for a suitable choice of the function $\psi(\cdot,\cdot)$. Moreover, they show that, under certain conditions on the choice of tuning parameters, nuisance parameter estimators constructed with a bagged 1-nearest-neighbor estimator satisfies this bound.

Although the right-hand-side of the bound \eqref{eq: stability chen} is larger than the bound \eqref{eq: stability specialize}, control of the moment \eqref{eq: sample stable chen} can be more difficult to achieve than control of the moment \eqref{eq: sample stable app}. To see the substantive differences between the conditions \eqref{eq: stability specialize} and \eqref{eq: sample stable chen}, we consider an artificial, but illustrative, example adapted from an example given in the introduction to \cite{chernozhukov2018double}. Suppose that the observation $D_i$ consists of the real-valued quantities $Y_i$, $W_i$, and $X_i$. Furthermore, suppose that the nuisance parameter $\eta(x)$ denotes the conditional expectation $\mathbb{E}[W_i \mid X_i = x]$ and that
\begin{equation}\label{eq: stylized}
\psi(D_i, \eta) = Y_i (W_i - \mathbb{E}[W_i \mid  X_i]) = Y_i (W_i - \eta(x)) ~.
\end{equation}
Suppose that the nuisance parameter estimator takes the form
\begin{equation}\label{eq: weighted est}
\hat{\eta}(D_{\mathsf{\tilde{s}}})(x) = \sum_{j \in \mathsf{\tilde{s}}} w_j(x) \hat{\eta}_j(x)
\end{equation}
where $\tilde{\mathsf{s}}$ is an arbitrary subset of $[n]$, $\hat{\eta}_j(x)$ is an estimator of the nuisance parameter $\eta(x)$ based on the observation $D_j$, and the weight $w_j(x)$ again depends on the observation $D_j$ and the set $\tilde{\mathsf{s}}$. We refer the reader to discussion in Section 1 of \cite{chernozhukov2018double} and Section 2 of \cite{chen2022debiased} for further context on why understanding quantities of the form \eqref{eq: stylized} is essential for establishing the consistency of DML estimators.\footnote{In particular, establishing the Stochastic Equicontinuity of a given Neyman Orthogonal moment amounts to bounding scaled averages of quantities of the form \eqref{eq: stylized}.}

To study the sample stability \eqref{eq: sample stable app} in this example, i.e., in the sense considered in this article, we are interested in the difference
\begin{align}
\psi(D_i, \hat{\eta} \left(D_{\tilde{\mathsf{s}}}\right))
-
\psi(D_i, \hat{\eta} (\tilde{D}^{(I)}_{\tilde{\mathsf{s}}}))
& = Y_i( w_I(X_i) \hat{\eta}_I(X_i) - w^\prime_I(X_i) \hat{\eta}^\prime_I(X_i))\label{eq: example diff}
\end{align}
where $w^\prime_I(X_i)$ and $\hat{\eta}^\prime_I(X_i))$ denote the updated weight and estimator associated with the perturbed data $\tilde{D}^{(I)}_{\tilde{\mathsf{s}}}$. On the other hand, in order to consider the stability \eqref{eq: sample stable chen}, i.e., in the sense of \cite{chen2022debiased}, we are interested in the difference
\begin{align}
\psi(D_i, \hat{\eta} \left(D\right))
-
\psi(D_i, \hat{\eta} (\tilde{D}^{(i)}))
& =  Y_i( w_i(X_i) \hat{\eta}_i(X_i) - w^\prime_i(X_i) \hat{\eta}^\prime_i(X_i)) \label{eq: chen diff}
\end{align}
where, as before, $w^\prime_i(X_i)$ and $\hat{\eta}^\prime_i(X_i)$ denote the updated weight and estimator associated with the perturbed data $\tilde{D}^{(i)}$.

The quantities \eqref{eq: example diff} and \eqref{eq: chen diff} will behave differently if the nuisance parameter estimator exhibits overfitting. In particular, to take a highly stylized example, suppose that the weights satisfy
\begin{equation}
w_i(X_j) = 
\begin{cases}
\vert \tilde{\mathsf{s}} \vert^{-1/3}, & i = j,\\
\vert \tilde{\mathsf{s}} \vert^{-1}, & \text{otherwise.}
\label{eq: overfitting weights}
\end{cases}
\end{equation}
That is, the nuisance parameter estimator \eqref{eq: weighted est} meaningfully up-weights the estimate obtained from the evaluation point $X_j$, if $X_j$ is part of the training data $\tilde{\mathsf{s}}$, and is otherwise equally balanced. Overfitting to the training data, in a perhaps less stylized form, is a well-known property of some machine learning estimators. 

It is easy to see that, in this case, if all other quantities are bounded, we should expect that 
\begin{align}
\sigma^{(2,1)}
&=\mathbb{E}\left[
\big(
\psi(D_i, \hat{\eta} \left(D_{\tilde{\mathsf{s}}}\right))
-
\psi(D_i, \hat{\eta} (\tilde{D}^{(I)}_{\tilde{\mathsf{s}}}))
\big)^2\right] \nonumber \\
& = \mathbb{E}\left[
\big(
Y_i( w_I(X_i) \hat{\eta}_I(X_i) - w^\prime_I(X_i) \hat{\eta}^\prime_I(X_i))
\big)^2\right] 
 \lesssim \left(\frac{1}{ \vert \tilde{\mathsf{s}} \vert } \right)^{2}\label{eq: stable evaluate}
\end{align}
whereas
\begin{align}
\tilde{\sigma}^{(2,1)}
&=\mathbb{E}\left[
\left(
\psi(D_i, \hat{\eta} \left(D\right))
-
\psi(D_i, \hat{\eta} (\tilde{D}^{(i)}))
\right)^2\right] \nonumber \\
& =\mathbb{E}\left[
\left(
Y_i( w_i(X_i) \hat{\eta}_i(X_i) - w^\prime_i(X_i) \hat{\eta}^\prime_i(X_i))
\right)^2\right]
 \lesssim \left(\frac{1}{ \vert \tilde{\mathsf{s}} \vert } \right)^{2/3}~.\label{eq: stable chen evaluate}
\end{align}
In particular, we can see that all that is needed for us to expect the moment \eqref{eq: stable evaluate} to be approximately of order $O((n-b)^{-2})$, is that the \emph{randomly selected} weights $w_I(X_i)$ \emph{tend} to be close to, or smaller than, $(n-b)^{-1}$. The moment \eqref{eq: stable chen evaluate}, by contrast, is highly sensitive to the weight $w_i(X_i)$ placed on the evaluation point $X_i$.

The main point of this example is that establishing that a nuisance parameter estimator is stable, in the sense of \cite{chen2022debiased}, requires showing that it does not overfit its training data, i.e., that evaluations of points \emph{in-sample} are not overly-determined by a few observations. By contrast, all that is needed to establish that a nuisance parameter estimator is sample stable, in the sense used in this article, is that evaluations of points \emph{out-of-sample} are not overly concentrated on a small number of training samples. In general, we should expect that out-of-sample balance is more likely to occur in practice than in-sample balance.

\subsubsection{Proof of Theorem \ref{thm: train}\label{app: train proof}}

By \cref{assu: regularity,assu: invertibility},
the objective function for the estimator (\ref{eq: eta M estimator})
is strongly convex. Thus, there is a unique solution to (\ref{eq: eta M estimator})
for any data $D$. Let $\hat{\eta}$ and $\hat{\eta}^{\prime}$ denote
the solutions to (\ref{eq: eta M estimator}) for the data $D$ and
$\tilde{D}^{(\mathsf{q})}$ respectively. Let $\partial f\left(x\right)$ denote
the subgradient set of the function $x\mapsto f\left(x\right)$. The
Karush-Kuhn-Tucker condition for the program (\ref{eq: eta M estimator})
is given by 
\begin{equation}
\bar{\nabla}_{\tilde{\mathsf{s}}}\ell\left(D,\eta\right)+\lambda_{2,n}\hat{\eta}+\lambda_{1,n}\hat{z}=0,\label{eq: KKT}
\end{equation}
where $\hat{z}\in\partial\|\hat{\eta}\|_{1}$ is the subgradient associated
with the Lasso penalty. Observe that, in this case, $\hat{z}\in\text{sign}\left(\hat{\eta}\right)$,
where we set $\text{sign}\left(0\right)=\left[-1,1\right]$. Let $\hat{z}$
and $\hat{z}^{\prime}$ denote the subgradients obtained from $D$
and $\tilde{D}^{(\mathsf{q})}$, respectively. As $\ell\left(d,\cdot\right)$
is twice continuously differentiable under \cref{assu: regularity}
(i), we have that
\begin{flalign}
 & \bar{\nabla}_{\tilde{\mathsf{s}}}\ell\left(D,\hat{\eta}^{\prime}\right)+\lambda_{2,n}\hat{\eta}^{\prime}+\lambda_{1,n}\hat{z}\nonumber \\
 & =\bar{\nabla}_{\tilde{\mathsf{s}}}\ell\left(D,\hat{\eta}\right)+\lambda_{2,n}\hat{\eta}
 +\lambda_{1,n}\hat{z}+\left(\bar{\nabla}_{\tilde{\mathsf{s}}}^{(2)}\ell\left(D,\tilde{\eta}\right)+\lambda_{2,n}I_{d}\right)\left(\hat{\eta}-\hat{\eta}^{\prime}\right)\nonumber \\
 & =\left(\bar{\nabla}_{\tilde{\mathsf{s}}}^{(2)}\ell\left(D,\tilde{\eta}\right)
 +\lambda_{2,n}I_{d}\right)\left(\hat{\eta}-\hat{\eta}^{\prime}\right)\label{eq: Taylor term}
\end{flalign}
for some vector $\tilde{\eta}$, by a Taylor expansion and the optimality
condition (\ref{eq: KKT}). On the other hand, we have that 
\begin{flalign}
\bar{\nabla}_{\tilde{\mathsf{s}}}\ell\left(D,\hat{\eta}^{\prime}\right)+\lambda_{2,n}\hat{\eta}^{\prime}+\lambda_{1,n}\hat{z} 
& =\bar{\nabla}_{\tilde{\mathsf{s}}}\ell(\tilde{D}^{(\mathsf{q})},\hat{\eta}^{\prime})+\lambda_{2,n}\hat{\eta}^{\prime}+\lambda_{1,n}\hat{z}^{\prime}\nonumber \\
& \quad+\frac{1}{n-b}\left(\sum_{i \in \mathsf{q}} \nabla \ell(D^\prime_i, \hat{\eta}^{\prime}) - \nabla\ell(D_i, \hat{\eta}^{\prime})\right)
+\lambda_{1,n}\left(\hat{z}-\hat{z}^{\prime}\right)\nonumber \\
 & =\frac{1}{n-b}\left(\sum_{i \in \mathsf{q}} \nabla \ell(D^\prime_i, \hat{\eta}^{\prime}) - \nabla\ell(D_i, \hat{\eta}^{\prime})\right)
+\lambda_{1,n}\left(\hat{z}-\hat{z}^{\prime}\right)\label{eq: swap term}
\end{flalign}
again by the optimality condition (\ref{eq: KKT}). Thus, we have
that
\begin{flalign}
& \left(\bar{\nabla}_{\tilde{\mathsf{s}}}^{(2)}\ell\left(D,\tilde{\eta}\right)+\lambda_{2,n}I_{p}\right)\left(\hat{\eta}-\hat{\eta}^{\prime}\right) \nonumber \\
& =\frac{1}{n-b}\left(\sum_{i \in \mathsf{q}} \nabla \ell(D^\prime_i, \hat{\eta}^{\prime}) -\nabla \ell(D_i, \hat{\eta}^{\prime})\right)
+\lambda_{1,n}\left(\hat{z}-\hat{z}^{\prime}\right)\label{eq: basic equality}
\end{flalign}
by (\ref{eq: Taylor term}) and (\ref{eq: swap term}). Observe that
the the matrix $\bar{\nabla}_{\tilde{\mathsf{s}}}^{(2)}\ell\left(D,\tilde{\eta}\right)+\lambda_{2,n}I_{p}$
is invertible by \cref{assu: invertibility}. Thus, we find
that 
\begin{flalign*}
\hat{\eta}-\hat{\eta}^{\prime} & =\frac{1}{n-b}\left(\bar{\nabla}^{(2)}\ell\left(D,\tilde{\eta}\right)+\lambda_{2}I_{p}\right)^{-1}
\left(\sum_{i \in \mathsf{q}} \nabla \ell(D^\prime_i, \hat{\eta}^{\prime}) - \nabla\ell(D_i, \hat{\eta}^{\prime})\right)\\
 & +\lambda_{1,n}\left(\bar{\nabla}^{(2)}\ell\left(D,\tilde{\eta}\right)+\lambda_{2}I_{d}\right)^{-1}\left(\hat{z}-\hat{z}^{\prime}\right)
\end{flalign*}
and that consequently
\begin{flalign}
\|\hat{\eta}-\hat{\eta}^{\prime}\|_{2} 
& \lesssim\frac{1}{n-b}
\frac{ \| \sum_{i \in \mathsf{q}} \nabla \ell(D^\prime_i, \hat{\eta}^{\prime}) - \nabla\ell(D_i, \hat{\eta}^{\prime})\|_2}{\rho+\lambda_{2,n}} 
+\lambda_{1,n}\frac{\| \hat{z}-\hat{z}^{\prime} \|_2}{\rho+\lambda_{2,n}}  \nonumber \\
& \lesssim\frac{1}{n-b}
\frac{\| \sum_{i \in \mathsf{q}} \nabla \ell(D^\prime_i, \hat{\eta}^{\prime}) - \nabla\ell(D_i, \hat{\eta}^{\prime})\|_2}{\rho+\lambda_{2,n}} 
+\lambda_{1,n}\frac{p}{\rho+\lambda_{2,n}}\label{eq: agg lipschitz}
\end{flalign}
by \cref{assu: invertibility}. 

Now, observe that
\begin{align}\label{eq: stability re-write}
\sigma^{(r,q)}
 & = \mathbb{E}\left[
  \left(
  \psi(D_i, \hat{\eta}) - \psi(D_i, \hat{\eta}^\prime)
  \right)^{r}
  \right] \nonumber \\
  & \lesssim
  \mathbb{E} \left[\|\hat{\eta}-\hat{\eta}^{\prime}\|_{2}^r  \right] \nonumber \\
    & \lesssim\mathbb{E}\left[
    \left(
    \frac{1}{n-b}
    \frac{ \| \sum_{i \in \mathsf{q}} \nabla \ell(D^\prime_i, \hat{\eta}^{\prime}) - \nabla\ell(D_i, \hat{\eta}^{\prime})\|_2}{\rho+\lambda_{2,n}} 
   + \lambda_{1,n}\frac{p}{\rho+\lambda_{2,n}}
    \right)^r
    \right]\\
   & \lesssim 2^r \left(\frac{1}{n-b}\frac{1}{\rho+\lambda_{2,n}}\right)^{r}
    \mathbb{E}\left[ \| \sum_{i \in \mathsf{q}} \nabla \ell(D^\prime_i, \hat{\eta}^{\prime}) - \nabla\ell(D_i, \hat{\eta}^{\prime})\|^r_2 \right] \nonumber \\
    & +2^r \left(\lambda_{1,n}\frac{p}{\rho+\lambda_{2,n}}\right)^{r}~,
\end{align}
where the first inequality follows from \cref{assu: psi eta lipschitz}, the second inequality follows from \eqref{eq: agg lipschitz}, and the third inequality follows from the Binomial Theorem and Cauchy-Schwarz. By \cref{assu: orlicz}, we have that 
\begin{equation}
\| \sum_{i\in\mathsf{q}} \nabla_j \ell_j(D^\prime_i, \hat{\eta}^{\prime}) - \nabla\ell_j(D_i, \hat{\eta}^{\prime}) \|_{\psi_1}
 \lesssim q \| \nabla_j \ell_j(D^\prime_i, \hat{\eta}^{\prime}) - \nabla\ell_j(D_i, \hat{\eta}^{\prime}) \|_{\psi_1}
 \lesssim q \kappa ~.
\end{equation}
Consequently, we have that
\begin{align}
& \mathbb{E}\left[ \| \sum_{i \in \mathsf{q}} \nabla \ell(D^\prime_i, \hat{\eta}^{\prime}) - \nabla\ell(D_i, \hat{\eta}^{\prime})\|^r_2 \right] \nonumber \\
& \lesssim 
\sum_{j \in [p]}  \mathbb{E}\left[ \bigg\vert \sum_{i\in\mathsf{q}} \nabla_j \ell_j(D^\prime_i, \hat{\eta}^{\prime}) - \nabla\ell_j(D_i, \hat{\eta}^{\prime})  \bigg\vert^{r/2} \right] \tag{Jensen}\\
& 
\lesssim
p ((r/2)!)^{r/2} (\kappa q)^{r/2}~
\end{align}
where, for the final inequality, we have used the fact that $\| X \|_{r/2} \lesssim ((r/2)!) \|X\|_{\psi_1} $ (see e.g., Section 2.2 of \cite{van1996weak}). Putting the pieces together, we find that
\begin{equation}
\sigma^{(r,q)} 
\lesssim 2^r ((r/2)!)^{r/2} p \left(\frac{q^{1/2}}{n-b}\frac{\kappa^{1/2}}{\rho+\lambda_{2,n}}\right)^{r}
     +2^r \left(\lambda_{1,n}\frac{p}{\rho+\lambda_{2,n}}\right)^{r}~,
\end{equation}
as required.\hfill\qed

\section{Proofs for Results Stated in \cref{sec: Reproducible Aggregation}\label{sec: proof asymptotic}}

\subsection{Proof of \cref{eq: general reproducibility}\label{app: proof of general reproducibility}}
For each collection $\mathsf{r} = (\mathsf{s}_{j})_{j=1}^{k}$ in $\mathcal{R}_{n,k,b}$, we write 
\[
\bar{a}(\mathsf{r},D)=\frac{1}{k}\sum_{j=1}^{k}T(\mathsf{s}_{j},D) - \mathbb{E}\left[T(\mathsf{s}_{j},D)\mid D\right].
\]
Define the oracle stopping time 
\[
g^{\star}=
\underset{g\geq2}{\arg\min}
\left\{ v_{g,k}(D) \leq\frac{1}{2}\left(\frac{\xi}{z_{1-\beta/2}}\right)^{2} \right\}.
\]
Observe that $v_{g,k}(D)= (1/g)\cdot v_{1,k}(D)$ as the collections of cross-splits are drawn independently and identically. Moreover, we have that
\begin{equation}
\frac{\hat{v}(\mathsf{R}_{g,k},D)}{v_{g,k}(D)}
=
\frac{g \hat{v}(\mathsf{R}_{g,k},D)}{v_{1,k}(D)}
\overset{\text{a.s}}{\to}1\label{eq: var consistency}
\end{equation}
as $g\to\infty$, by the strong law of large numbers. Now, observe that
\begin{equation}
\frac{\hat{g} \cdot \hat{v}(\mathsf{R}_{\hat{g},k},D)}{v_{1,k}(D)}
\leq
\frac{\hat{g}}{2v_{1,k}(D)}\left(\frac{\xi}{z_{1-\beta/2}}\right)^{2}
\leq
\frac{\hat{g}\cdot\hat{v}(\mathsf{R}_{\hat{g}-1,k},D)}{v_{1,k}(D)},\label{eq: sandwich}
\end{equation}
by definition. Consequently, as $\hat{g}\to\infty$ and 
\[
g^\star \left(2v_{1,k}(D) \left(\frac{z_{1-\beta/2}}{\xi}\right)^{2}\right) \to 1
\]
as $\xi\to0$, we have that
\begin{equation}
\hat{g}/g^{\star}\overset{\text{a.s}}{\to} 1\label{eq: cond var a.s limit}
\end{equation}
as $\xi\to0$ by (\ref{eq: var consistency}) and (\ref{eq: sandwich}).

Define the objects
\begin{flalign}
U\left(\mathsf{R}_{g,k},D\right) & =\frac{1}{g^\star}\left(\sum_{i=1}^{g}\bar{a}\left(\mathsf{r}_{i},D\right)-\sum_{i=1}^{g^{\star}}\bar{a}\left(\mathsf{r}_{i},D\right)\right)\quad\text{and}\label{eq: maximal term proof}\\
Q\left(\mathsf{R}_{g,k},D\right) & =\left(1-\frac{g}{g^{\star}}\right)\left(a(\mathsf{R}_{g,k},D)-\mathbb{E}\left[T(\mathsf{s}_{j},D)\mid D\right]\right)\label{eq: cross product proof}
\end{flalign}
Observe that we can decompose
\begin{flalign}
 & \left(a(\mathsf{R}_{\hat{g},k},D)-a(\mathsf{R}_{\hat{g}^{\prime},k}^{\prime},D)\right)/\sqrt{2 v_{g^\star,k}(D)}\nonumber\\
 &  =  \sqrt{\frac{g^{\star}}{2 v_{1,k}(D)}}\left(a(\mathsf{R}_{g^{\star},k},D)-a(\mathsf{R}_{g^{\star},k}^{\prime},D)\right)\nonumber\\
 & +\sqrt{\frac{g^{\star}}{2 v_{1,k}(D)}} \left(U(\mathsf{R}_{\hat{g},k},D)-U(\mathsf{R}_{\hat{g}^{\prime},k}^{\prime},D)\right)
  +\sqrt{\frac{g^{\star}}{2 v_{1,k}(D)}}\left(Q(\mathsf{R}_{\hat{g},k},D)-Q(\mathsf{R}_{\hat{g}^{\prime},k}^{\prime},D)\right).\label{eq: sequential decomposition}
\end{flalign}
First, we have that
\[
\sqrt{\frac{g^{\star}}{2 v_{1,k}(D)}} \left(Q(\mathsf{R}_{\hat{g},k},D)-Q(\mathsf{R}_{\hat{g}^{\prime},k}^{\prime},D)\right) = o_p(1)
\]
almost surely as $\xi \to 0$ by \eqref{eq: cond var a.s limit}. To handle the terms involving \eqref{eq: maximal term proof}, fix $\varepsilon>0$ and observe that 
\begin{flalign}
 & P\left\{ \bigg\vert\sum_{i=1}^{\hat{g}}a(\mathsf{r}_{i},D)-\sum_{i=1}^{g^\star}a(\mathsf{r}_{i},D)\bigg\vert
 >
 \varepsilon\sqrt{g^\star}\mid D\right\} \nonumber \\
 & \leq P\left\{ \bigg\vert\sum_{i=1}^{\hat{g}}a(\mathsf{r}_{i},D)-\sum_{i=1}^{g^\star}a(\mathsf{r}_{i},D)\bigg\vert
 >
 \varepsilon\sqrt{g^\star},\hat{g}\in\left[g^\star\left(1-\varepsilon^{3}\right),g^\star\left(1+\varepsilon^{3}\right)\right]\mid D\right\} \nonumber \\
 & +P\left\{ \bigg\vert\sum_{i=1}^{\hat{g}}a(\mathsf{r}_{i},D)-\sum_{i=1}^{g^\star}a(\mathsf{r}_{i},D)\bigg\vert
 >
 \varepsilon\sqrt{g^\star},\hat{g}^{\prime}\not\in\left[g^\star\left(1-\varepsilon^{3}\right),g^\star\left(1+\varepsilon^{3}\right)\right]\mid D\right\} \nonumber \\
 & \leq P\left\{ \max_{g^\star\left(1-\varepsilon^{3}\right)\leq g\leq g^\star}\bigg\vert\sum_{i=1}^{g}a(\mathsf{r}_{i},D)\bigg\vert
 >\varepsilon\sqrt{g^\star}\mid D\right\} \label{eq: break into three pieces}\\
 & +P\left\{ \max_{g^\star\leq g\leq\left(1+\varepsilon^{3}\right)g^\star}\bigg\vert\sum_{i=1}^{g}a(\mathsf{r}_{i},D)\bigg\vert
 >\varepsilon\sqrt{g^{\star}}\mid D\right\} \nonumber \\
 & +P\left\{ \hat{g}\not\in\left[g^\star\left(1-\varepsilon^{3}\right),g^\star\left(1+\varepsilon^{3}\right)\right]\mid D\right\} \nonumber~.
\end{flalign}
The third term in (\ref{eq: break into three pieces}) is smaller than $\varepsilon$ for all sufficiently small $\xi$ by \eqref{eq: cond var a.s limit}. To handle the first two terms, observe that
\begin{flalign*}
 & P\left\{ \max_{g^\star\leq g\leq\left(1+\varepsilon^{3}\right)g^\star}\bigg\vert\sum_{i=1}^{g}a(\mathsf{r}_{i},D)\bigg\vert  > \varepsilon\sqrt{\hat{g}}\mid D\right\} \\
& \leq\frac{1}{\varepsilon^{2}}\frac{1}{g^\star}\Var\left(\sum_{i=1}^{\lfloor \varepsilon^{3} g^\star \rfloor}a(\mathsf{r}_{i},D)\mid D\right) \\
& \leq \varepsilon \Var(a(\mathsf{r}_{i},D)),
\end{flalign*}
almost surely, where the first inequality follows from Kolmogorov's maximal inequality
(see e.g., Proposition 2.3.16 of \cite{dembo2021probability}) and the second inequality follows from the fact that the cross-splits $\mathsf{r}_i$ are independent. An analogous bound
holds for the remaining term.  Thus, we have
\[
\sqrt{\frac{g^{\star}}{2 v_{1,k}(D)}}\left(U(\mathsf{R}_{\hat{g},k},D)-U(\mathsf{R}_{\hat{g}^{\prime},k}^{\prime},D)\right)=o_{p}\left(1\right)
\]
almost surely as $\xi\to0$. Hence, we have that 
\begin{flalign*}
 & P\left\{\big\vert a(\mathsf{R}_{\hat{g},k},D)-a(\mathsf{R}_{\hat{g}^{\prime},k}^{\prime},D)\big\vert>\xi\mid D\right\} \\
 & =P\left\{ \bigg\vert\sqrt{\frac{g^{\star}}{2 v_{1,k}(D)}}\left(\frac{1}{g^{\star}}\sum_{i=1}^{g^{\star}}a(\mathsf{r}_{i,k},D)-a(\mathsf{r}_{i,k}^{\prime},D)\right)\bigg\vert
 >
 z_{1-\beta/2}\mid D\right\} +o\left(1\right)\\
 & =1-\beta+o\left(1\right)
\end{flalign*}
as $\xi\to0$ almost surely, by the central limit theorem.\hfill\qed

\subsection{Proof of \cref{thm: sequential Jensen}}

Recall that the cross-split $\mathsf{r}$ is drawn independently and uniformly from the collection $\mathcal{R}_{n,k,b}$. The result is based on the bound
\begin{flalign}
\mathbb{E}\left[ f( a(\mathsf{R}_{g,k}, D) ) \right] 
& = 
\mathbb{E}\left[ \mathbb{E}\left[ f\left( \frac{1}{\hat{g}} \sum_{i=1}^{\hat{g}} a(\mathsf{r}_i, D) \right)  \mid \hat{g} \right] \right] \nonumber \\
& 
\leq \mathbb{E}\left[ \mathbb{E}\left[ \frac{1}{\hat{g}} \sum_{i=1}^{\hat{g}} f\left( a(\mathsf{r}_i, D) \right)  \mid \hat{g} \right] \right]  \tag{Jensen}\\
&
= \mathbb{E}\left[ \frac{1}{\hat{g}} \sum_{i=1}^{\hat{g}} f\left( a(\mathsf{r}_{i}, D) \right) \right] \nonumber \\
&
= 
\mathbb{E}\left[ f\left( a(\mathsf{r}, D) \right) \right] 
+
\mathbb{E}\left[ \frac{1}{\hat{g}} \sum_{i=1}^{\hat{g}} \left(f\left( a(\mathsf{r}_{i}, D) \right) - \mathbb{E}\left[ f\left( a(\mathsf{r}, D) \right) \mid D \right] \right)\right]~. \label{eq: jensen pre martingale}
\end{flalign}
In particular, it will suffice to show that the second term in \eqref{eq: jensen pre martingale} is weakly negative. To do this, we apply the following optimal stopping theorem. 

\begin{theorem}[Theorem 5.7.5, \cite{durrett2019probability}]\label{thm: optimal stopping}
Suppose that $X_g$ is a real-valued supermartingale, that $\mathcal{F}_g$ is the $\sigma$-algebra generated by $X_1,\ldots,X_g$, and that there exists some constant $B$ such that 
\begin{equation}\label{eq: martingale stopping condition}
\mathbb{E}\left[ \vert X_{g+1} - X_g \vert \mid \mathcal{F}_g \right] \leq B
\end{equation}
almost surely. If $\hat{g}$ is a stopping time with respect to the filtration $\{\mathcal{F}_g\}_{g=1}^\infty$, then $\mathbb{E}\left[ X_{\hat{g}} \right] \leq \mathbb{E}\left[X_1 \right]$.
\end{theorem}

\noindent Consider the sequence of random variables
\begin{equation}
X_g = \frac{1}{g} \sum^g_{i=1}  \left(f\left( a(\mathsf{r}_{i}, D) \right) - \mathbb{E}\left[ f\left( a(\mathsf{r}, D) \right) \mid D \right]  \right)
\end{equation}
Observe that 
\begin{flalign}
\mathbb{E}\left[ X_{g+1} \mid \mathcal{F}_g, D \right] 
&
= \mathbb{E}\left[ \frac{1}{g+1} \sum_{i=1}^{g+1}  \left(f\left( a(\mathsf{r}_{i}, D) \right) - \mathbb{E}\left[ f\left( a(\mathsf{r}, D) \right) \mid D \right]  \right) \mid \mathcal{F}_g, D \right] \nonumber \\
& 
= \frac{1}{g+1} \mathbb{E}\left[ f\left( a(\mathsf{r}_{i}, D) \right) - \mathbb{E}\left[ f\left( a(\mathsf{r}, D) \right) \mid D \right] \mid \mathcal{F}_g, D \right] \nonumber \\
& 
+ 
\frac{g}{g+1} \frac{1}{g} \sum_{i=1}^{g} \left(f\left( a(\mathsf{r}_{i}, D) \right) - \mathbb{E}\left[ f\left( a(\mathsf{r}, D) \right) \mid D \right]  \right) \nonumber \\
& = \frac{g}{g+1} X_g~.
\end{flalign}
Consequently, $X_g$ is a supermartingale. To verify the condition \eqref{eq: martingale stopping condition}, observe that, as the collection $\mathcal{R}_{n,k,b}$ is finite, there exists some positive, finite function $B(D)$ such that
\begin{equation}
f(a(\mathsf{r}, D)) \leq B(D)
\end{equation}
almost surely. Thus, we find that
\begin{flalign}
\mathbb{E}\left[ \vert X_{g+1} - X_g \vert \mid \mathcal{F}_g, D \right] 
&
= \mathbb{E}\left[ \bigg\vert \frac{1}{g+1} \sum_{i=1}^{g+1} f\left( a(\mathsf{r}_{i}, D) \right) - \frac{1}{g} \sum_{i=1}^{g} f\left( a(\mathsf{r}_{i}, D) \right) \bigg\vert \mid \mathcal{F}_g, D \right] \nonumber \\
&
\leq \frac{1}{g+1} \mathbb{E}\left[ \vert f\left( a(\mathsf{r}, D) \right) \vert \right]
+
\frac{1}{g(g+1)} \sum^g_{i=1} \mathbb{E}\left[ \vert f\left( a(\mathsf{r}_i, D) \right) \vert \right] \nonumber \\
&
\leq \frac{2}{g+1} B(D)
\end{flalign}
almost surely. Hence, by \cref{thm: optimal stopping}, we have
\begin{flalign}
&\mathbb{E}\left[ \frac{1}{\hat{g}} \sum_{i=1}^{\hat{g}} \left(f\left( a(\mathsf{r}_{i}, D) \right) - \mathbb{E}\left[ f\left( a(\mathsf{r}, D) \right) \mid D \right] \right) \mid D \right] \nonumber \\
&\leq
\mathbb{E}\left[ f\left( a(\mathsf{r}, D) \right) - \mathbb{E}\left[ f\left( a(\mathsf{r}, D) \right) \mid D \right]  \mid D\right] = 0
\end{flalign}
almost surely, completing the proof.\hfill\qed

\section{General Non-Asymptotic Results\label{sec: general results}}

In this appendix, we give a general, non-asymptotic analysis of the performance of \cref{alg: sequential aggregation}. We begin, in \cref{sec: concentration and normal approx app}, by studying the concentration and normal approximation of the aggregate statistic $a(\mathsf{R}_{k,g}, D)$. Equipped with these results, in \cref{sec: app Reproducibility}, we characterize the accuracy of the nominal error rate of \cref{alg: sequential aggregation}. In \cref{sec: proof outline}, we show that these general results imply the results stated in \cref{sec: Guarantees}, by restricting attention to sample-stable, cross-split statistics. Proofs are given in \cref{sec: construction,sec: general main proofs}.

\subsection{Concentration and Normal Approximation\label{sec: concentration and normal approx app}}

We begin by giving a non-asymptotic characterization of the concentration and normal approximation of the statistic $a(\mathsf{R}_{g,k}, D)$. To state these results, we require a slightly more general definition of sample stability. 

\begin{defn}[General Sample Stability] \label{def: sample stability app}
Let $D^\prime$ denote an independent and identical copy of the data $D$. For each $\mathsf{q}\subseteq [n]$, let $\tilde{D}^{(\mathsf{q})}$ be constructed by replacing $D_i$ with $D^\prime_i$ in $D$ for each $i$ in $\mathsf{q}$. Fix a set $\mathsf{s}\subseteq\mathcal{S}_{n,b}$. Let $i$ be an arbitrary element of $\mathsf{s}$. Let $\mathsf{q}$ be a randomly selected subset of $\tilde{\mathsf{s}}$ of cardinality $q$. We refer to the quantities
\begin{align*}
\sigma_{\mathsf{valid}}^{(r)}
& =\mathbb{E}\left[
\big\vert 
\psi(D_I, \hat{\eta} \left(D_{\tilde{\mathsf{s}}}\right))
-
\psi(D^\prime_I, \hat{\eta} \left(D_{\tilde{\mathsf{s}}}\right))
\big\vert^{r}\right]
\quad\text{and}\\
\sigma_{\mathsf{train}}^{(r,q)}
&=\mathbb{E}\left[
\big\vert
\psi(D_i, \hat{\eta} \left(D_{\tilde{\mathsf{s}}}\right))
-
\psi(D_i, \hat{\eta} (\tilde{D}^{(\mathsf{q})}_{\tilde{\mathsf{s}}}))
\big\vert^{r}\right]
\end{align*}
as the $r$th-order validation and $(r,q)$th-order training sample stabilities, respectively. Similarly, we refer to the quantity 
\begin{equation} \label{eq: sample stability app}
\sigma^{(r)}_{\mathsf{max}} = \max\left\{\sigma_{\mathsf{valid}}^{(r)}, \sigma_{\mathsf{train}}^{(r,1)}\right\}
\end{equation}
as the $r$th-order full sample stability.
\end{defn}
\noindent In other words, we define separate stabilities associated with resampling data in the subsets $\mathsf{s}$ and $\tilde{\mathsf{s}}$. We refer to these subsets as the ``validation" and ``training'' samples, respectively. The training sample stability $\sigma_{\mathsf{train}}^{(r,q)}$ is identical to the sample stability defined in \cref{sec: Guarantees}.

With this in place, we give a large deviations bound for the statistic $a(\mathsf{R}_{g,k}, D)$ about its conditional mean. The primary difficulty is accommodating the dependence in the summands of $a(\mathsf{R}_{g,k},D)$ across elements of the same cross-split. We tackle this problem through the method of exchangeable pairs \citep{stein1986approximate}. In particular, we construct a suitable exchangeable pair with a coupling argument due to \cite{chatterjee2005concentration} and apply a method for deriving concentration inequalities with exchangeable pairs, also due to \cite{chatterjee2005concentration,chatterjee2007stein}. Some preliminary results that facilitate the application of this approach are given in \cref{sec: construction}. Proofs for results stated in this appendix are then given in \cref{sec: general main proofs}. 

\begin{theorem}
\label{thm: cross split concentration appendix}Suppose that \cref{assu: invariance,assu: linearity} hold, that the data $D$ are independently and identically distributed, and that the statistic $a(\mathsf{R}_{1,k}, D)$ has a non-zero variance conditional on $D$ almost surely. If the fourth-order split stability $\zeta^{(4)}$ is finite, then, for each $\varepsilon>0$, the inequality
\begin{align}
P\left\{ \vert a(\mathsf{R}_{g,k}, D) - \bar{a}\left(D\right) \vert
\leq 
\sqrt{\frac{2^4(2 - \varphi k - \varphi)^2}{\delta}\frac{\Gamma_{k,\varphi,b} \log(\varepsilon^{-1})}{g}}
\mid D\right\} \geq 1 - \frac{\varepsilon}{2}
\label{eq: cross-splitting concentration app}
\end{align}
holds with probability greater than $1-\delta$ as $D$ varies, where 
\begin{align}
\Gamma_{k,\varphi,b} & = \left(\frac{1-k\varphi}{1-\varphi}\right)4\sigma^{(2)}_{\mathsf{max}} + \left(\frac{k\varphi - \varphi}{1-\varphi}\right)\sigma^{(2,b-1)}_{\mathsf{train}}
\label{eq: Gamma 1 def}
\end{align}
and $\varphi = b/n$.
\end{theorem}

\begin{remark}
The quantity $\Gamma_{k,\varphi,b}$ interpolates between $\sigma^{(2)}_{\mathsf{max}}$ and $\sigma^{(2,b-1)}_{\mathsf{train}}$ as $k$ varies between its bounds $1$ and $1/\varphi$. In the case of independent splitting, i.e., $k = 1$, the rate of concentration for the statistic $a(\mathsf{R}_{g,k}, D)$ is driven by the full sample stability $\sigma^{(2)}_{\mathsf{max}}$. On the other hand, in the case of cross-splitting, concentration depends only on the training sample stability $\sigma^{(2,b-1)}_{\mathsf{train}}$. In a previous draft of this paper, we give estimates of $\sigma^{(2)}_{\mathsf{valid}}$ and $\sigma^{(2,b-1)}_{\mathsf{train}}$ in applications to treatment effect estimation and cross-validated risk estimation. We show that the training stability is dramatically smaller than the validation stability and decreases rapidly as $b$ decreases.\footnote{See Figure 3 of the preprint posted at \url{https://arxiv.org/abs/2311.14204v2}.} To see this, observe that we should next expect the validation sample-stability $\sigma_{\mathsf{valid}}^{(r)}$ to change with $k$. That is, residual randomness is smaller for full cross-splitting, i.e., $k=n/b$, than for independent splitting.\hfill$\blacksquare$
\end{remark}

Next, we derive bounds for the centered conditional moments of $a\left(\mathsf{R}_{g,k}, D\right)$ through an argument closely related to the proof of \cref{thm: cross split concentration appendix}. Bounds of this form are known as Burkholder-Davis-Gundy inequalities \citep{burkholder1973distribution}.
\begin{theorem}
\label{cor: Moment bound}Suppose that \cref{assu: invariance,assu: linearity} hold and that the data $D$ are independent and identically distributed. If $r=2^{c-1}$ for some positive integer $c$ and the $4r$th-order split stability $\zeta^{(4r)}$ is finite, then, for each $\delta > $, the inequality
\begin{align}
\label{eq: applied BDG}
\mathbb{E}\left[\left(a\left(\mathsf{R}_{g,k}, D\right)-\bar{a}\left(D\right)\right)^{2r} \mid D\right]
 & \leq 
(2r-1)^r \left(\frac{2^4(2-\varphi k - \varphi)^2}{g}\right)^r \Gamma_{k,\varphi,b}^{(r)}~,
\end{align}
holds with probability greater than $1-\delta$ as $D$ varies, where
\begin{align}
 \Gamma_{k,\varphi,b}^{(r)}   &= 
 \left(\frac{1-k\varphi}{1-\varphi}\right) 
 2^{2r} \sigma_{\mathsf{max}}^{(2r,1)} 
 + \left(\frac{k\varphi - \varphi}{1-\varphi}\right) \sigma_{\mathsf{valid}}^{(2r,b-1)}
\end{align}
and $\varphi = b/n$.
\end{theorem}

\begin{remark}
By setting $r=1$, the inequality \eqref{eq: applied BDG} gives a variance bound. In this case, the right-hand side of the inequality \eqref{eq: applied BDG} is equal to negative one times the inverse of the right-hand side of the inequality \eqref{eq: cross-splitting concentration app}. In this sense, the concentration inequality \cref{thm: cross split concentration appendix} can be thought of as an exponential Efron-Stein inequality adapted to the dependence inherent in our problem \citep{boucheron2003concentration}. For values of $r$ greater than 1, the unconditional quantity $\Gamma_{k,\varphi,b}^{(r)}$ again interpolates between the $2r$th order full sample stability $\sigma_{\mathsf{max}}^{(2r,1)}$ and the $(2r,b-1)$th order training sample stability $\sigma_{\mathsf{valid}}^{(2r,b-1)}$ as $k$ increases from $1$ to $1/\varphi$. \hfill\qed
\end{remark}

\begin{remark}
An analogous, unconditional, result will follow from the same argument. In particular, under the same conditions, we have that 
\begin{equation}
\mathbb{E}\left[\left(a\left(\mathsf{R}_{g,k}, D\right)-\bar{a}\left(D\right)\right)^{2r}\right]
\leq 
(2r-1)^r \left(\frac{2^4(2-\varphi k - \varphi)^2}{g}\right)^r \Gamma_{k,\varphi,b}^{(r)}~.
\end{equation}
We opt to state the conditional version of the result thought, as, arguably, the conditional behavior of the statistic is of primary interest.
\hfill\qed
\end{remark}

In turn, we bound the distance between the conditional distribution of $a\left(\mathsf{R}_{g,k}, D\right)$ and the standard normal distribution. The result follows by combining the bounds derived in \cref{cor: Moment bound} with a standard Berry-Esseen inequality \citep{shevtsova2011absolute}.

\begin{theorem} \label{thm: normal approximation appendix}
Suppose that \cref{assu: invariance,assu: linearity} hold, that the data $D$ are independently and identically distributed, and that the statistic $a(\mathsf{R}_{1,k}, D)$ has a non-zero variance conditional on $D$ almost surely. If the eighth-order split stability $\zeta^{(8)}$ is finite, then for all $g$, the inequality
\begin{flalign}
&\sup_{z \in \mathbb{R}}
\bigg\vert 
P\left\{ \frac{a\left(\mathsf{R}_{g,k}, D\right)-\bar{a}\left(D\right)}{\sqrt{v_{g,k}\left(D\right)}} \leq z \mid D\right\}
-
\Phi(z)  \bigg\vert \nonumber \\
&\quad\quad\quad\quad\lesssim 
\frac{2^6 3^2}{\delta}\frac{\left(2-\varphi k-\varphi \right)^{3}}{g^{1/2}} \left(\frac{(\Gamma_{k,\varphi,b}^{(2)})^{1/2}}{v_{1,k}(D)}\right)^{3/2} \label{eq: cross-splitting berry esseen appendix}
\end{flalign}
is satisfied with probability greater than $1-\delta$ as $D$ varies, where $\Phi(\cdot)$ is the standard normal c.d.f.
\end{theorem}

\begin{remark}
To the best of our knowledge, the only Stein's method central limit theorem in the Kolmogorov distance, applicable to our setting, is given in \cite{zhang2022berry}, which generalizes the argument of \cite{shao2019berry}. In \cref{app: zheng comparion}, we show that this central limit theorem implies a bound that does not reduce as either $g$ or $k$ increase.\hfill$\blacksquare$
\end{remark}

\subsection{Reproducibility\label{sec: app Reproducibility}}

We now characterize accuracy of the nominal reproducibility error of \cref{alg: sequential aggregation}.
\begin{theorem}
\label{thm: stable symmetric reproduce appendix}
Suppose that the conditional variance $\Var(a(\mathsf{r}, D)\mid D)$ is strictly positive, almost surely, where $\mathsf{r}$ denotes a random collection drawn uniformly from $\mathcal{R}_{n,k,b}$. Suppose that the collections $\mathsf{R}_{\hat{g},k}$ and $\mathsf{R}^\prime_{\hat{g}^\prime,k}$ are independently obtained using \cref{alg: sequential aggregation}. If \cref{assu: invariance,assu: linearity} hold, the data $D$ are independent and identically distributed, and the eighth-order split stability $\zeta^{(8)}$ is finite, then, for all sufficiently small $\xi$, the inequality
\begin{align}
&P\left\{ \vert a(\mathsf{R}_{\hat{g},k},D)-a(\mathsf{R}_{\hat{g}^{\prime},k}^{\prime},D) \vert \leq\xi\mid D\right\} -\left(1-\beta\right)\vert \nonumber\\
& \quad\quad\quad\quad
\lesssim
 \left(
\frac{\xi}{z_{1-\beta/2}}
\frac{(2-\varphi k-\varphi)^{4}\Gamma^{(1)}_{k,\varphi,b} (\Gamma_{k,\varphi,b}^{(2)})^{1/2}}{\delta^{3/2}(v_{1,k}\left(D\right))^{5/2}}
\right)^{1/2} \nonumber \\
& \quad\quad\quad\quad \quad\quad \cdot \log^{3/4} \left( \frac{z_{1-\beta/2}}{\xi} \frac{\delta (v_{1,k}(D))^2}{(2-\varphi k - \varphi)^3 (\Gamma^{(2)}_{k,\varphi,b})^{3/4}} \right)~,
\label{eq: reproduce BE bound app}
\end{align}
holds with probability greater than $1-\delta$ as $D$ varies.
\end{theorem}

\begin{remark}
Recall the definition of the oracle stopping time $g^\star$ given in \eqref{eq: oracle splits main}. The bound \eqref{eq: reproduce BE bound app} is obtained by decomposing the reproducibility error into two terms. The first term involves the error in a normal approximation to the difference $a(\mathsf{R}_{g^\star,k}, D) - a(\mathsf{R}^\prime_{g^\star,k}, D)$. As $g^\star$ is deterministic conditional on $D$, a bound on this term follows from an argument very similar to the proof of \cref{thm: normal approximation appendix}. The second term involves the differences $a(\mathsf{R}_{\hat{g},k}, D) - a(\mathsf{R}_{g^\star,k}, D)$ and $a(\mathsf{R}^\prime_{\hat{g}^\prime,k}, D) - a(\mathsf{R}^\prime_{g^\star,k}, D)$. The key step in bounding these quantities involves deriving a high probability bound for the difference $\hat{g} - g^\star$. This follows by combining a concentration inequality analogous to \cref{thm: cross split concentration appendix} for the estimator $\hat{v}_{g,k}(D)$ with a maximal inequality due to \cite{steiger1970bernstein}.\hfill$\blacksquare$
\end{remark}

\subsection{Specialization to Cross-Split, Sample-Stable Statistics\label{sec: proof outline}}

In this section, we specialize the results stated in \cref{sec: concentration and normal approx app,sec: app Reproducibility} to cross-fit, sample-stable statistics. That is, we focus attention to the setting where $n = k\cdot b$ and impose \cref{def: split stability}. By doing this, we prove each of the results stated in \cref{sec: Guarantees}.

We begin by stating the following corollary of \cref{thm: cross split concentration appendix}, which is referenced in \cref{fn: concentration}. 

\begin{cor}
\label{thm: cross split concentration}Suppose that \cref{assu: invariance,assu: linearity,def: split stability} hold, that the data $D$ are independently and identically distributed, and that the statistic $a(\mathsf{R}_{1,k}, D)$ has a non-zero variance conditional on $D$ almost surely. If the fourth-order split stability $\zeta^{(4)}$ is finite, then, for each $\varepsilon>0$, the inequality
\begin{align}
P\left\{ \vert a(\mathsf{R}_{g,k}, D) - \bar{a}\left(D\right) \vert
\leq 
\sqrt{\frac{b-1}{n^2} \frac{1}{g} \frac{\log(\varepsilon^{-1})}{\delta}}
\mid D\right\} \geq 1 - \varepsilon
\label{eq: cross-splitting concentration}
\end{align}
holds with probability greater than $1-\delta$ as $D$ varies.
\end{cor}

\begin{proof}
\cref{thm: cross split concentration} follows immediately from \cref{thm: cross split concentration appendix} by restricting attention to the case that $n = k\cdot b$ and applying \cref{def: split stability}. In particular, observe that if $n = k\cdot b$, then $\Gamma_{k,\varphi,b} = \sigma^{(2,b-1)}_{\mathsf{train}}$ and 
\begin{equation}\label{eq: rewrite scale term}
(2 - \varphi k - \varphi)^2 = (1 - \varphi)^2 = \left(\frac{n-b}{b}\right)^2~.
\end{equation}
Applying \cref{def: split stability} then gives
\begin{equation}\label{eq: rewrite Gamma app}
\Gamma_{k,\varphi,b} = \sigma^{(2,b-1)}_{\mathsf{train}} \lesssim \left(\frac{\sqrt{b-1}}{n-b}\right)^2~.
\end{equation}
Plugging \eqref{eq: rewrite scale term} and \eqref{eq: rewrite Gamma app} into \cref{thm: cross split concentration appendix} gives the desired bound.\hfill
\end{proof}

Next, we show that \cref{cor: Moment bound} implies the following generalization of \cref{thm: cross split moment} to higher-order moments. \cref{thm: cross split moment} is a strictly weaker result, as it only holds for the second moment. 
\begin{cor}
\label{thm: cross split moment app}Suppose that \cref{assu: invariance,assu: linearity,def: split stability} hold, the data $D$ are independently and identically distributed, and $n = k\cdot b$. If the 4$r$th-order split stability $\zeta^{(4)}$ is finite, then, for each $\delta > 0$, the inequality
\begin{equation} \label{eq: variance bound display app}
\mathbb{E}\left[\left(a\left(\mathsf{R}_{g,k}, D\right)-\bar{a}\left(D\right)\right)^{2r} \mid D \right] \lesssim \left(\frac{b-1}{n^2} \frac{1}{g}\right)^{r/2}
\end{equation}
holds with probability greater than $1-\delta$ as $D$ varies.
\end{cor}

\begin{proof}
If $n=k\cdot b$, then $\Gamma_{k,\varphi,b}^{(r)} = \sigma_{\mathsf{valid}}^{(2r,b-1)}$. Applying \cref{def: split stability} then gives
\begin{equation}\label{eq: rewrite Gamma higher app}
\Gamma_{k,\varphi,b}^{(r)}  \lesssim \left(\frac{\sqrt{b-1}}{n-b}\right)^{2r}~.
\end{equation}
Plugging \eqref{eq: rewrite scale term} and \eqref{eq: rewrite Gamma higher app} into \eqref{eq: applied BDG} gives
\begin{equation}
\label{eq: applied moment bound r gtr 1}
\mathbb{E}\left[\left(a\left(\mathsf{R}_{g,k}, D\right)-\bar{a}\left(D\right)\right)^{2r}\right]
\lesssim
\left(\frac{1}{nkg}\right)^r~,
\end{equation}
where we omit constants the depend only on $r$.\hfill
\end{proof}

In turn, we specialize the non-asymptotic central limit theorem \cref{thm: normal approximation appendix} cross-fit, sample-stable statistics. This result suggests that, although the aggregate statistic $a(\mathsf{R}_{g,k}, D)$ is more concentrated at larger values of $k$, the quality of a normal approximation to its conditional distribution does not necessarily improve. This result is referenced in \cref{fn: normal approx}.
\begin{cor} \label{thm: normal approximation}
Suppose that \cref{assu: invariance,assu: linearity,def: split stability} hold, the data $D$ are independently and identically distributed, $n = k\cdot b$, and the conditional variance $v_{1,k}(D) = \Var(a(\mathsf{r}, D) \mid D)$ is strictly positive almost surely, where $\mathsf{r}$ denotes a random collection drawn uniformly from $\mathcal{R}_{n,k,b}$. If the 8th-order split stability $\zeta^{(8)}$ is finite, then, for all $g$, the inequality
\begin{flalign}
\label{eq: cross-splitting berry esseen}
\sup_{z \in \mathbb{R}}
\bigg\vert 
P\left\{ \frac{a\left(\mathsf{R}_{g,k}, D\right)-\bar{a}\left(D\right)}{\sqrt{v_{g,k}\left(D\right)}} \leq z \mid D\right\}
-
\Phi(z)  \bigg\vert 
& 
\lesssim 
\frac{1}{\delta} \left(\frac{1}{n} \frac{1}{k} \frac{b-1}{v_{1,k}(D)}\right)^{3/2} \sqrt{\frac{1}{g}} 
\end{flalign}
is satisfied with probability greater than $1-\delta$ as $D$ varies, where $\Phi(\cdot)$ is the standard normal c.d.f.
\end{cor}

\begin{proof}
As before, \cref{thm: normal approximation} follows by restricting attention to the case that $n=k\cdot b$ and imposing \cref{def: split stability}. In particular, \cref{thm: normal approximation} is obtained by plugging  the bounds \eqref{eq: rewrite scale term} and \eqref{eq: rewrite Gamma higher app} into \eqref{eq: cross-splitting berry esseen appendix}.\hfill
\end{proof}

\begin{remark} In this case, the quantities $g$ and $k$ enter differently. In particular, observe that, by applying the variance bound \eqref{eq: variance bound display} for the case $g=1$, to lower bound the rate \eqref{eq: cross-splitting berry esseen}, we get
\begin{equation}
\left(\frac{1}{n} \frac{1}{k} \frac{1}{v_{1,k}(D)}\right)^{3/2} \sqrt{\frac{1}{g}} \gtrsim \sqrt{\frac{1}{g}}~. 
\end{equation}
Written differently, \cref{thm: normal approximation} suggests that the quality of a normal approximation to the conditional distribution of the aggregate statistic $a\left(\mathsf{R}_{g,k}, D\right)$ does not increase with $k$---it only depends on $g$.\hfill$\blacksquare$
\end{remark}

Finally, we show that \cref{thm: stable symmetric reproduce} is a corollary of \cref{thm: stable symmetric reproduce appendix} by restricting attention to the case that $n=k \cdot b$ and applying \cref{def: split stability}. We restate \cref{thm: stable symmetric reproduce} below as a corollary for ease of reference.

\begin{cor}
\label{cor: stable symmetric reproduce}
Suppose that the collections $\mathsf{R}_{\hat{g},k}$ and $\mathsf{R}^\prime_{\hat{g}^\prime,k}$ are independently obtained using \cref{alg: sequential aggregation}.  If \cref{assu: invariance,assu: linearity,def: split stability} hold, the conditional variance $v_{1,k}(D) = \Var(a(\mathsf{r}, D)\mid D)$ is strictly positive, almost surely, the data $D$ are independent and identically distributed, and the eighth-order split stability $\zeta^{(8)}$ is finite, then for all sufficiently small $\xi$, the inequality
\begin{align}
& \bigg\vert P\bigg\{\big\vert a(\mathsf{R}_{\hat{g},k}, D) - a(\mathsf{R}^\prime_{\hat{g}^\prime,k}, D) \big\vert \geq \xi \mid D\bigg\} 
 - \beta \bigg\vert \nonumber \\
&\quad\quad\quad\quad
\lesssim 
\frac{1}{\delta^{3/4}}
\frac{1}{k n}
\left(\frac{1}{v_{1,k}(D)}\right)^{5/4}
\left(
\frac{\xi}{z_{1-\beta/2} }
\right)^{1/2}~,\label{eq: reproduce BE bound app}
\end{align}
holds with probability greater than $1-\delta$ as $D$ varies, where in writing \eqref{eq: reproduce BE bound app}, we have omitted a multiplicative term that converges to zero logarithmically as $\xi$ decreases to zero.
\end{cor}

\begin{proof}
Observe that if $n=k \cdot b$, then the right-hand-side of the inequality \eqref{eq: reproduce BE bound app} can be written 
\begin{align}
& \left(
\frac{\xi}{z_{1-\beta/2} } \frac{\sigma^{(2,b-1)}_{\mathsf{train}} (\sigma^{(4,b-1)}_{\mathsf{train}})^{1/2}}{\delta^{3/2}(v_{1,k}\left(D\right))^{5/2}}
\left(\frac{n-b}{n}\right)^{4}
\right)^{1/2} \nonumber \\
& \quad\quad\quad
 \cdot \log^{3/4} \left( \frac{z_{1-\beta/2}}{\xi} \left(\frac{n}{n-b}\right)^3 \frac{\delta (v_{1,k}(D))^2}{ (\sigma^{(4,b-1)}_{\mathsf{train}})^{3/4}} \right)~.
\end{align}
By applying \cref{def: split stability}, the non-logarithmic factor reduces to
\begin{align}
\left(
\frac{\xi}{z_{1-\beta/2} } \frac{1}{\delta^{3/2}(v_{1,k}\left(D\right))^{5/2}}
\left(\frac{\sqrt{b-1}}{n}\right)^{4}
\right)^{1/2}
& \leq 
\left(
\frac{1}{\delta^{3/2}} \frac{\xi}{z_{1-\beta/2} } \left(\frac{1}{k n}\right)^2 \left(\frac{1}{v_{1,k}(D)}\right)^{5/2}
\right)^{1/2}
\end{align}
as required.\hfill
\end{proof}

\subsection{Constructing a Stein Representer\label{sec: construction}} Suppose that we are interested in studying the statistic $f\left(X\right)$, where $X$ is random variable valued on the separable metric space $\mathcal{X}$. The method of exchangeable pairs has two ingredients. First, we need to construct a random variable $X^{\prime}$ such that $\left(X,X^{\prime}\right)$ is an exchangeable pair. Second, we need to construct an antisymmetric function $F\left(X,X^{\prime}\right)$ such that 
\begin{equation}
f\left(X\right)=\mathbb{E}\left[F\left(X,X^{\prime}\right)\mid X\right]\label{eq: stein representer}
\end{equation}
almost surely. We will refer to the function $F\left(X,X^{\prime}\right)$ as a ``Stein representer'' for $f\left(X\right)$.

In many applications, the Stein representer \eqref{eq: stein representer} induced by a suitable exchangeable pair $\left(X,X^{\prime}\right)$ can be derived in closed form. See e.g \cite{chen2011normal} and \cite{ross2011fundamentals}. A closed form derivation is more challenging in our setting, where the collection $\mathsf{R}_{g,k}$ takes the place of the random variable $X$. To address this issue, we apply a method due to \cite{chatterjee2005concentration} for constructing Stein representers through a pair of coupled Markov chains induced by an appropriately chosen exchangeable pair. 

Chatterjee's construction is founded on a observation, due to \cite{stein1986approximate}, that an exchangeable pair $\left(X,X^{\prime}\right)$ induces a reversible Markov kernel $K$ through
\[
K g(x) = \mathbb{E}\left[ g(X^\prime) \mid X = x\right]~,
\]
where $g$ is any function satisfying $\mathbb{E}\left[\vert g(X)\vert\right]<\infty$. Suppose that $\{X_m\}_{m\geq0}$ and $\{X^\prime_m\}_{m\geq0}$ are two Markov chains constructed with the kernel induced by $\left(X,X^{\prime}\right)$ and coupled in such a way that the marginal distributions of $X_m$ and $X^\prime_m$ depend only on the initial conditions $X_0$ and $X^\prime_0$, respectively. Chatterjee makes the following observation.  If there exists a constant $C$ such that
\begin{equation} \label{eq: general finiteness}
\sum_{m=0}^\infty \vert \mathbb{E}\left[ f(X_m) - f(X_m^\prime) \mid X_0 = x, X_0^\prime = y \right] \vert \leq C~,
\end{equation}
for each $x$ and $y$ in $\mathcal{X}$, then the function 
\[
F(x,y) = \sum_{m=0}^\infty \mathbb{E}\left[ f(X_m) - f(X_m^\prime) \mid X_0 = x, X_0^\prime = y \right] 
\]
is a Stein representer for $f\left(X\right)$. See \cite{paulin2013deriving,paulin2016efron} for applications of this idea to the derivation of matrix concentration inequalities.
To apply this construction to our setting, we have two tasks. First, we need to specify a suitable exchangeable pair. Second, we need to construct a pair of coupled Markov chains induced by this pair that satisfy the finiteness condition \eqref{eq: general finiteness}. Throughout, for a vector $x=(x_i)_{i=1}^n$ we let $(x_{-l},y)$ denote the vector formed by replacing the $l$th component of $x$ with $y$.

Our construction is premised
on the observation that the random collection of splits $\mathsf{R}_{g,k}$
can be generated by a random collection of permutations. To see this,
let $\mathcal{P}_{n}$ denote the set of permutations of the set $\left[n\right]$,
treating each $\pi\in\mathcal{P}_{n}$ as a bijection from $\left[n\right]$
to $\left[n\right]$. Observe that each permutation $\pi\in\mathcal{P}_{n}$
can be associated with an element of $\mathcal{R}_{n,k,b}$,
denoted by $\mathsf{\bm{r}}_k\left(\pi\right)=\left(\mathsf{s}_{1}\left(\pi\right),\ldots,\mathsf{s}_{k}\left(\pi\right)\right)$,
through 
\[
\mathsf{s}_{i}\left(\pi\right)=\left\{ \pi\left(k\cdot\left(i-1\right)+1\right),\ldots,\pi\left(k\cdot\left(i-1\right)+b\right)\right\} .
\]
If $\bm{\pi}=\left(\pi_{i}\right)_{i=1}^{g}$denotes a collection
of permutations drawn independently and uniformly at random from $\mathcal{P}_{n}$,
then the collection $\mathsf{R}_{g,k}\left(\bm{\pi}\right)=\left(\mathsf{r}_k\left(\pi_{i}\right)\right)_{i=1}^{g}$
is equidistributed with the collection $\mathsf{R}_{g,k}$ defined in \cref{sec: Reproducible}. 

Now, we construct an exchangeable pair $(\bm{\pi},\bm{\pi}^\prime)$, keeping in mind that our aim is to verify a condition of the form \eqref{eq: general finiteness}. For any permutation $\pi\in\mathcal{P}_{n}$ and indices $i,j\in\left[n\right]$,
define the updated permutation 
\[
\hat{\pi}\left(i,j\right)\left(x\right)=\begin{cases}
j, & x=i,\\
\pi\left(i\right), & x=\pi^{-1}\left(j\right),\\
\pi\left(x\right), & \text{otherwise.}
\end{cases}
\]
In other words, $\hat{\pi}\left(i,j\right)$ is identical to $\pi$,
except that $i$ maps to $j$ and $\pi^{-1}\left(j\right)$ maps to
$\pi\left(i\right)$. Let $L$ be distributed uniformly on $\left[g\right]$ and let $I$ and $J$
be independently and uniformly distributed on $\left[n\right]$. Define
the modified collection
\begin{flalign*}
\bm{\pi}^{\prime} & =(\pi_{-L},\hat{\pi}_{L}(I,J)).
\end{flalign*}
and observe that $(\bm{\pi},\bm{\pi}^\prime)$ is an exchangeable pair. 

With the exchangeable pair $(\bm{\pi},\bm{\pi}^\prime)$ in place, we choose a coupled
pair of Markov chains that it induces. For each $m\geq1$, let $L_{m}$ be distributed uniformly on $\left[g\right]$
and let $I_{m}$ and $J_{m}$ be distributed uniformly on $\left[n\right]$.
Construct $\left(\bm{\pi}_{m},\bm{\pi}_{m}^{\prime}\right)$ from
the pair $\left(\bm{\pi}_0,\bm{\pi}^{\prime}_0\right)= \left(\bm{\pi},\bm{\pi}^{\prime}\right)$ by setting 
\begin{flalign}
\bm{\pi}_{m} & =(\pi_{m-1,-L_{m}},\hat{\pi}_{m-1,L_{m}}(I_{m},J_{m}))
\quad\text{and}\quad
\bm{\pi}_{m}^{\prime}=(\pi_{m-1,-L_{m}}^{\prime},\hat{\pi}_{m-1,L_{m}}^{\prime}(I_{m},J_{m}))\label{eq: kernel coupling state-1}
\end{flalign}
recursively. In other words, the Markov chain $\bm{\pi}_0^\prime$ is initialized by choosing one permutation in $\bm{\pi}_0$ and swapping one pair of indices. In the $m$th iteration of the Markov chain $(\bm{\pi}_m,\bm{\pi}_m^\prime)_{m\geq0}$, the permutations $\pi_{m-1,L_m}$ and  $\pi^\prime_{m-1,L_m}$ are selected and updated so that $I_m$ now maps to $J_m$. As the iterations proceed, $I_m$ will continue to map to the same index in the $L_m$th permutation in both collections. That is, in each iteration that a new permutation $L$ and index $I$ are selected, the two collections become more similar. Eventually, every $L$ and $I$ will have been selected, the two collections $\bm{\pi}_m$ and $\bm{\pi}_m^\prime$ will be the same, and so $a(\mathsf{R}_{g,k}(\bm{\pi}_m), D) - a(\mathsf{R}_{g,k}(\bm{\pi}_m^\prime), D) = 0$. This will imply the finiteness \eqref{eq: general finiteness}. This argument is formalized in the proof of the following Lemma. 
\begin{lemma}
\label{lem: representer construction applied}The function 
\[
A\left(\bm{\pi},\bm{\pi}^\prime\mid D\right)
=\sum_{m=0}^{\infty}
\mathbb{E}\left[
\left(a\left(\mathsf{R}_{g,k}(\bm{\pi}_m), D\right) -a\left(\mathsf{R}_{g,k}(\bm{\pi}_m^\prime), D\right) \right)
\mid \bm{\pi}_0=\bm{\pi},\bm{\pi}_0^{\prime}=\bm{\pi}^\prime,D\right],
\]
is finite, antisymmetric, and satisfies the equality 
\begin{equation}
\mathbb{E}\left[
A\left(\bm{\pi},\bm{\pi}^\prime\mid D\right)
\mid \bm{\pi}, D\right]
= a\left(\mathsf{R}_{g,k}(\bm{\pi}), D\right) 
- \bar{a}\left(D\right)\label{eq: conditional  representer def}
\end{equation}
almost surely. 
\end{lemma}

\subsection{Proofs\label{sec: general main proofs}}

\subsubsection{Proofs for Non-Sequential Results\label{sec: non-sequential appendix}}

To obtain the concentration inequality \cref{thm: cross split concentration appendix} and Burkholder-Davis-Gundy inequality \cref{cor: Moment bound}, we apply the following result due to \citet{chatterjee2005concentration,chatterjee2007stein}, which we have augmented to be applicable under a set of assumptions considered in \citet{paulin2016efron}. 
\begin{lemma}
\label{thm: chatterjee concentration}Let $\mathcal{X}$ be a separable
metric space and suppose that $\left(X,X^{\prime}\right)$ is an exchangeable
pair of $\mathcal{X}$-valued random variables. Let $f:\mathcal{X}\to\mathbb{R}$ be a square integrable function and let $F:\mathcal{X}\times\mathcal{X}\to\mathbb{R}$ 
be a Stein representer for $f$. For each positive integer $(r)$, define the quantities 
\begin{equation}
U_{f}^{(r)}\left(X\right)=\frac{1}{2}\mathbb{E}\left[\left(f\left(X\right)-f\left(X^{\prime}\right)\right)^{2r}\mid X\right]
\quad\text{and}\quad 
U_{F}^{(r)}\left(X\right)=\frac{1}{2}\mathbb{E}\left[F\left(X,X^{\prime}\right)^{2r}\mid X\right].\label{eq: conditional variance def}
\end{equation}
If there exist nonnegative constants $u$ and $s$ such that
\begin{equation}
U^{(1)}_{f}\left(X\right)\leq s^{-1}u\quad\text{and}\quad U^{(1)}_{F}\left(X\right)\leq su,\label{eq: conditional variance bound}
\end{equation}
then the concentration inequality
\begin{equation} \label{eq: chatterjee concentration display}
P\left\{ \vert f\left(X\right)\vert\geq\delta\right\} \leq2\exp\left(-t^{2}/2u\right)
\end{equation}
holds for all $t\geq0$. Moreover, the moment inequality
\begin{equation} \label{eq: chatterjee moment display}
\mathbb{E}\left[f\left(X\right)^{2r}\right]\leq\left(2r-1\right)^{r}\left(s\mathbb{E}\left[U_{F}^{(r)}(X)\right]+s^{-1}\mathbb{E}\left[U_{f}^{(r)}(X)\right]\right)
\end{equation}
holds for all positive integers $r$ for any positive constant $s$.
\end{lemma}
\noindent
Throughout, to ease notation, we use the short hand 
\[
Z = (\mathsf{R}_{g,k}(\bm{\pi}), D)
\quad\text{and}\quad
 Z^\prime = (\mathsf{R}_{g,k}(\bm{\pi}^\prime), D)~.
\] 
The Markov chains $(Z_m)_{m\geq0}$ and $(Z^\prime_m)_{m\geq0}$ are defined accordingly and we simplify $A\left(\bm{\pi},\bm{\pi}^\prime\mid D\right)$ to $A(Z,Z^\prime)$. To apply \cref{thm: chatterjee concentration}, we are required to develop bounds for the objects
\begin{flalign*}
U_{a}^{(r)}\left(Z\right) & =\frac{1}{2}\mathbb{E}\left[\left(a\left(Z\right)-a\left(Z^{\prime}\right)\right)^{2r}\mid Z\right]
\quad\text{and}\quad 
U_{A}^{(r)}\left(Z\right)=\frac{1}{2}\mathbb{E}\left[A\left(Z,Z^{\prime}\right)^{2r}\mid Z\right].
\end{flalign*}
Our approach is based on the following Lemma, which combines
a generalization of an idea due to Lemma 10.4 of \citet{paulin2016efron}
with a Markov type bound.
\begin{lemma}
\label{lem: variance bounds}Let $r=2^{c}$ for some positive integer
$c$. Let $f:\mathcal{Z}\to\mathbb{R}$ be a function such that there
exists a square integrable random variable $W$ with
\begin{equation}
\left(\sum_{m=0}^{\infty}\mathbb{E}\left[f\left(Z_{m}\right)-f\left(Z_{m}^{\prime}\right)\mid Z_{0}=Z,Z_{0}^{\prime}=Z^{\prime}\right]\right)^{r}\leq W\label{eq: abstract integrability}
\end{equation}
almost surely. If the inequality
\begin{equation}
\mathbb{E}\left[\mathbb{E}\left[f(Z_{m})-f(Z_{m}^{\prime})\mid Z_{0}=Z,Z_{0}^{\prime}=Z^{\prime}\right]^{r}\mid\bm{\pi}\right]\leq h_{m}^{r}\label{eq: c_i general bound}
\end{equation}
holds for each $m\geq0$ and each collection $\bm{\pi}$, where $\left(h_{m}\right)_{m\geq0}$
is a deterministic sequence of nonnegative numbers, then the inequalities
\begin{flalign}
\frac{1}{2}\mathbb{E}\left[\left(f\left(Z\right)-f\left(Z^{\prime}\right)\right)^{r}\mid Z\right] & \leq\frac{h_{0}^{r}}{\delta}\quad\text{and}\label{eq: cond variance bound}\\
\frac{1}{2}\mathbb{E}\left[\left(\sum_{m=0}^{\infty}\mathbb{E}\left[f\left(Z_{m}\right)-f\left(Z_{m}^{\prime}\right)\mid Z_{0}=Z,Z_{0}^{\prime}=Z^{\prime}\right]\right)^{r}\mid Z\right] & \leq\frac{1}{\delta}\left(\sum_{m=0}^{\infty}h_{m}\right)^{r}
\end{flalign}
both hold with probability greater than $1-\delta$ as $D$ varies
\end{lemma}
\noindent To apply this Lemma, we begin by noting that the inequality
\begin{flalign}
& \left(
\sum_{m=0}^{\infty}\mathbb{E}\left[a\left(Z_{m}\right)-a\left(Z_{m}^{\prime}\right)\mid Z_{0}=Z,Z_{0}^{\prime}=Z^{\prime}\right]
\right)^2  \nonumber \\
 & \lesssim g^2n^{4}\max_{\mathsf{s},\mathsf{s}^{\prime}\in\mathsf{S}_{n,b}}\left(T\left(\mathsf{s},U,D\right)-T\left(\mathsf{s}^{\prime},U,D\right)\right)^2
 \label{eq: bound for sq representer}
\end{flalign}
follows from \cref{lem: finiteness}, stated in \cref{app: proof of representer construction applied}. Moreover, the right hand
side of \eqref{eq: bound for sq representer} is square integrable,
as the fourth order split-stability $\zeta^{(4)}$ is finite by assumption. Deterministic
bounds of the form \eqref{eq: c_i general bound} are obtained through
the following Lemma.
\begin{lemma}
\label{lem: deterministic bound}Under \cref{assu: invariance,assu: linearity},
for all integers $m\geq0$ and $r\geq1$, the inequality 
\begin{flalign*}
 & \mathbb{E}\left[\mathbb{E}\left[a(Z_{m})-a(Z_{m}^{\prime})\mid Z_{0}=Z,Z_{0}^{\prime}=Z^{\prime}\right]^{2r}\mid\bm{\pi}\right] \\ 
 & \leq 2^{4r}  \left(1-\frac{2}{gn^{2}}\right)^{2mr}
\left(\frac{2n-bk-b}{gn^{2}}\right)^{2r} \Gamma_{k,\varphi,b}^{(r)}
\end{flalign*}
holds almost surely, where $\Gamma_{k,\varphi,b}^{(r)}$ is defined in the statement of \cref{cor: Moment bound}.
\end{lemma}
\noindent
Combining \cref{lem: variance bounds,lem: deterministic bound}, we have that  
\begin{flalign*}
U_{A}^{(r)}(Z) & \leq \frac{1}{\delta}\left(\frac{gn^{2}}{2}\right)^{r}\left(\frac{2^4\left(2n-bk-b\right)^{2}}{gn^{2}} \right)^r \Gamma_{k,\varphi,b}^{(r)}
\end{flalign*}
and 
\begin{flalign*}
U_{a}^{(r)}(Z) & \leq \frac{1}{\delta}\left(\frac{2}{gn^{2}}\right)^{r}\left(\frac{2^4\left(2n-bk-b\right)^{2}}{gn^{2}} \right)^r \Gamma_{k,\varphi,b}^{(r)}
\end{flalign*}
with probability $1-\delta$. \cref{thm: cross split concentration} is obtained by setting $r=1$ and applying \eqref{eq: chatterjee concentration display} of \cref{thm: chatterjee concentration} with $s=gn^2/2$ and
\[
u=\frac{1}{\delta}\frac{2^4\left(2n-bk-b\right)^{2}}{gn^{2}} \Gamma_{n,k,b}^{(1)}.
\]
Similarly, \cref{cor: Moment bound} is obtained by applying  \eqref{eq: chatterjee moment display} of \cref{thm: chatterjee concentration} with $s=(gn^2/2)^r$.

The normal approximation error bound \cref{thm: normal approximation appendix} follows from \cref{cor: Moment bound}. To see this, consider the centered statistic 
\begin{equation}
a\left(Z\right)-\bar{a}\left(D\right)=\frac{1}{g}\sum_{l=1}^{g}\bar{a}\left(\mathsf{r}_{l},D\right),
\quad\text{where}\quad
\bar{a}\left(\mathsf{r}_{l},D\right)=\frac{1}{k}\sum_{i=1}^{k}\left(T(\mathsf{s}_{l,i},D)-\bar{a}\left(D\right)\right)\label{eq: centered}
\end{equation}
for each $l$ in $\left[g\right]$. Conditional on the data $D$, the statistics $\bar{a}\left(\mathsf{r}_{l},D\right)$ are independent, identically distributed, and mean zero. Moreover, we have that 
\begin{flalign}
\mathbb{E}\left[\vert\bar{a}\left(\mathsf{r}_{l},D\right)\vert^{3}\right] &
\leq
3^2 2^6 \left(2-\varphi k-\varphi\right)^{3}(\Gamma_{k,\varphi,b}^{(2)})^{3/4}
\label{eq: apply fourth moment bound}
\end{flalign}
by H\"older's inequality and \cref{cor: Moment bound}. Hence, we find that 
\begin{flalign*}
\frac{\mathbb{E}\left[\vert\bar{a}\left(\mathsf{r}_{l},D\right)\vert^{3}\mid D \right]}
{g^{1/2}\left(v_{1,k}(D)\right)^{3/2}} & 
\leq
\frac{3^2 2^6}{\delta}\frac{\left(2-\varphi k-\varphi \right)^{3}}{\left(v_{1,k}(D)\right)^{3/2}} \frac{(\Gamma_{k,\varphi,b}^{(2)})^{3/4}}{g^{1/2}}
\end{flalign*}
holds with probability greater than $1-\delta,$ by combining \eqref{eq: apply fourth moment bound} with the Markov inequality. The proof then follows by the Berry-Esseen inequality. See e.g., Corollary 1 of \cite{shevtsova2011absolute}.\hfill\qed

\subsubsection{Proofs for Sequential Results\label{sec: sequential appendix}}The first step for verifying \cref{thm: stable symmetric reproduce appendix} is deriving a large deviation bound for the variance estimator $\hat{v}\left(\mathsf{R}_{g,k}, D\right)$ defined in \eqref{eq: var estimator}. This is obtained in the following Lemma, which follows from an argument very similar to the proof of \cref{thm: cross split concentration appendix}.
\begin{lemma}
\label{thm: variance estimator concentration}Suppose that \cref{assu: invariance,assu: linearity} hold and that the data $D$ are independent and identically distributed. If the eighth-order split-stability
$\zeta^{(8)}$ is finite, then the conditional concentration inequality
\begin{flalign*}
 & \log \frac{1}{4}P\left\{ 
 \bigg\vert \frac{\hat{v}(\mathsf{R}_{g,k}, D)}{v_{g,k}\left(D\right)} - 1\bigg\vert
 \geq t\mid D\right\}
 \lesssim
  - \frac{\delta \left(v_{1,k}(D)\right)^2}{\left(2- \varphi k-\varphi\right)^{4}} \frac{g t^{2}}{\Gamma_{k,\varphi,b}^{(2)}}
\end{flalign*}
holds for all $t>0$ with probability greater than $1-\delta$ as $D$ varies.
\end{lemma}
\noindent We focus our analysis on the error
\begin{equation}
\vert P\left\{ a(\mathsf{R}_{\hat{g},k},D)-a(\mathsf{R}_{\hat{g}^{\prime},k}^{\prime},D)\leq\xi\mid D\right\} -\left(1-\beta/2\right)\vert.\label{eq: rep upper tail}
\end{equation}
An analogous argument will yield the same bound for the lower tail.
We begin by bounding \eqref{eq: rep upper tail} with quantities that
will be easier to handle in isolation. Define the objects
\begin{flalign}
U\left(\mathsf{R}_{g,k},D\right) & =\frac{1}{g^\star}\left(\sum_{i=1}^{g}\bar{a}\left(\mathsf{r}_{i,k},D\right)-\sum_{i=1}^{g^{\star}}\bar{a}\left(\mathsf{r}_{i,k},D\right)\right)\quad\text{and}\label{eq: maximal term}\\
Q\left(\mathsf{R}_{g,k},D\right) & =\left(1-\frac{g}{g^{\star}}\right)\left(a(\mathsf{R}_{g,k},D)-\bar{a}\left(D\right)\right)~.\label{eq: cross product}
\end{flalign}
The following Lemma bounds the error \eqref{eq: rep upper tail} in
terms of the error in the normal approximation to the quantity $a(\mathsf{R}_{g^\star,k},D)-a(\mathsf{R}_{g^\star,k}^{\prime},D)$ and generic high
probability bounds on \eqref{eq: maximal term} and \eqref{eq: cross product}. 
\begin{lemma}
\label{lem: reproducability decomposition} Define the events
\begin{flalign*}
\mathcal{U}_{k,\lambda}\left(D\right) & =\left\{ \vert U(\mathsf{R}_{\hat{g},k},D)-U(\mathsf{R}_{\hat{g}^{\prime},k}^{\prime},D)\vert\leq\lambda\sqrt{2v_{g^{\star},k}\left(D\right)}\right\} \quad\text{and}\\
\mathcal{Q}_{k,\lambda}(D) & =\left\{ \Big\vert Q(\mathsf{R}_{g,k},D)-Q(\mathsf{R}_{\hat{g}^{\prime},k}^{\prime},D)\Big\vert\leq\lambda\sqrt{2v_{g^{\star},k}\left(D\right)}\right\} .
\end{flalign*}
The quantity \eqref{eq: rep upper tail} is bounded above by 
\begin{align}
\sup_{z \in \mathbb{R}}
&\bigg\vert 
P\left\{ \frac{a(\mathsf{R}_{g^{\star},k},D)-a(\mathsf{R}_{g^{\star},k}^{\prime},D)}{\sqrt{2v_{g^{\star},k}\left(D\right)}} \leq z \mid D\right\}
-
\Phi(z)  \bigg\vert \nonumber\\
& \quad\quad\quad\quad
+2\lambda+\left(1-P\left\{ \mathcal{U}_{k,\lambda}\left(D\right)\cap\mathcal{Q}_{k,\lambda}(D)\mid D\right\} \right)~,\label{eq: bound with three terms}
\end{align}
where $\Phi(\cdot)$ is the standard normal c.d.f.
\end{lemma}

Thus, it remains to give suitable bounds for the objects in \eqref{eq: bound with three terms}. We state these bounds in terms of the quantities
\begin{align}
\rho_{k,\varphi,b}(\xi,\beta \mid D)  &=
\frac{\xi}{z_{1-\beta/2}} \frac{(2 - \varphi k - \varphi)^3 (\Gamma^{(2)}_{k,\varphi,b})^{3/4}}{\delta \left(v_{1,k}\left(D\right)\right)^{2}}\quad\text{and}\label{eq: rho term}\\
\lambda_{k,\varphi,b}(\xi,\beta \mid D)  &= \left(
\frac{\xi}{z_{1-\beta/2}}
\frac{(2-\varphi k-\varphi)^{4}\Gamma^{(1)}_{k,\varphi,b} (\Gamma_{k,\varphi,b}^{(2)})^{1/2}}{\delta^{3/2}(v_{1,k}\left(D\right))^{5/2}}
\right)^{1/2}~,\label{eq: lambda term}
\end{align}
respectively. Each part of following Lemma follows from an application of \cref{lem: reproducability decomposition}. 

\begin{lemma}
\label{lem: three components}
Suppose that \cref{assu: invariance,assu: linearity} hold, that the data $D$ are independent and identically distributed, that the conditional variance $v_{1,k}(D) = \Var(a(\mathsf{r}, D) \mid D)$ is strictly positive almost surely, and that the eighth-order split stability $\zeta^{(8)}$ is finite.

\noindent \textbf{(i)} The inequality
\begin{flalign*}
\bigg\vert 
P\left\{ \frac{a(\mathsf{R}_{g^{\star},k},D)-a(\mathsf{R}_{g^{\star},k}^{\prime},D)}{\sqrt{2v_{g^{\star},k}\left(D\right)}} \leq z \mid D\right\}
-
\Phi(z)  \bigg\vert
 & \lesssim\rho_{k,\varphi,b}\left(\xi,\beta\mid D\right)
\end{flalign*}
is satisfied with probability greater than $1-\delta$ as $D$ varies. 

\noindent \textbf{(ii)} The conditional concentration inequality 
\begin{flalign*}
 & P\left\{ \frac{\vert U(\mathsf{R}_{\hat{g},k},D)-U(\mathsf{R}_{\hat{g}^{\prime},k}^{\prime},D)\vert}{\sqrt{2v_{g^{\star},k}\left(D\right)}}\geq
 \lambda_{k,\varphi,b}\left(\xi,\beta\mid D\right) \log^{3/4}\left(\rho_{k,\varphi,b}\left(\xi,\beta\mid D\right)^{-1}\right) \mid D\right\} \\
 & \quad\quad\quad\quad \lesssim\rho_{k,\varphi,b}\left(\xi,\beta\mid D\right)
\end{flalign*}
holds with probability greater than $1-\delta$ as $D$ varies.

\noindent \textbf{(iii)} For all sufficiently small $\xi$, the conditional concentration inequality 
\[
P\left\{ \frac{\vert Q(\mathsf{R}_{\hat{g},k},D)-Q(\mathsf{R}_{\hat{g}^{\prime},k}^{\prime},D)\vert}{\sqrt{2v_{g^{\star},k}\left(D\right)}}\geq\lambda_{k,\varphi,b}\left(\xi,\beta\mid D\right)\right\} \lesssim\rho_{k,\varphi,b}\left(\xi,\beta\mid D\right)
\]
holds with probability greater than $1-\delta$ as $D$ varies.
\end{lemma}

Now, observe that $\rho_{k,\varphi,b}\left(\xi,\beta\mid D\right)$ is smaller than $\lambda_{k,\varphi,b}(\xi,\beta \mid D)$ for all sufficiently small $\xi$, almost surely. Consequently, by combining \cref{lem: reproducability decomposition} and \cref{lem: three components}, we find that
\begin{align}
&\vert P\left\{ \vert a(\mathsf{R}_{\hat{g},k},D)-a(\mathsf{R}_{\hat{g}^{\prime},k}^{\prime},D) \vert \leq\xi\mid D\right\} -\left(1-\beta\right)\vert  \nonumber\\
& \quad\quad\quad\quad
\lesssim
 \left(
\frac{\xi}{z_{1-\beta/2}}
\frac{(2-\varphi k-\varphi)^{4}\Gamma^{(1)}_{k,\varphi,b} (\Gamma_{k,\varphi,b}^{(2)})^{1/2}}{\delta^{3/2}(v_{1,k}\left(D\right))^{5/2}}
\right)^{1/2} \nonumber \\
& \quad\quad\quad\quad \quad\quad \cdot \log^{3/4} \left( \frac{z_{1-\beta/2}}{\xi} \frac{\delta (v_{1,k}(D))^2}{(2-\varphi k - \varphi)^3 (\Gamma^{(2)}_{k,\varphi,b})^{3/4}} \right)
\end{align}
with probability greater than $1-\delta$, as required. \hfill\qed

\subsection{Comparison with \cite{zhang2022berry}\label{app: zheng comparion}}
We state a Berry-Esseen bound for $a\left(\mathsf{R}_{g,k}, D\right)$ through an application of a result due to \citet{zhang2022berry}. In contrast to the bound stated in \cref{thm: normal approximation appendix}, the bound obtained here does not shrink as $g$ increases. On the other hand, the bound obtained below is unconditional. It is straightforward to modify our argument to give an analogous high-probability conditional bound. 
\begin{theorem}
\label{thm: zhang normal approximation}Let $W$ denote a standard
normal random variable. Suppose that \cref{assu: invariance,assu: linearity} hold
and that the data $D$ are independent and identically distributed. If the eighth-order split stability $\zeta^{(8)}$ is finite,
then the Berry-Esseen inequality
\begin{flalign*}
d_{\text{K}}\left(\frac{a\left(\mathsf{R}_{g,k}, D\right)-\bar{a}\left(D\right)}{\sqrt{\mathbb{E}\left[v_{g,k}\left(D\right)\right]}},W\right) & 
\leq
\frac{4 \left(2-\varphi k-\varphi\right)^{2} (\Gamma_{k,\varphi,b}^{(2)})^{1/2}}{\mathbb{E}\left[v_{1,k}(D)\right]}
\end{flalign*}
is satisfied.
\end{theorem}

\subsubsection{Proof of \cref{thm: zhang normal approximation}}

We apply the following central limit theorem, due to \citet{zhang2022berry}.
This result generalizes Theorem 2.1 of \citet{shao2019berry} to accommodate
general Stein representers. 
\begin{theorem}[{\citealp[Theorem 4.1, ][]{zhang2022berry}}]
\label{thm: kol approx} Let $\mathcal{X}$ be a separable metric
space and suppose that $\left(X,X^{\prime}\right)$ is an exchangeable
pair of $\mathcal{X}$-valued random variables. Suppose that $f:\mathcal{X}\to\mathbb{R}$
and $F:\mathcal{X}\times\mathcal{X}\to\mathbb{R}$ are square-integrable
functions such that $F$ is antisymmetric and 
\[
\mathbb{E}\left[F\left(X,X^{\prime}\right)\mid X\right]=f\left(X\right)
\]
almost surely. Assume that $\Var\left(f\left(X\right)\right)$ is
finite and non-zero and that $\mathbb{E}\left[f\left(X\right)\right]=0$.
Define the objects
\begin{flalign}
\bar{f}\left(X\right) & =f\left(X\right)/\sqrt{\Var\left(f\left(X\right)\right)}.\label{eq: T def}\\
G\left(X\right) & =\frac{1}{2}\mathbb{E}\left[\left(f\left(X\right)-f\left(X^{\prime}\right)\right)F\left(X,X^{\prime}\right)\mid X\right]\quad\text{and}\label{eq: delta def}\\
\bar{G}\left(X\right) & =\frac{1}{2}\mathbb{E}\left[\left(f\left(X\right)-f\left(X^{\prime}\right)\right)\vert F\left(X,X^{\prime}\right)\vert\mid X\right].\label{eq: delta bar def}
\end{flalign}
Let $W$ denote a standard normal random variable. The bound
\begin{flalign}
d_{K}\left(\bar{f}\left(X\right),W\right) & \leq\frac{\mathbb{E}\left[\vert G\left(X\right)-\mathbb{\mathbb{E}}\left[G\left(X\right)\right]\vert\right]+\mathbb{E}\left[\vert\bar{G}\left(X\right)\vert\right]}{\Var\left(f\left(X\right)\right)}\label{eq: KS bound}
\end{flalign}
is satisfied.
\end{theorem}
\noindent
To this end, define the objects
\begin{flalign*}
B\left(Z\right) & =\frac{1}{2}\mathbb{E}\left[\left(a\left(Z\right)-a\left(Z^{\prime}\right)\right)A\left(Z,Z^{\prime}\right)\mid Z\right]\quad\text{and}\\
\bar{B}\left(Z\right) & =\frac{1}{2}\mathbb{E}\left[\left(a\left(Z\right)-a\left(Z^{\prime}\right)\right)\vert A\left(Z,Z^{\prime}\right)\vert\mid Z\right].
\end{flalign*}
It will suffice to bound the
quantity 
\begin{equation}
\frac{\mathbb{E}\left[\vert B\left(Z\right)-\mathbb{\mathbb{E}}\left[B\left(Z\right)\right]\vert\right]+\mathbb{E}\left[\vert\bar{B}\left(Z\right)\vert\right]}{\Var\left(a\left(Z\right)-\bar{a}\left(Z\right)\right)}.\label{eq: BE rhs}
\end{equation}
Observe that
\begin{flalign*}
\mathbb{E}\left[\vert B\left(X\right)-\mathbb{\mathbb{E}}\left[B\left(X\right)\right]\vert\right] & \leq\sqrt{\Var\left(B\left(Z\right)\right)}\quad\text{and}\\
\mathbb{E}\left[\vert\bar{B}\left(Z\right)\vert\right] & \leq\sqrt{\Var\left(\bar{B}\left(Z\right)\right)}
\end{flalign*}
by the Cauchy-Schwarz inequality and the fact that $\mathbb{E}\left[\bar{B}\left(X\right)\right]=0$ by exchangeability. Consequently, as
\begin{flalign*}
\Var\left(B\left(Z\right)\right) & \leq\mathbb{E}\left[B\left(Z\right)^{2}\right]=\mathbb{E}\left[\left(a\left(Z\right)-a\left(Z^{\prime}\right)\right)^{2}A\left(Z,Z^{\prime}\right)^{2}\right]=\Var\left(\bar{B}\left(Z\right)\right)
\end{flalign*}
it will suffice to bound
\begin{equation}
2\mathbb{E}\left[\left(a\left(Z\right)-a\left(Z^{\prime}\right)\right)^{2}A\left(Z,Z^{\prime}\right)^{2}\right].\label{eq: simplified be rhs}
\end{equation}
By Young's inequality, we have that 
\begin{flalign*}
\mathbb{E}\left[\left(a\left(Z\right)-a\left(Z^{\prime}\right)\right)^{2}A\left(Z,Z^{\prime}\right)^{2}\right] 
& \leq\frac{1}{2}\left(s^{-1}\mathbb{E}\left[U_a^{(2)}(Z)\right] + s\mathbb{E}\left[U_A^{(2)}(Z)\right]\right)~,
\end{flalign*}
where $U_a^{(2)}(Z)$ and $U_A^{(2)}(Z)$ are defined in \cref{sec: non-sequential appendix}. 

Observe that the bound 
\begin{align}
& \left(\sum_{m=0}^{\infty}\mathbb{E}\left[a\left(Z_{m}\right)-a\left(Z_{m}^{\prime}\right)\mid Z_{0}=Z,Z_{0}^{\prime}=Z^{\prime}\right]\right)^{4} \nonumber \\
& \quad\quad \leq
\left(2gn^{2}\max_{\mathsf{s},\mathsf{s}^{\prime}\in\mathsf{S}_{n,b}}\left(T\left(\mathsf{s},D\right)-T\left(\mathsf{s}^{\prime},D\right)\right)\right)^{4}\label{eq: be general integrability}
\end{align}
holds by \cref{lem: finiteness} and that the right-hand side
of (\ref{eq: be general integrability}) is square-integrable as the
eighth-order split stability $\zeta^{(8)}$ is finite. Thus, by combining
\cref{lem: sequence bound} and \cref{lem: deterministic bound},
we have that 
\[
\bar{U}_{A}^{(2)}\leq\
\left(\frac{gn^{2}}{2}\right)^{2}\left(\frac{4\left(2n-bk-b\right)^{2}}{gn^{2}} \right)^2 \Gamma_{k,\varphi,b}^{(2)}
\]
and 
\[
\bar{U}_{a}^{(2)}\leq
\left(\frac{2}{gn^{2}}\right)^{2}\left(\frac{4\left(2n-bk-b\right)^{2}}{gn^{2}} \right)^2 \Gamma_{k,\varphi,b}^{(2)} ~.
\]
Hence, by taking $s=\left(gn^{2}/2\right)^{2}$, we find that \eqref{eq: simplified be rhs} is bounded above by
\[
\frac{4}{gn^{2}}\left(\left(2n-bk-b\right)^{2} (\Gamma_{k,\varphi,b}^{(2)})^{1/2}\right)
\]
Now, we have that 
\[
\text{Var}\left(a\left(Z\right)-\bar{a}\left(D\right)\right)=\mathbb{E}\left[v_{g,k}\left(D\right)\right]
\]
by the law of total variance and the fact that $\mathbb{E}\left[a\left(Z\right)-\bar{a}\left(D\right)\mid D\right]=0$.
Consequently we can decompose
\begin{flalign}
\mathbb{E}\left[v_{g,k}\left(D\right)\right] & = \frac{\mathbb{E}[\phi_{n,b}(D)] + \left(k-1\right)\mathbb{E}[\gamma_{n,b}(D)]}{kg}.\label{eq: covariance decomposition-1}
\end{flalign}
Hence, we have that \eqref{eq: BE rhs} is bounded above by
\[
\frac{4 \left(2-\varphi k-\varphi\right)^{2} (\Gamma_{k,\varphi,b}^{(2)})^{1/2}}{\mathbb{E}\left[v_{1,k}(D)\right]},
\]
as required.\hfill\qed

\section{Proofs for Lemmas Stated in \cref{sec: general results}\label{app: proof of main text}}

\subsection{Proof of \cref{lem: representer construction applied}\label{app: proof of representer construction applied}}

We use the following Lemma in several places.
\begin{lemma}
\label{lem: finiteness}Let $\bm{\psi}$ and $\bm{\psi}^{\prime}$
be two sets, each containing $g$ elements of $\mathcal{P}_{n}$. The inequality
\begin{flalign*}
 & \sum_{m=0}^{\infty}\bigg\vert
 \mathbb{E}\left[
 a\left(\mathsf{R}_{g,k}(\bm{\pi}_m)\right) - 
 a\left(\mathsf{R}_{g,k}(\bm{\pi}^\prime_m)\right)
 \mid 
 \bm{\pi}_0=\bm{\psi},
 \bm{\pi}^\prime_0=\bm{\psi}^\prime, D\right]\bigg\vert\\
 & \leq 2gn^{2}
 \max_{\mathsf{s},\mathsf{s}^{\prime}\in\mathsf{S}_{n,b}}\left(T\left(\mathsf{s},D\right)-T\left(\mathsf{s}^{\prime},D\right)\right)
\end{flalign*}
holds almost surely. 
\end{lemma}
\noindent Observe that the quantity
\[
\max_{\mathsf{s},\mathsf{s}^{\prime}\in\mathsf{S}_{n,b}}\left(T\left(\mathsf{s},D\right)-T\left(\mathsf{s}^{\prime},D\right)\right)
\]
is finite almost surely. Thus, the convergence of the series defining
$A\left( \bm{\pi}, \bm{\pi}^\prime \mid D\right)$ follows from Lemma \ref{lem: finiteness}.
Define the operator 
\begin{flalign*}
K:\mathcal{F} & \to\mathcal{F}\\
f\left(\cdot\right) & \mapsto\mathbb{E}\left[f\left(\bm{\pi}^\prime\right)\mid \bm{\pi}=\cdot\right],
\end{flalign*}
where $\mathcal{F}$ is the set of all measurable functions supported
on the domain of $\bm{\pi}$. Observe that 
\begin{flalign*}
 & \mathbb{E}\left[
 a\left(\mathsf{R}_{g,k}(\bm{\pi}_m), D\right) - 
 a\left(\mathsf{R}_{g,k}(\bm{\pi}^\prime_m), D\right)
 \mid 
 \bm{\pi}_0= \bm{\pi},
 \bm{\pi}^\prime_0= \bm{\pi}^\prime, D\right]\\
 & = \mathbb{E}\left[
 a\left(\mathsf{R}_{g,k}(\bm{\pi}_m), D\right) - 
\bar{a}(D)
 \mid 
 \bm{\pi}_0= \bm{\pi}, D\right]
 - 
 \mathbb{E}\left[
 a\left(\mathsf{R}_{g,k}(\bm{\pi}_m^\prime), D\right) - 
\bar{a}(D)
 \mid 
 \bm{\pi}^\prime_0= \bm{\pi}^\prime, D\right]\\
 & = 
  K^{m}\left(
 a\left(\mathsf{R}_{g,k}(\bm{\pi})\right) - 
\bar{a}(D)
\right)
 -K^{m+1}\left(
 a\left(\mathsf{R}_{g,k}(\bm{\pi}^\prime)\right) - 
\bar{a}(D)
\right).
\end{flalign*}
Thus, for any $m^\prime$, we have that 
\begin{align}
&\sum_{m=0}^{m^\prime}\mathbb{E}\left[
 a\left(\mathsf{R}_{g,k}(\bm{\pi}_m), D\right) - 
 a\left(\mathsf{R}_{g,k}(\bm{\pi}^\prime_m), D\right)
 \mid 
 \bm{\pi}_0= \bm{\pi},
 \bm{\pi}^\prime_0= \bm{\pi}^\prime, D\right] \nonumber\\
& =
a\left(\mathsf{R}_{g,k}(\bm{\pi}), D\right)-\bar{a}\left(D\right)
-K^{m^\prime+1}
\left(a\left(\mathsf{R}_{g,k}(\bm{\pi}), D\right)-\bar{a}\left(D\right)\right).\label{eq: partial sums}
\end{align}
By Lemma \ref{lem: finiteness}, the partial sums (\ref{eq: partial sums})
converge almost everywhere and so the sequence 
\begin{equation}
\left(K^{m+1}
\left(a\left(\mathsf{R}_{g,k}(\bm{\pi}), D\right)-\bar{a}\left(D\right)\right)
\right)_{m\geq0}\label{eq: K a sequence}
\end{equation}
also converges almost everywhere. Lemma \ref{lem: finiteness} also
implies that the limit of (\ref{eq: K a sequence}) depends only on
$D$, as 
\[
K^{m}\left(a\left(\left(\mathsf{R}_{g,k}\left(\bm{\psi}\right),D\right)\right)-\bar{a}\left(D\right)\right)
-K^{m}a\left(\left(\mathsf{R}_{g,k}\left(\bm{\psi}^{\prime}\right),D\right)-\bar{a}\left(D\right)\right)\to0
\]
for any $\bm{\psi}$ and $\bm{\psi}^{\prime}$ each containing $g$
elements of $\mathcal{P}_{n}$. Therefore, we
have that 
\begin{flalign*}
&\mathbb{E}\left[A\left(\bm{\pi},\bm{\pi}^{\prime}\mid D \right)\mid D\right] \\
& =\mathbb{E}\left[\lim_{n\to\infty}\left(a\left(\mathsf{R}_{g,k}(\bm{\pi}), D\right)-\bar{a}\left(D\right)
-K^{m+1}a\left(\mathsf{R}_{g,k}(\bm{\pi}), D\right)-\bar{a}\left(D\right)\right)\mid D \right]\\
 & =\mathbb{E}\left[a\left(\mathsf{R}_{g,k}(\bm{\pi}), D\right)-\bar{a}\left(D\right)\mid D\right]-b\left(D\right),
\end{flalign*}
for some quantity
\[
b\left(D\right)=\lim_{m\to\infty}K^{m}\left(a\left(\mathsf{R}_{g,k}(\bm{\pi}), D\right)-\bar{a}\left(D\right)\right)
\]
that depends only on $D$. Observe that
\begin{flalign*}
& \mathbb{E}\left[A\left(\bm{\pi},\bm{\pi}^{\prime}\mid D \right)\mid D\right]\\ 
& =\mathbb{E}\left[\lim_{m^\prime\to\infty}\sum_{m=0}^{m^\prime}
\mathbb{E}\left[a\left(\mathsf{R}_{g,k}(\bm{\pi}_m), D\right)
-a\left(\mathsf{R}_{g,k}(\bm{\pi}_m^\prime), D\right)
\mid \bm{\pi}_0=\bm{\pi},\bm{\pi}^\prime_0=\bm{\pi}^{\prime}, D\right]\mid D\right]\\
 & =\lim_{m^\prime\to\infty}\sum_{m=0}^{m^\prime}
 \mathbb{E}\left[a\left(\mathsf{R}_{g,k}(\bm{\pi}_m), D\right)
-a\left(\mathsf{R}_{g,k}(\bm{\pi}_m^\prime), D\right)
\mid D\right]\tag{Dominated Conv.}\\
& = 0,\tag{Exchangeability}
\end{flalign*}
where the applicability of the Dominated Convergence Theorem follows
from Lemma \ref{lem: finiteness}. Thus, as 
\begin{flalign*}
\mathbb{E}\left[a\left(\mathsf{R}_{g,k}(\bm{\pi}), D\right)-\bar{a}\left(D\right)\mid D\right] & =0,
\end{flalign*}
we can conclude that $b\left(D\right)=0$ almost surely. Hence, we
find that 
\begin{flalign*}
& \mathbb{E}\left[A\left(\bm{\pi},\bm{\pi}^{\prime}\mid D \right)\mid \bm{\pi}, D\right]\\ 
& =\mathbb{E}\left[\lim_{m^\prime\to\infty}\sum_{m=0}^{m^\prime}
\mathbb{E}\left[a\left(\mathsf{R}_{g,k}(\bm{\pi}_m), D\right)
-a\left(\mathsf{R}_{g,k}(\bm{\pi}_m^\prime), D\right)
\mid \bm{\pi}_0=\bm{\pi},\bm{\pi}^\prime_0=\bm{\pi}^{\prime}, D\right]\mid \bm{\pi}, D\right]\\
 & 
 =\lim_{m^\prime\to\infty}\sum_{m=0}^{m^\prime}
 \mathbb{E}\left[a\left(\mathsf{R}_{g,k}(\bm{\pi}_m), D\right)
-a\left(\mathsf{R}_{g,k}(\bm{\pi}_m^\prime), D\right)
\mid\bm{\pi}_0=\bm{\pi},\bm{\pi}^\prime_0=\bm{\pi}^{\prime}, D\right]
\tag{Dominated Conv.}\\
 & -\lim_{m^\prime\to\infty}K^{m^\prime+1}\left(a\left(\mathsf{R}_{g,k}(\bm{\pi}), D\right)-\bar{a}\left(D\right)\right)\\
 & =a\left(\mathsf{R}_{g,k}(\bm{\pi}), D\right)-\bar{a}\left(D\right),
\end{flalign*}
completing the proof.\hfill$\qed$

\subsection{Proof of \cref{thm: chatterjee concentration}\label{subsec: proof of chatterjee concentration}}
First, observe that 
\[
U_{F}\left(X\right)\geq\frac{1}{2}\left(\mathbb{E}\left[F\left(X,X^{\prime}\right)\mid X\right]\right)^{2}=\frac{1}{2}f\left(X\right)^{2},
\]
where the inequality follows from Jensen's inequality and the definition
of the Stein representer $F$. By \eqref{eq: cond variance bound},
we have that 
\[
su\geq\frac{1}{2}f\left(X\right)^{2}.
\]
Hence, the random variable $f\left(X\right)$ is bounded almost surely. 

Now, suppose $h:\mathcal{X}\to\mathbb{R}$ is any measurable function
such that $\mathbb{E}\left[h\left(X\right)F\left(X,X^{\prime}\right)\right]<\infty$.
Then, $\mathbb{E}\left[h\left(X\right)f\left(X\right)\right]=\mathbb{E}\left[h\left(X\right)F\left(X,X^{\prime}\right)\right]$.
Using the exchangeability of $X$ and $X^{\prime}$ and the fact that
$F$ is antisymmetric, we have that 
\[
\mathbb{E}\left[h\left(X\right)F\left(X,X^{\prime}\right)\right]=\mathbb{E}\left[h\left(X^{\prime}\right)F\left(X^{\prime},X\right)\right]=-\mathbb{E}\left[h\left(X^{\prime}\right)F\left(X,X^{\prime}\right)\right]
\]
and that therefore 
\begin{equation}
\mathbb{E}\left[h\left(X\right)f\left(X\right)\right]=\frac{1}{2}\mathbb{E}\left[\left(h\left(X\right)-h\left(X^{\prime}\right)\right)F\left(X,X^{\prime}\right)\right].\label{eq: kernel symmetry}
\end{equation}
Let 
\[
m\left(\theta\right)=\mathbb{E}\left[\exp\left(\theta f\left(X\right)\right)\right]
\]
denote the moment generating function of $f\left(X\right)$. As $f\left(X\right)$
is bounded almost surely, we can exchange differentiation and expectation
in the differentiation of $m\left(\theta\right)$. Thus, we obtain
\begin{flalign}
m^{\prime}\left(\theta\right) & =\mathbb{E}\left[\exp\left(\theta f\left(X\right)\right)f\left(X\right)\right].\nonumber \\
 & =\frac{1}{2}\mathbb{E}\left[\left(\exp\left(\theta f\left(X\right)\right)-\exp\left(\theta f\left(X^{\prime}\right)\right)\right)F\left(X,X^{\prime}\right)\right],\label{eq: mgf symmetry}
\end{flalign}
where the second inequality follows from (\ref{eq: kernel symmetry}). To bound $m^{\prime}\left(\theta\right)$ we apply the following exponential
mean-value inequality, stated in a more general form in \citet{paulin2016efron}.
\begin{lemma}
\label{lem: exp mv ineq}For all constants $x$, $y$, and $c$ in
$\mathbb{R}$ and $s>0$, it holds that
\[
\vert\left(e^{x}-e^{y}\right)c\vert\leq\frac{1}{4}\left(s\left(x-y\right)^{2}+s^{-1}c^{2}\right)\left(e^{x}+e^{y}\right).
\]
\end{lemma}
\noindent In particular, by (\ref{eq: mgf symmetry}) and Lemma \ref{lem: exp mv ineq},
we obtain the bound
\begin{flalign*}
\vert m^{\prime}\left(\theta\right)\vert & \leq\frac{1}{2}\mathbb{E}\left[\vert\left(\exp\left(\theta f\left(X\right)\right)-\exp\left(\theta f\left(X^{\prime}\right)\right)\right)F\left(X,X^{\prime}\right)\vert\right]\\
 & \leq\frac{1}{8}\inf_{t>0}\mathbb{E}\left[\left(t\left(\theta f\left(X\right)-\theta f\left(X^{\prime}\right)\right)^{2}+t^{-1}F\left(X,X^{\prime}\right)^{2}\right)\left(\exp\left(\theta f\left(X\right)\right)+\exp\left(\theta f\left(X^{\prime}\right)\right)\right)\right]\\
 & =\frac{\vert\theta\vert}{4}\inf_{t>0}\mathbb{E}\left[\left(t\left(f\left(X\right)-f\left(X^{\prime}\right)\right)^{2}+t^{-1}F\left(X,X^{\prime}\right)^{2}\right)\exp\left(\theta f\left(X\right)\right)\right]\\
 & =\frac{\vert\theta\vert}{2}\inf_{t>0}\mathbb{E}\Bigg[\left(\frac{t}{2}\mathbb{E}\left[\left(f\left(X\right)-f\left(X^{\prime}\right)\right)^{2}\mid X\right]+\frac{1}{2t}\mathbb{E}\left[F\left(X,X^{\prime}\right)^{2}\mid X\right]\right)  \\
 & \quad\quad\quad\quad\quad\quad \cdot \mathbb{E}\left[\exp\left(\theta f\left(X\right)\right)\mid X\right]\Bigg]\\
 & =\frac{\vert\theta\vert}{2}\inf_{t>0}\mathbb{E}\left[\left(tU_{f}\left(X\right)+t^{-1}U_{F}\left(X\right)\right)\mathbb{E}\left[\exp\left(\theta f\left(X\right)\right)\mid X\right]\right]\\
 & \leq\frac{\vert\theta\vert}{2}\mathbb{E}\left[\left(sU_{f}\left(X\right)+s^{-1}U_{F}\left(X\right)\right)\mathbb{E}\left[\exp\left(\theta f\left(X\right)\right)\mid X\right]\right]\\
 & \leq\vert\theta\vert v\mathbb{E}\left[\exp\left(\theta f\left(X\right)\right)\right]
\end{flalign*}
for all $\theta\in\mathbb{R}$. Thus, we have that 
\[
m^{\prime}\left(\theta\right)\leq u\theta m\left(\theta\right)
\]
for all $\theta>0$. As $m\left(\cdot\right)$ is a convex function
and $m^{\prime}\left(0\right)=0$, $m^{\prime}\left(\theta\right)$
always has the same sign as $\theta$, we find that
\[
\frac{\text{d}}{\text{d}\theta}\log m\left(\theta\right)\leq u\theta.
\]
As a consequence, and by $m\left(0\right)=1$, we have that 
\[
\log m\left(\theta\right)\leq\int_{0}^{\theta}ut\text{d}t\leq\frac{u\theta^{2}}{2}.
\]
By the Chernoff bound (see e.g., \citealp[Equation 2.5, ][]{wainwright2019high}),
we have
\[
\log P\left\{ f\left(X\right)\geq\delta\right\} \le\inf_{\theta\ge0}\left(\log m\left(\theta\right)-\theta\delta\right)\leq\inf_{\theta\geq0}\left(\frac{u\theta^{2}}{2}-\theta\delta\right)
\]

\noindent Solving this minimization with $\theta=\delta/u$, we find
\[
P\left\{ f\left(X\right)\geq t\right\} \leq\exp\left(\frac{-\delta^{2}}{2u}\right),
\]
as required. The analogous lower tail bound follows from an identical argument, which completes the proof of \eqref{eq: chatterjee concentration display}. 
To prove \eqref{eq: chatterjee moment display} we apply the following result stated in \citet{chatterjee2007stein}. 
\begin{theorem}[Theorem 1.5, (iii), \citealp{chatterjee2007stein}]Reintroduce the notation and assumptions from the statement of Theorem \ref{thm: chatterjee concentration}.
Define 
\[
\Delta\left(X\right)=\frac{1}{2}\mathbb{E}\left[\vert F\left(X,X^{\prime}\right)\left(f\left(X\right)-f\left(X^{\prime}\right)\right)\vert\mid X\right].
\]
 The Burkholder-Davis-Gundy inequality 
\[
\mathbb{E}\left[f\left(X\right)^{2r}\right]\leq\left(2r-1\right)^{r}\mathbb{E}\left[\Delta\left(X\right)^{r}\right]
\]
holds for any positive integer $r$. 
\end{theorem}
\noindent
The inequality
\begin{flalign*}
\mathbb{E}\left[\Delta\left(X\right)^{r}\right] & =\mathbb{E}\left[\mathbb{E}\left[\vert\left(f\left(X\right)-f\left(X^{\prime}\right)\right)F\left(X,X^{\prime}\right)\vert\mid X\right]^{r}\right]\\
 & \leq\mathbb{E}\left[\mathbb{E}\left[\vert\left(f\left(X\right)-f\left(X^{\prime}\right)\right)F\left(X,X^{\prime}\right)\vert^{r}\mid X\right]\right]\tag{Jensen}\\
 & =\mathbb{E}\left[\mathbb{E}\left[\left(\left(s^{-1}\left(f\left(X\right)-f\left(X^{\prime}\right)\right)^{2r}\right)\left(sF\left(X,X^{\prime}\right)^{2r}\right)\right)^{1/2}\mid X\right]\right]\\
 & \leq\mathbb{E}\left[\mathbb{E}\left[s^{-1}\left(f\left(X\right)-f\left(X^{\prime}\right)\right)^{2r}+sF\left(X,X^{\prime}\right)^{2r}\mid X\right]\right]\tag{Young}\\
 & =s^{-1}\mathbb{E}\left[{U}_{f}^{(r)}(X)\right]+s\mathbb{E}\left[{U}_{F}^{(r)}(X)\right]
\end{flalign*}
then completes the proof.\hfill$\qed$

\subsection{Proof of \cref{lem: variance bounds}\label{subsec: proof of variance bounds}}
We apply the following Lemma.
\begin{lemma}
\label{lem: sequence bound}Let $\left(X_{m}\right)_{m\geq0}$ be
a sequence of real-valued random variables. Suppose that the inequality
\begin{equation}
\mathbb{E}\left[X_{m}^{2^{c}}\right]\leq h_{m}^{2^{c}}\label{eq: h sequence bound}
\end{equation}
holds for each $m\geq0$ and positive integer $c$, where $\left(h_{m}\right)_{m\geq0}$
is a deterministic sequence of nonnegative real numbers. If there
exists a square integrable random variable $W$ such that 
\[
\left(\sum_{m=0}^{\infty}X_{m}\right)^{2^{c}}\leq W
\]
almost surely, then the inequality 
\[
\mathbb{E}\left[\left(\sum_{m=0}^{\infty}X_{m}\right)^{2^{c}}\right]\leq\left(\sum_{m=0}^{\infty}h_{m}\right)^{2^{c}}
\]
holds almost surely. 
\end{lemma}
\noindent
Define the objects 
\begin{flalign*}
U_{f}^{(r)}\left(Z\right)=\frac{1}{2}\mathbb{E}\left[\left(f\left(Z\right)-f\left(Z^{\prime}\right)\right)^{2r}\mid Z\right] & \quad\text{and}\\
U_{F}^{(r)}\left(Z\right)=\frac{1}{2}\mathbb{E}\left[\left(\sum_{m=0}^{\infty}\mathbb{E}\left[f\left(Z_{m}\right)-f\left(Z_{m}^{\prime}\right)\mid Z_{0}=Z,Z_{0}^{\prime}=Z^{\prime}\right]\right)^{2r}\mid Z\right].
\end{flalign*}
By the (\ref{eq: abstract integrability}), we have that
\[
\mathbb{E}\left[\sum_{m=0}^{\infty}f(Z_{m})-f(Z_{m}^{\prime})\mid Z_{0}=Z,Z_{0}^{\prime}=Z^{\prime}\right]^{2^{c}}\leq W.
\]
Thus, (\ref{eq: c_i general bound}) guarantees that
the conditions of Lemma \ref{lem: sequence bound} are satisfied,
and we have that 
\[
\mathbb{E}\left[U_{F}^{(2^{c-1})}\left(Z\right)\mid\bm{\pi}\right]\leq\frac{1}{2}\left(\sum_{i=0}^{\infty}h_{i}\right)^{2^{c}}\quad\text{and}\quad\mathbb{E}\left[U_{f}^{(2^{c-1})}\left(Z\right)\mid\bm{\pi}\right]\leq\frac{1}{2}h_{0}^{2^{c}}.
\]
By Markov's inequality, we obtain
\[
P\left\{ U_{F}^{(2^{c-1})}\left(Z\right)\geq\frac{2}{\delta}\mathbb{E}\left[U_{F}^{(2^{c-1})}\left(Z\right)\mid\bm{\pi}\right]\quad\text{or}\quad U_{f}^{(2^{c-1})}\left(Z\right)\geq\frac{2}{\delta}\mathbb{E}\left[U_{f}^{(2^{c-1})}\left(Z\right)\mid\bm{\pi}\right]\mid\bm{\pi}\right\} \leq\delta.
\]
Hence, by Lemma \ref{lem: sequence bound} and DeMorgan's law, the
probability that both
\begin{flalign*}
U_{F}^{(2^{c-1})}\left(Z\right) & \leq\frac{2}{\delta}\mathbb{E}\left[U_{F}^{(2^{c-1})}\left(Z\right)\mid\bm{\pi}\right]\leq\frac{1}{\delta}\left(\sum_{m=0}^{\infty}h_{m}\right)^{2^{c}}\quad\text{and}\\
U_{f}^{(2^{c-1})}\left(Z\right) & \leq\frac{2}{\delta}\mathbb{E}\left[U_{f}^{(2^{c-1})}\left(Z\right)\mid\bm{\pi}\right]\leq\frac{h_{0}^{2^{c}}}{\delta}
\end{flalign*}
hold is greater than $1-\delta$.\hfill$\qed$

\subsection{Proof of \cref{lem: deterministic bound}\label{subsec: proof of deterministic bound}}
Recall the definition of the collections $\left(\bm{\pi}_{m},\bm{\pi}_{m}^{\prime}\right)_{m\geq0}$
given in \eqref{eq: kernel coupling state-1}. Fix $Z_{0}=Z$ and $Z_{0}^{\prime}=Z^{\prime}$
throughout. For any $i$ in $\left[n\right]$, let $\mathsf{s}_{m,\ell}\left(i\right)$
denote the element of the collection 
\[
\mathsf{r}(\pi_{m,\ell})=\left(\mathsf{s}_{1}(\pi_{m,\ell}),\ldots,\mathsf{s}_{g}(\pi_{m,\ell})\right)
\]
that contains the index $i$ and set $\mathsf{s}_{m,\ell}(i)$ equal
to $\varnothing$ if no element of $\mathsf{r}(\pi_{m,\ell})$ contains
$i$. 

We begin by defining three events that will determine the structure
of our argument. By construction, the collections $\bm{\pi}_{m}$
and $\bm{\pi}_{m}^{\prime}$ are either identical or differ in exactly
two indices in their $L$th element. Let $\mathcal{E}_{m}$ denote
the event that $\bm{\pi}_{m}$ and $\bm{\pi}_{m}^{\prime}$ differ.
On the event $\mathcal{E}_{m}$, let $i_{1,m}$ and $i_{2,m}$ denote
the two indices in which the $L$th elements of $\bm{\pi}_{m}$ and
$\bm{\pi}_{m}^{\prime}$ differ. Define the random variables $B_{m}$
and $C_{m}$ such that, conditional on $\mathcal{E}_{m}$, each is
uniformly distributed on $\left\{ i_{1,m},i_{2,m}\right\} $ such
that $B_{m}\neq C_{m}$. On the complement of $\mathcal{E}_{m}$,
set these indices uniformly at random. Let $\mathcal{F}_{m}$ denote
the event that 
\[
\mathsf{s}_{m,L}\left(B_{m}\right)\neq\mathsf{s}_{m,L}\left(C_{m}\right),
\]
i.e., the event that the indices that differ are not in the same element
of the collection $\mathsf{r}(\pi_{m,\ell})$. Finally, let $\mathcal{G}_{m}$
denote the event that 
\[
\mathsf{s}_{m,L}\left(B_{m}\right)\neq\varnothing\quad\text{and}\quad\mathsf{s}_{m,L}\left(C_{m}\right)\neq\varnothing,
\]
i.e., the event that the indices that differ are both in the collection
$\mathsf{r}(\pi_{m,\ell})$. 

By \cref{assu: invariance}, the event $\mathcal{H}_{m}=\mathcal{E}_{m}\cap\mathcal{F}_{m}$
is a necessary condition for $a\left(Z_{m}\right)-a\left(Z_{m}^{\prime}\right)\neq0$. Thus, we can compute 
\begin{flalign*}
P\left\{ \mathcal{H}_{m} \mid \mathcal{E}_0 \right\}  & =P\left\{ \mathcal{F}_{m}\mid\mathcal{E}_{m},\mathcal{E}_{0}\right\} P\left\{ \mathcal{E}_{m}\mid\mathcal{E}_{0}\right\} \\
 & =\Big(P\left\{ \mathsf{s}_{m,L}\left(B_{m}\right)\neq\mathsf{s}_{m,L}\left(C_{m}\right)\mid\mathsf{s}_{m,L}\left(B_{m}\right)\neq\varnothing,\mathcal{E}_{m}\right\} P\left\{ \mathsf{s}_{m,L}\left(B_{m}\right)\neq\varnothing\mid\mathcal{E}_{m}\right\} \\
 & \quad+P\left\{ \mathsf{s}_{m,L}\left(B_{m}\right)\neq\mathsf{s}_{m,L}\left(C_{m}\right)\mid\mathsf{s}_{m,L}\left(B_{m}\right)=\varnothing,\mathcal{E}_{m}\right\} P\left\{ \mathsf{s}_{m,L}\left(B_{m}\right)=\varnothing\mid\mathcal{E}_{m}\right\} \Big)\\
 & \quad\quad\cdot P\left\{ \mathcal{E}_{m}\mid\mathcal{E}_{0}\right\} \\
 & =\left(\frac{n-b}{n-1}\frac{kb}{n}+\frac{kb}{n-1}\frac{n-kb}{b}\right)\left(1-\frac{2}{gn^{2}}\right)^{m}\\
 & =\left(\frac{kb\left(2n-kb-b\right)}{n(n-1)}\right)\left(1-\frac{2}{gn^{2}}\right)^{m}
\end{flalign*}
for all $m\geq0$ by the law of total probability.  Moreover, we have that
\[
P\left\{ \mathcal{H}_{m} \mid Z,Z^{\prime}\right\} \leq P\left\{ \mathcal{H}_{m} \mid \mathcal{E}_0 \right\}
\]
almost surely. Consequently, we find that 
\begin{flalign}
 & \mathbb{E}\left[\mathbb{E}\left[a\left(Z_{m}\right)-a\left(Z_{m}^{\prime}\right)\mid Z,Z^{\prime}\right]^{2r}\mid\bm{\pi}\right]\nonumber \\
 & = \mathbb{E}\left[\left(P\left\{ \mathcal{H}_{m} \mid Z,Z^{\prime}\right\}\mathbb{E}\left[a\left(Z_{m}\right)-a\left(Z_{m}^{\prime}\right)\mid Z,Z^{\prime},\mathcal{H}_{m}\right]\right)^{2r}\mid\bm{\pi}\right]\nonumber \\
 & \leq \left(P\left\{ \mathcal{H}_{m} \mid \mathcal{E}_0 \right\}\right)^{2r}
  \mathbb{E}\left[\mathbb{E}\left[\left(a\left(Z_{m}\right)-a\left(Z_{m}^{\prime}\right)\right)^{2r} \mid Z,Z^{\prime},\mathcal{H}_{m}\right]\mid\bm{\pi}\right]\nonumber\tag{Jensen}\\
 & =\left(1-\frac{2}{gn^{2}}\right)^{2mr}\left(\frac{kb\left(2n-kb-b\right)}{n(n-1)}\right)^{2r}\mathbb{E}\left[\mathbb{E}\left[\left(a\left(Z_{m}\right)-a\left(Z_{m}^{\prime}\right)\right)^{2r}\mid\mathcal{H}_{m},\bm{\pi}\right]\mid\bm{\pi}\right]\nonumber\\
 & \leq\left(1-\frac{2}{gn^{2}}\right)^{2mr}\left(\frac{2kb\left(2n-kb-b\right)}{n^2}\right)^{2r}
\mathbb{E}\left[\left(a\left(Z_{m}\right)-a\left(Z_{m}^{\prime}\right)\right)^{2r}\mid\mathcal{H}_{m}\right]~,\label{eq: unconditional difference}
\end{flalign}
where the final inequality follows from the fact that $\bm{\pi}$ is uniformly distributed independently of $D$ and the elementary inequality $n/(n-1)\leq2$. Thus, it remains to bound the expectation in (\ref{eq: unconditional difference}). 

To ease notation, we now drop the dependence on $m$ and $L$. Observe that 
\begin{flalign*}
P\left\{ \mathcal{G}\mid\mathcal{H}\right\}  & =P\left\{ \mathsf{s}\left(B\right)\neq\varnothing\mid\mathsf{s}\left(C\right)\neq\varnothing\right\} =\frac{kb-b}{n-b}
\end{flalign*}
and that therefore
\begin{flalign}
\mathbb{E}\left[\left(a\left(Z\right)-a\left(Z^{\prime}\right)\right)^{2r}\mid\mathcal{H},\bm{\pi}\right] & =\left(\frac{kb-b}{n-b}\right)\mathbb{E}\left[\left(a\left(Z\right)-a\left(Z^{\prime}\right)\right)^{2r}\mid\mathcal{\mathcal{G}}\right]\nonumber \\
 & +\left(\frac{n-kb}{n-b}\right)\mathbb{E}\left[\left(a\left(Z\right)-a\left(Z^{\prime}\right)\right)^{2r}\mid\mathcal{\mathcal{H}}\setminus\mathcal{G}\right].\label{eq: break up square}
\end{flalign}
Define the sets 
\[
\hat{\mathsf{s}}_{i}=\mathsf{s}(i)\setminus i\quad\text{and}\quad\bar{\mathsf{s}}=\tilde{\mathsf{s}}(B)\cap\tilde{\mathsf{s}}(C).
\]
With an abuse of notation, we let 
\begin{flalign*}
\psi(D_{i},\hat{\eta}(D_{j},D_{\hat{\mathsf{s}}_{k}})) & =\psi(D_{i},\hat{\eta}(D_{j}\cap D_{\hat{\mathsf{s}}_{k}}\cap D_{\bar{\mathsf{s}}}))\quad\text{and}\\
h(D_{i},D_{j},D_{\hat{\mathsf{s}}_{k}}) & =\psi(D_{i},\hat{\eta}(D_{j},D_{\hat{\mathsf{s}}_{k}}))-\psi(D_{j},\hat{\eta}(D_{i},D_{\hat{\mathsf{s}}_{k}})).
\end{flalign*}
Observe that 
\begin{flalign}
 & \mathbb{E}\left[\left(a\left(Z\right)-a\left(Z^{\prime}\right)\right)^{2r}\mid\mathcal{\mathcal{H}}\setminus\mathcal{G}\right]\nonumber \\
 & =\mathbb{E}\left[\left(a\left(Z\right)-a\left(Z^{\prime}\right)\right)^{2r}\mid\mathcal{\mathcal{H}}\setminus\mathcal{G}\right]=\left(\frac{1}{gkb}\right)^{2r}\mathbb{E}\left[h(D_{B},D_{C},D_{\tilde{s}_{C}})^{2r}\right]~.\label{eq: missing one}
\end{flalign}
On the other hand, we can decompose
\begin{flalign}
 & \mathbb{E}\left[\left(a\left(Z\right)-a\left(Z^{\prime}\right)\right)^{2}\mid\mathcal{G}\right]\nonumber \\
 & =\left(\frac{1}{gkb}\right)^{2r}\mathbb{E}\left[\left(h(D_{B},D_{C},D_{\hat{\mathsf{s}}_{C}})+h(D_{C},D_{B},D_{\hat{\mathsf{s}}_{B}})\right)^{2r}\right]\nonumber \\
 & =\left(\frac{1}{gkb}\right)^{2r}\mathbb{E}\left[\left(h(D_{B},D_{C},D_{\hat{\mathsf{s}}_{C}})-h(D_{B},D_{B},D_{\hat{\mathsf{s}}_{B}})\right)^{2r}\right]~.\label{eq:  none missing}
\end{flalign}
Consequently, by (\ref{eq: break up square}), (\ref{eq: missing one}),
and (\ref{eq:  none missing}), it suffices to express suitable bounds
for the expectations 
\begin{flalign}
 & \mathbb{E}\left[\left(\psi(D_{B},\hat{\eta}(D_{C},D_{\hat{\mathsf{s}}_{C}}))-\psi(D_{C},\hat{\eta}(D_{B},D_{\hat{\mathsf{s}}_{C}}))\right)^{2r}\right]\quad\text{and}\nonumber \\
 & \mathbb{E}\Bigg[\bigg(\left(\psi(D_{B},\hat{\eta}(D_{C},D_{\hat{\mathsf{s}}_{C}}))-\psi(D_{C},\hat{\eta}(D_{B},D_{\hat{\mathsf{s}}_{C}}))\right)\label{eq: double diff}\\
 & \quad\quad\quad-\left(\psi(D_{B},\hat{\eta}(D_{C},D_{\hat{\mathsf{s}}_{B}}))-\psi(D_{C},\hat{\eta}(D_{B},D_{\hat{\mathsf{s}}_{B}}))\right)\bigg)^{2r}\Bigg]\nonumber 
\end{flalign}
respectively. To this end, recall that $\tilde{D}_{i}$ are independent
copies of $D_{i}$ for each $i$. Observe that 
\begin{flalign}
 & \mathbb{E}\left[\left(\psi(D_{B},\hat{\eta}(D_{C},D_{\hat{\mathsf{s}}_{C}}))-\psi(D_{C},\hat{\eta}(D_{B},D_{\hat{\mathsf{s}}_{C}}))\right)^{2r}\right]\nonumber \\
 & =\mathbb{E}\left[\left(\psi(D_{B},\hat{\eta}(D_{C},D_{\hat{\mathsf{s}}_{C}}))-\psi(\tilde{D}_{B},\hat{\eta}(\tilde{D}_{C},D_{\hat{\mathsf{s}}_{C}}))+\psi(\tilde{D}_{B},\hat{\eta}(\tilde{D}_{C},D_{\hat{\mathsf{s}}_{C}}))-\psi(D_{C},\hat{\eta}(D_{B},D_{\hat{\mathsf{s}}_{C}}))\right)^{2r}\right]\nonumber \\
 & =\sum_{q=0}^{2r}{2r \choose q}\mathbb{E}\bigg[\left(\psi(D_{B},\hat{\eta}(D_{C},D_{\hat{\mathsf{s}}_{C}}))-\psi(\tilde{D}_{B},\hat{\eta}(\tilde{D}_{C},D_{\hat{\mathsf{s}}_{C}}))\right)^{2r-q}\nonumber \\
 & \quad\quad\quad\quad\quad\quad\left(\psi(\tilde{D}_{B},\hat{\eta}(\tilde{D}_{C},D_{\hat{\mathsf{s}}_{C}}))-\psi(D_{C},\hat{\eta}(D_{B},D_{\hat{\mathsf{s}}_{C}}))\right)^{q}\bigg]\tag{Binomial Theorem}\nonumber \\
 & \leq2^{2r}\mathbb{E}\left[\left(\psi(D_{B},\hat{\eta}(D_{C},D_{\hat{\mathsf{s}}_{C}}))-\psi(\tilde{D}_{B},\hat{\eta}(\tilde{D}_{C},D_{\hat{\mathsf{s}}_{C}}))\right)^{2r}\right]\label{eq: move to two ind swaps}
\end{flalign}
where the inequality follows from the fact that $B$ and $C$ are
exchangeable and the Hölder inequality. Similarly, we have that 
\begin{flalign}
 & \mathbb{E}\left[\left(\psi(D_{B},\hat{\eta}(D_{C},D_{\hat{\mathsf{s}}_{C}}))-\psi(\tilde{D}_{B},\hat{\eta}(\tilde{D}_{C},D_{\hat{\mathsf{s}}_{C}}))\right)^{2r-q}\right]\nonumber \\
 & =\mathbb{E}\left[\left(\psi(D_{B},\hat{\eta}(D_{C},D_{\hat{\mathsf{s}}_{C}}))-\psi(D_{B},\hat{\eta}(\tilde{D}_{C},D_{\hat{\mathsf{s}}_{C}}))+\psi(D_{B},\hat{\eta}(\tilde{D}_{C},D_{\hat{\mathsf{s}}_{C}}))-\psi(\tilde{D}_{B},\hat{\eta}(\tilde{D}_{C},D_{\hat{\mathsf{s}}_{C}}))\right)^{2r}\right]\nonumber \\
 & =\sum_{q=0}^{2r}{2r \choose q}\mathbb{E}\bigg[\left(\psi(D_{B},\hat{\eta}(D_{C},D_{\hat{\mathsf{s}}_{C}}))-\psi(D_{B},\hat{\eta}(\tilde{D}_{C},D_{\hat{\mathsf{s}}_{C}}))\right)^{2r-q}\nonumber \\
 & \quad\quad\quad\quad\quad\quad\quad\quad\left(\psi(D_{B},\hat{\eta}(\tilde{D}_{C},D_{\hat{\mathsf{s}}_{C}}))-\psi(\tilde{D}_{B},\hat{\eta}(\tilde{D}_{C},D_{\hat{\mathsf{s}}_{C}}))\right)^{q}\bigg]\tag{Binomial Theorem}\nonumber \\
 & \leq\sum_{q=0}^{2r}{2r \choose q}\left(\sigma_{\mathsf{valid}}^{(2r)}\right)^{\frac{2r-q}{2r}}\left(\sigma_{\mathsf{train}}^{(2r,1)}\right)^{\frac{q}{2r}}\tag{Hölder}\nonumber \\
 & \leq2^{2r}\sigma_{\mathsf{max}}^{(2r)},\label{eq: in terms of stabilities}
\end{flalign}
where the last inequality follows by the definitions of $\sigma_{\mathsf{valid}}^{(2r)}$
and $\sigma_{\mathsf{train}}^{(2r,1)}$. Thus, we have that 
\begin{equation}
\label{eq: contibute sigma}
\mathbb{E}\left[\left(\psi(D_{B},\hat{\eta}(D_{C},D_{\hat{\mathsf{s}}_{C}}))-\psi(D_{C},\hat{\eta}(D_{B},D_{\hat{\mathsf{s}}_{C}}))\right)^{2r}\right]\leq2^{4r}\sigma_{\mathsf{max}}^{(2r)}
\end{equation}
by (\ref{eq: move to two ind swaps}) and (\ref{eq: in terms of stabilities}). 

Next, we consider the double difference term (\ref{eq: double diff}).
In this case, we have that 
\begin{flalign}
 & \mathbb{E}\Bigg[\bigg(\left(\psi(D_{B},\hat{\eta}(D_{C},D_{\hat{\mathsf{s}}_{C}}))-\psi(D_{C},\hat{\eta}(D_{B},D_{\hat{\mathsf{s}}_{C}}))\right)\nonumber\\
 & \quad\quad\quad-\left(\psi(D_{B},\hat{\eta}(D_{C},D_{\hat{\mathsf{s}}_{B}}))-\psi(D_{C},\hat{\eta}(D_{B},D_{\hat{\mathsf{s}}_{B}}))\right)\bigg)^{2r}\Bigg]\nonumber\\
 & =\mathbb{E}\Bigg[\bigg(\left(\psi(D_{B},\hat{\eta}(D_{C},D_{\hat{\mathsf{s}}_{C}}))-\psi(D_{B},\hat{\eta}(D_{C},D_{\hat{\mathsf{s}}_{B}}))\right)\nonumber\\
 & \quad\quad\quad-\left(\psi(D_{C},\hat{\eta}(D_{B},D_{\hat{\mathsf{s}}_{C}}))-\psi(D_{C},\hat{\eta}(D_{B},D_{\hat{\mathsf{s}}_{B}}))\right)\bigg)^{2r}\Bigg]\nonumber\\
 & =\sum_{q=0}^{2r}{2r \choose q}\mathbb{E}\Bigg[\bigg(\left(\psi(D_{B},\hat{\eta}(D_{C},D_{\hat{\mathsf{s}}_{C}}))-\psi(D_{B},\hat{\eta}(D_{C},D_{\hat{\mathsf{s}}_{B}}))\right)\nonumber\\
 & \quad\quad\quad\quad\quad\quad\cdot\left(\psi(D_{C},\hat{\eta}(D_{B},D_{\hat{\mathsf{s}}_{C}}))-\psi(D_{C},\hat{\eta}(D_{B},D_{\hat{\mathsf{s}}_{B}}))\right)\bigg)^{2r}\Bigg]\tag{Binomial Theorem}\nonumber\\
 & \leq2^{2r}\mathbb{E}\left[\left(\psi(D_{B},\hat{\eta}(D_{C},D_{\hat{\mathsf{s}}_{C}}))-\psi(D_{B},\hat{\eta}(D_{C},D_{\hat{\mathsf{s}}_{B}}))\right)^{2r}\right]=2^{2r}\sigma_{\mathsf{train}}^{(2r,b-1)}\label{eq: contribute sigma train}
\end{flalign}
where the final inequality follows from the fact that $B$ and $C$
are exchangeable and the Hölder inequality. Putting the pieces together,
we have that 
\begin{flalign*}
 & \mathbb{E}\left[\mathbb{E}\left[a\left(Z_{m}\right)-a\left(Z_{m}^{\prime}\right)\mid Z,Z^{\prime}\right]^{2r}\mid\bm{\pi}\right]\\
 & \leq2^{4r}\left(1-\frac{2}{gn^{2}}\right)^{2mr}\left(\frac{2n-kb-b}{gn^{2}}\right)^{2r}
 \left( 
 \left(\frac{n-kb}{n-b}\right) 2^{2r} \sigma_{\mathsf{max}}^{(2r)}
 + 
\left(\frac{kb-b}{n-b}\right) \sigma_{\mathsf{train}}^{(2r,b-1)}
 \right)
\end{flalign*}
as required.\hfill$\qed$

\subsection{Proof of \cref{thm: variance estimator concentration}}
Throughout, we us the short hand $\hat{v}_{g,k}\left(Z\right)$ to denote $\hat{v}(\mathsf{R}_{g,k}, D).$ 
We begin by decomposing the estimator $\hat{v}_{g,k}\left(Z\right)$ into
two parts that will each be easier to handle when considered in isolation.
To this end, define the statistics 
\begin{flalign*}
\tilde{v}_{g,k}\left(Z\right) & =\frac{1}{g^{2}k^{2}}\sum_{l=1}^{g}\sum_{i,i^{\prime}=1}^{k}\left(T(\mathsf{s}_{l,i},Y)-\bar{a}\left(D\right)\right)\left(T(\mathsf{s}_{l,i^{\prime}},D)-\bar{a}\left(D\right)\right)\quad\text{and}\\
\check{v}_{g,k}\left(Z\right) & =\left(a\left(Z\right)-\bar{a}\left(D\right)\right)^{2}.
\end{flalign*}
Observe that both $\tilde{v}_{g,k}\left(Z\right)$ and $\check{v}_{g,k}\left(Z\right)$
are unbiased for $v\left(Z\right)$. Moreover, we can write
\begin{flalign}
\hat{v}_{g,k}\left(Z\right) 
& =\frac{1}{g\left(g-1\right)k^{2}}\sum_{l=1}^{g}\sum_{i,i^{\prime}=1}^{k}\left(T(\mathsf{s}_{l,i},D)-a\left(Z\right)\right)\left(T(\mathsf{s}_{l,i^{\prime}},D)-a\left(Z\right)\right)\nonumber \\
 & =\frac{1}{g\left(g-1\right)k^{2}}\sum_{l=1}^{g}\sum_{i,i^{\prime}=1}^{k}\left(T(\mathsf{s}_{l,i},D)-\bar{a}\left(D\right)\right)\left(T(\mathsf{s}_{l,i^{\prime}},D)-a\left(Z\right)\right)\nonumber \\
 & \quad+\frac{1}{g\left(g-1\right)k}\sum_{l=1}^{g}\sum_{i=1}^{g}\left(\bar{a}\left(D\right)-a\left(Z\right)\right)\left(T(\mathsf{s}_{l,i},D)-a\left(Z\right)\right)\nonumber \\
 & =\frac{1}{g\left(g-1\right)k^{2}}\sum_{l=1}^{g}\sum_{i,i^{\prime}=1}^{k}\left(T(\mathsf{s}_{l,i},D)-\bar{a}\left(D\right)\right)\left(T(\mathsf{s}_{l,i^{\prime}},D)-a\left(Z\right)\right)\nonumber \\
 & =\frac{g}{g-1}\tilde{v}\left(Z\right)+\frac{1}{g\left(g-1\right)k}\sum_{l=1}^{g}\sum_{i=1}^{k}\left(T(\mathsf{s}_{l,i},D)-\bar{a}\left(D\right)\right)\left(\bar{a}\left(D\right)-a\left(Z\right)\right)\nonumber \\
 & =\frac{g}{g-1}\tilde{v}_{g,k}\left(Z\right)-\frac{1}{g-1}\check{v}_{g,k}\left(Z\right).\label{eq: break up v hat}
\end{flalign}
Hence, it will suffice to characterize the exponential rates of conditional
concentration for the statistics $\tilde{v}_{g,k}\left(Z\right)$ and $\check{v}_{g,k}\left(Z\right)$.
These are established through the application of the following Lemma,
which follows from an argument very similar to the proof of Theorem
\ref{thm: cross split concentration}.
\begin{lemma}
\label{lem: a(Z) sq concentration}Suppose that \cref{assu: invariance,assu: linearity} hold and that the data $D$ are independently and identically distributed. If the eighth-order split-stability $\zeta^{(8)}$ is finite,
then:

\noindent \textbf{(i)} The conditional concentration inequality 
\begin{flalign}
 & \log \frac{1}{2}P\left\{ \vert\tilde{v}\left(Z\right)-v\left(D\right)\vert\geq t\mid D\right\} 
 \lesssim - 
 \frac{\delta}{(2 - \varphi k - \varphi)^4} \frac{g^3 t ^2}{\Gamma_{k,\varphi,b}^{(2)}}~,\label{eq: var tilde concentration}
\end{flalign}
holds for all $t>0$ with probability greater than $1-\delta$.

\noindent \textbf{(ii)} The conditional concentration inequality 
\begin{flalign}
 & \log \frac{1}{2}P\left\{ \vert\check{v}\left(Z\right)-v\left(D\right)\vert\geq t\mid D\right\} 
  \lesssim - 
 \frac{\delta}{(2 - \varphi k - \varphi)^4} \frac{g^2 t ^2}{\Gamma_{k,\varphi,b}^{(2)}}~,\label{eq: a(Z) sq concentration}
\end{flalign}
holds for all $t>0$ with probability greater than $1-\delta$.
\end{lemma}
\noindent
Putting the pieces together, by (\ref{eq: break up v hat}), we have that
\begin{flalign*}
 & P\left\{ \vert\hat{v}_{g,k}\left(Z\right)-v_{g,k}\left(D\right)\vert\leq t\mid D\right\} \\
 & = P\left\{ 
 \vert\frac{g}{g-1}\left(\tilde{v}_{g,k}\left(Z\right)-v_{g,k}\left(D\right)\right)
 -\frac{1}{g-1}\left(\check{v}_{g,k}\left(Z\right)-v_{g,k}\left(D\right)\right)\vert\leq t\mid D\right\} \\
 & \geq P\left\{ 
 \frac{g}{g-1}\vert\tilde{v}_{g,k}\left(Z\right)
 -v_{g,k}\left(D\right)\vert
 +\frac{1}{g-1}\vert\check{v}_{g,k}\left(Z\right)
 -v_{g,k}\left(D\right)\vert\leq t\mid D\right\} \\
 & \geq P\left\{ \frac{g}{g-1}\vert\tilde{v}_{g,k}\left(Z\right)
 -v\left(D\right)\vert\leq\frac{\sqrt{g}}{1+\sqrt{g}}t,
 \frac{1}{g-1}\vert\check{v}\left(Z\right)
 -v\left(D\right)\vert\leq\frac{1}{1+\sqrt{g}}t\mid D\right\} \\
 & \geq 
 P\left\{ \vert\tilde{v}_{g,k}\left(Z\right)
 -v\left(D\right)\vert\leq\frac{1}{\sqrt{g}}\frac{g-1}{1+\sqrt{g}}t\mid D\right\} 
 +P\left\{ \vert\check{v}_{g,k}\left(Z\right)-
 v\left(D\right)\vert\leq\frac{g-1}{1+\sqrt{g}}t\mid D\right\} -1\\
 & \geq 1- 
 4 \exp\left(- \frac{\delta}{C  (2 - \varphi k - \varphi)^4} \frac{g^3 t ^2}{\Gamma_{k,\varphi,b}^{(2)}}\right)
\end{flalign*}
with probability greater than $1-\delta$ for some universal constant $C$, where the last inequality
follows by Lemma \ref{lem: a(Z) sq concentration} and the facts
that $4g\geq\left(1+\sqrt{g}\right)^{2}$ for all $g\geq1$ and $\left(1/4\right)g^{2}\leq\left(g-1\right)^{2}$
for all $g\geq2$.\hfill\qed

\subsection{Proof of \cref{lem: reproducability decomposition}}
Define the event 
\[
\mathcal{W}_{\lambda}\left(D\right)=\mathcal{U}_{k,\lambda}\left(D\right)\cap\mathcal{Q}_{k,\lambda}(D),
\]
the quantity 
\[
W\left(D\right)=\left(U(\mathsf{R}_{\hat{g},k},D)-U(\mathsf{R}_{\hat{g}^{\prime},k}^{\prime},D\right)+\left(Q(\mathsf{R}_{\hat{g},k},D)-Q(\mathsf{R}_{\hat{g}^{\prime},k}^{\prime},D)\right).
\]
By the decomposition \eqref{eq: sequential decomposition}, we have that
\begin{flalign*}
 & \bigg\vert P\left\{ a(\mathsf{R}_{\hat{g},k},D)-a(\mathsf{R}_{\hat{g}^{\prime},k}^{\prime},D)\leq\xi\mid D\right\} -\left(1-\beta/2\right)\bigg\vert\\
 & =\bigg\vert P\left\{ \frac{a(\mathsf{R}_{\hat{g},k},D)-a(\mathsf{R}_{\hat{g}^{\prime},k}^{\prime},D)}{\sqrt{2v_{g^{\star},k}\left(D\right)}}\leq\frac{\xi}{\sqrt{2v_{g^{\star},k}\left(D\right)}}\mid D\right\} -\left(1-\beta/2\right)\bigg\vert\\
 & =\bigg\vert P\left\{ \frac{a(\mathsf{R}_{g^{\star},k},D)-a(\mathsf{R}_{g^{\star},k}^{\prime},D)}{\sqrt{2v_{g^{\star},k}\left(D\right)}}\leq\frac{\xi}{\sqrt{2v_{g^{\star},k}\left(D\right)}}-\frac{W\left(D\right)}{\sqrt{2v_{g^{\star},k}\left(D\right)}}\mid D\right\} -\left(1-\beta/2\right)\bigg\vert\\
 & \leq\sup_{q\in\left[-2\lambda,2\lambda\right]}\bigg\vert P\left\{ \frac{a(\mathsf{R}_{g^{\star},k},D)-a(\mathsf{R}_{g^{\star},k}^{\prime},D)}{\sqrt{2v_{g^{\star},k}\left(D\right)}}\leq z_{1-\beta/2}+q\right\} -\left(1-\beta/2\right)\bigg\vert\\
 &+\left(1-P\left\{ \mathcal{W}_{\lambda}\left(D\right)\mid D\right\} \right)\\
 & \leq \sup_{z \in \mathbb{R}}
\bigg\vert 
P\left\{ \frac{a(\mathsf{R}_{g^{\star},k},D)-a(\mathsf{R}_{g^{\star},k}^{\prime},D)}{\sqrt{2v_{g^{\star},k}\left(D\right)}} \leq z \mid D\right\}
-
\Phi(z)  \bigg\vert 
+2\lambda+\left(1-P\left\{ \mathcal{W}_{\lambda}\left(D\right)\mid D\right\} \right),
\end{flalign*}
as required.\hfill\qed

\subsection{Proof of \cref{lem: three components}}
\subsubsection{Part (i)}

Observe that we can write 
\[
a(\mathsf{R}_{g^{\star},k},D)-a(\mathsf{R}_{g^{\star},k}^{\prime},D)=\frac{1}{g^{\star}}\sum_{l=1}^{g}\left(\bar{a}(\mathsf{r}_{l},D)-\bar{a}(\mathsf{r}_{l}^{\prime},D)\right),
\]
where $\bar{a}\left(\cdot,D\right)$ is defined in \eqref{eq: centered}. We have that
\begin{flalign*}
\mathbb{E}\left[\vert\bar{a}(\mathsf{r}_{l},D)-\bar{a}(\mathsf{r}_{l}^{\prime},D)\vert^{3}\right] & \leq\left(\mathbb{E}\left[\left(\bar{a}(\mathsf{r}_{l},D)-\bar{a}(\mathsf{r}_{l}^{\prime},D)\right)^{4}\right]\right)^{3/4}\tag{Hölder}\\
 & \leq2^{6}\left(\mathbb{E}\left[\left(\bar{a}(\mathsf{r}_{l},D)\right)^{4}\right]\right)^{3/4}\\
 & \leq2^6 3^2 (2 - \varphi k - \varphi)^3 (\Gamma^{(2)}_{k,\varphi,b})^{3/4}~,
\end{flalign*}
where the second inequality follows from Hölder's inequality, the
binomial theorem, and the fact that $\mathsf{r}_{l}$ and $\mathsf{r}_{l}^{\prime}$
are exchangeable. The final inequality follows from \cref{cor: Moment bound}. Consequently,
we find that 
\begin{flalign*}
\frac{\mathbb{E}\left[\vert\bar{a}(\mathsf{r}_{l},D)-\bar{a}(\mathsf{r}_{l}^{\prime},D)\vert^{3} \mid D\right]}{\left(g^{\star}\right)^{1/2}\left(2v_{1,k}\left(D\right)\right)^{3/2}}
& \lesssim \frac{(2 - \varphi k - \varphi)^3 (\Gamma^{(2)}_{k,\varphi,b})^{3/4}}{\delta \left(g^{\star}\right)^{1/2}\left(v_{1,k}\left(D\right)\right)^{3/2}}\tag{Markov}\\
& \lesssim \frac{\xi}{z_{1-\beta/2}} \frac{(2 - \varphi k - \varphi)^3 (\Gamma^{(2)}_{k,\varphi,b})^{3/4}}{\delta \left(v_{1,k}\left(D\right)\right)^{2}}~,
\end{flalign*}
holds with probability $1-\delta$. The Lemma then follows by the Berry-Esseen inequality.\hfill\qed

\subsubsection{Part (ii)}

Our bound is based on the following two conditional concentration
inequalities. Both arguments are based on a Chernoff-type maximal
inequality, due to \citet{steiger1970bernstein}. 
\begin{lemma}
\label{lem: star hat diff}Suppose that \cref{assu: invariance,assu: linearity} hold,
 that the data $D$ are independent and identically distributed,
and that the eighth-order split $\zeta^{8}$ stability is finite. 

\noindent \textbf{(i)} For all $t>0$ and $c>0$, the condition concentration
inequality
\begin{flalign}
  & \log \frac{1}{2} P\left\{ \vert U(\mathsf{R}_{\hat{g},k},D)\vert\geq t\sqrt{v_{g^{\star},k}\left(D\right)},\vert\hat{g}/g^{\star}-1\vert\leq c\mid D\right\} \\
 & \quad \quad \quad \leq - \frac{\delta v_{1,k}(D)}{2^4 (2 - \varphi k - \varphi)^2 \Gamma^{(1)}_{k,\varphi,b}} \frac{t^{2}}{c} \label{eq: sum control}\
\end{flalign}
holds with probability greater than $1-\delta$ as $D$ varies.

\noindent \textbf{(ii)} If $g^{\star}c\geq2$ and $1/2>c>0$, then
the conditional concentration inequality 
\begin{flalign}
\log \frac{1}{8} P\left\{ \vert\hat{g}-g^{\star}\vert>cg^{\star}\mid D\right\}  
 \lesssim - \frac{\delta (v_{1,k}(D))^3}{(2 - \varphi k - \varphi)^4} \frac{z^2_{1-\beta/2}}{\Gamma_{k,\varphi,b}^{(2)}} \frac{c^2}{\xi^2} \label{eq: star hat diff}
\end{flalign}
holds with probability greater than $1-\delta$ as $D$ varies.
\end{lemma}
\noindent Observe that 
\begin{flalign}
 & P\left\{ \vert U(\mathsf{R}_{\hat{g},k},D)-U(\mathsf{R}_{\hat{g}^{\prime},k}^{\prime},D)\vert\geq\lambda\sqrt{2v_{g^{\star},k}\left(D\right)}\right\} \nonumber \\
 & \lesssim P\left\{ \vert U(\mathsf{R}_{\hat{g},k},D)\vert\geq\lambda\sqrt{2v_{g^{\star},k}\left(D\right)}\right\} \nonumber \\
 & \leq P\left\{ \vert U(\mathsf{R}_{\hat{g},k},D)\vert\geq\lambda\sqrt{2v_{g^{\star},k}\left(D\right)},\vert\hat{g}/g^{\star}-1\vert\leq c\mid D\right\} +P\left\{ \vert\hat{g}/g^{\star}-1\vert\geq c\mid D\right\} \nonumber \\
 & \lesssim 
 \exp\left(-\frac{\delta v_{1,k(D)}}{2^4(2 - \varphi k - \varphi)^2 \Gamma^{(1)}_{k,\varphi,b}} \frac{\lambda^{2}}{c} \right)
  +
  \exp\left(-C \frac{\delta (v_{1,k}(D))^3}{(2 - \varphi k - \varphi)^4} \frac{z^2_{1-\beta/2}}{\Gamma_{k,\varphi,b}^{(2)}} \frac{c^2}{\xi^2} \right)\label{eq: sum of two errors}
\end{flalign}
for some universal constant $C$. Hence, it remains to choose $c$ and $\lambda$ such that the
quantity (\ref{eq: sum of two errors}) is less than $\rho_{k,\varphi,b}\left(\xi,\beta\mid D\right)$.
First, we choose $c$ such that 
\begin{flalign*}
\frac{\delta (v_{1,k}(D))^3}{(2 - \varphi k - \varphi)^4} \frac{z^2_{1-\beta/2}}{\Gamma_{k,\varphi,b}^{(2)}} \frac{c^2}{\xi^2}
& \lesssim\log\left(\rho_{k,\varphi,b}\left(\xi,\beta\mid D\right)^{-1}\right)
\end{flalign*}
Choosing $c$ by 
\begin{equation}
c=\xi\left(\frac{\left(2-\varphi k-\varphi \right)^{2}(\Gamma_{k,\varphi,b}^{(2)})^{1/2}}{z_{1-\beta/2} \delta^{1/2} (v_{1,k}(D))^{3/2}}\right)
\log^{1/2}\left(\rho_{k,\varphi,b}\left(\xi,\beta\mid D\right)^{-1}\right)\label{eq: c to satisfy}
\end{equation}
suffices. Next, we choose $\lambda$ such that 
\[
\frac{\delta v_{1,k(D)}}{2^4(2 - \varphi k - \varphi)^2 \Gamma^{(1)}_{k,\varphi,b}} \frac{\lambda^{2}}{c}
\lesssim\log\left(\rho_{k,\varphi,b}\left(\xi,\beta\mid D\right)^{-1}\right).
\]
We can rewrite this condition by plugging in our choice of $c$ through
\begin{flalign*}
 & \frac{\lambda^{2}}{\xi}
 \frac{z_{1-\beta/2} \delta^{3/2}(v_{1,k}(D))^{5/2}}{\left(2-\varphi k-\varphi\right)^{4} \Gamma^{(1)}_{k,\varphi,b} (\Gamma_{k,\varphi,b}^{(2)})^{1/2}}
 \lesssim \log^{3/2}\left(\rho_{k,\varphi,b}\left(\xi,\beta\mid D\right)^{-1}\right)
\end{flalign*}
Choosing $\lambda$ by 
\begin{flalign*}
\lambda & =
\left(
\frac{\xi}{z_{1-\beta/2}}
\frac{(2-\varphi k-\varphi)^{4}\Gamma^{(1)}_{k,\varphi,b} (\Gamma_{k,\varphi,b}^{(2)})^{1/2}}{\delta^{3/2}(v_{1,k}\left(D\right))^{5/2}}
\right)^{1/2}
\tilde{\lambda}_{k,\varphi,b}\left(D\right),\quad\text{where}\\
\tilde{\lambda}_{k,\varphi,b}\left(D\right) & =\log^{3/4}\left(\rho_{k,\varphi,b}\left(\xi,\beta\mid D\right)^{-1}\right)
\end{flalign*}
will then suffice, as required.\hfill\qed

\subsubsection{Part(iii)}

Observe that 
\begin{flalign}
 & P\left\{ \vert\left(1-\frac{\hat{g}}{g^{\star}}\right)\left(a(\mathsf{R}_{\hat{g},k},D)-\bar{a}\left(D\right)\right)\vert\geq\lambda\sqrt{2v_{g^{\star},k}\left(D\right)},\vert\begin{gathered}g^{\star}-\hat{g}\vert\leq cg^{\star}\end{gathered}
\mid D\right\} \nonumber \\
 & =P\left\{ \vert\left(a(\mathsf{R}_{\hat{g},k},D)-\bar{a}\left(D\right)\right)\vert\geq\frac{\lambda}{c}\sqrt{2v_{g^{\star},k}\left(D\right)},\vert\begin{gathered}g^{\star}-\hat{g}\vert\leq cg^{\star}\end{gathered}
\mid D\right\} \nonumber \\
 & \leq P\left\{ \max_{\vert g^{\star}-g \vert\leq cg^{\star}} 
 \vert\left(a(\mathsf{R}_{g,k},D)-\bar{a}\left(D\right)\right)\vert\geq\frac{\lambda}{c}\sqrt{2v_{g^{\star},k}\left(D\right)}\mid D\right\} \nonumber \\
 & \leq P\left\{ \max_{\vert g^{\star}-g \vert\leq cg^{\star}}
 \vert\left(a(\mathsf{R}_{g,k},D)-\bar{a}\left(D\right)\right)\vert\geq\frac{\lambda}{c}\left(\frac{\xi}{z_{1-\beta/2}}\right)\mid D\right\} \nonumber \\
 & \leq\sum_{\vert g^{\star}-g \vert\leq cg^{\star}}
 P\left\{ \vert\left(a(\mathsf{R}_{g,k},D)-\bar{a}\left(D\right)\right)\vert\geq\frac{\lambda}{c}\left(\frac{\xi}{z_{1-\beta/2}}\right)\right\}~. \label{eq: union bound}
\end{flalign}
By \cref{thm: cross split concentration appendix}, we have that 
\begin{flalign*}
 & P\left\{ \vert\left(a(\mathsf{R}_{g,k},D)-\bar{a}\left(D\right)\right)\vert\geq\frac{\lambda}{c}\left(\frac{\xi}{z_{1-\beta/2}}\right)\right\} \\
 & \leq 2 \exp\left(-\frac{g \delta}{2^4(2-\varphi k - \varphi)^2\Gamma_{k,\varphi,b}(1)} \frac{\lambda^{2}}{c^{2}}\left(\frac{\xi}{z_{1-\beta/2}}\right)^{2}\right)
\end{flalign*}
and so is (\ref{eq: union bound}) bounded from above by 
\begin{flalign*}
  &4cg^{\star}\exp\left(-\frac{g^\star (1-c) \delta}{4(2-\varphi k - \varphi)^2\Gamma_{k,\varphi,b}(1)} \frac{\lambda^{2}}{c^{2}}\left(\frac{\xi}{z_{1-\beta/2}}\right)^{2}\right) \\
  &\lesssim v_{1,k}\left(D\right)c\left(\frac{z_{1-\beta/2}}{\xi}\right)^{2}
 \exp\left(-\frac{\delta v_{1,k}(D)}{2^5(2-\varphi k - \varphi)^2\Gamma_{k,\varphi,b}(1)} \frac{\lambda^{2}}{c^{2}}\right)
\end{flalign*}
if $1-c\geq1/2$ by the definition of $g^{\star}$. Thus, we have
that 
\begin{flalign}
 & P\left\{ \vert Q(\mathsf{R}_{\hat{g},k},D)-Q(\mathsf{R}_{\hat{g}^{\prime},k}^{\prime},D)\vert\geq\lambda\sqrt{2v_{g^{\star},k}\left(D\right)}\mid D\right\} \nonumber \\
 & \lesssim P\left\{ \vert\left(1-\frac{\hat{g}}{g^{\star}}\right)\left(a(\mathsf{R}_{\hat{g},k},D)-\bar{a}\left(D\right)\right)\vert\geq\lambda\sqrt{2v_{g^{\star},k}\left(D\right)},\vert\begin{gathered}g^{\star}-\hat{g}\vert\leq cg^{\star}\end{gathered}
\mid D\right\} \nonumber \\
& +P\left\{ \vert\begin{gathered}g^{\star}-\hat{g}\vert\geq cg^{\star}\end{gathered}
\mid D\right\} \nonumber \\
 & \lesssim v_{1,k}\left(D\right)c\left(\frac{z_{1-\beta/2}}{\xi}\right)^{2}
 \exp\left(-\frac{\delta v_{1,k}(D)}{2^5(2-\varphi k - \varphi)^2\Gamma_{k,\varphi,b}^{(1)}} \frac{\lambda^{2}}{c^{2}}\right) \label{eq: union bound super poly}\\
 & \quad + \exp\left(-\frac{1}{C}\frac{\delta (v_{1,k}(D))^3}{(2 - \varphi k - \varphi)^4} \frac{z_{1-\beta/2}}{\Gamma_{k,\varphi,b}^{(2)}} \frac{c^2}{\xi^2}\right)\label{eq: Q c error}
\end{flalign}
for some universal constant $C$ by Lemma \ref{lem: star hat diff},
Part (ii). If $c$ is chosen to satisfy (\ref{eq: c to satisfy}),
the term (\ref{eq: Q c error}) will be bounded above by $\rho_{\varphi,k}\left(\xi,\beta\mid D\right)$ for all sufficiently small $\xi$. By plugging this value of $c$ and $\lambda=\lambda_{\varphi,k}\left(\xi,\beta\mid D\right)$
into (\ref{eq: union bound super poly}), we obtain the term
\begin{align}
&
\left(\frac{z_{1-\beta/2}^2}{\xi}\right)
\left(
\frac{\left(2-\varphi k-\varphi \right)^{2}(\Gamma_{k,\varphi,b}^{(2)})^{1/2}}{z_{1-\beta/2} \delta^{1/2} (v_{1,k}(D))^{1/2}}\right)
\log^{1/2}\left(\rho_{k,\varphi,b}\left(\xi,\beta\mid D\right)^{-1}
\right)\nonumber\\
& \exp
\left(
- \frac{1}{\xi} 
\frac{
\delta^{1/2} (v_{1,k}(D))^{3/2}}{(2-\varphi k - \varphi)^2(\Gamma_{k,\varphi,b}^{(2)})^{1/2}} 
\tilde{\lambda}_{k,\varphi,b}^{U}\left(D\right)
\log\left(\rho_{k,\varphi,b}\left(\xi,\beta\mid D\right)^{-1}\right) 
\right)\label{eq: sub polynomial}
\end{align}
as required. Observe that the term \eqref{eq: sub polynomial} is bounded above by  $\rho_{\varphi,k}\left(\xi,\beta\mid D\right)$, for all sufficiently small $\xi$, completing the proof. \hfill\qed

\subsection{Proofs for Supporting Lemmas\label{app: auxiliary}}

\subsubsection{Proof of \cref{lem: finiteness}}
Let $T_{i}\left(\pi_{m,l}\right)=T\left(\mathsf{s}_{i}\left(\pi_{l,m}\right),D\right)$
and define $T_{i}(\pi_{m,l}^{\prime})$ analogously. Define 
\[
R_{m}\left(l\right)=\left\{ i\in\left[n\right]:\pi_{m,l}\left(i\right)\neq\pi_{m,l}^{\prime}\left(i\right)\right\} 
\]
and let $\bar{R}_{m}\left(l\right)=\vert R_{m}\left(l\right)\vert$
denote the number of shared indices in the permutations in $\pi_{m,l}$
and $\pi_{m,l}^{\prime}$. Observe that
\begin{equation}
\frac{1}{k}\sum_{i=1}^{k}\left(T_{i}\left(\pi_{m,l}\right)-T_{i}\left(\pi_{m,l}^{\prime}\right)\right)=0\label{eq: zero condition}
\end{equation}
almost surely for all $m^{\prime}\geq m$ if and only if $\bar{R}_{m}\left(l\right)=n$.
Let $N\left(l\right)$ denote the values of the smallest index $m$
with $\bar{R}_{m}\left(l\right)=n$. Observe that 
\begin{flalign}
 & \sum_{m=0}^{\infty}\bigg\vert
 \mathbb{E}\left[a(\mathsf{R}_{g,k}(\bm{\pi}_m), D)-a(\mathsf{R}_{g,k}(\bm{\pi}_m^\prime), D)\mid 
 \bm{\pi}_{0}=\bm{\psi},\bm{\pi}_{0}^{\prime}=\bm{\psi}^{\prime},D\right]\bigg\vert\label{eq: kernel sum}\\
 & \leq\sum_{m=0}^{\infty}\frac{1}{g}\frac{1}{k}\sum_{l=1}^{g}\sum_{i=1}^{k}
 \mathbb{E}\left[\bigg\vert\left(T_{i}\left(\pi_{m,l}\right)-T_{i}\left(\pi_{m,l}^{\prime}\right)\right)\bigg\vert
 \mid  \bm{\pi}_{0}=\bm{\psi},\bm{\pi}_{0}^{\prime}=\bm{\psi}^{\prime},D\right]\nonumber \\
 & \leq\frac{1}{g}\frac{1}{k}\sum_{l=1}^{g}\sum_{i=1}^{k}
 \max_{\mathsf{s},\mathsf{s}^{\prime}\in\mathsf{S}_{n,b}}
 \big\vert T\left(\mathsf{s},D\right)-T\left(\mathsf{s}^{\prime},D\right)\big\vert\mathbb{E}\left[N\left(l\right)\right]\nonumber 
\end{flalign}
Hence, it suffices to bound the quantity $\mathbb{E}\left[N\left(l\right)\right]$.

To that end, let $N_{r}\left(l\right)$ be the value of the smallest
index $m$ with $\bar{R}_{m}\left(l\right)\geq r$. We proceed analogously
to standard analysis of the coupon collector's problem (see e.g.,
Section 2.2 of \citealp{levin2017markov}). We can evaluate
\begin{flalign*}
& P\left\{ \bar{R}_{m}\left(l\right)\geq r+1\mid\bar{R}_{m-1}\left(l\right)=r\right\}  \\
& =\frac{1}{g}P\left\{ \pi_{l}\left(I_{m}\right)\neq\pi_{l}^{\prime}\left(I_{m}\right),\pi_{l}^{-1}\left(J_{m}\right)\neq\pi_{l}^{\prime-1}\left(J_{m}\right)\right\} =\frac{\left(n-r\right)^{2}}{gn^{2}}
\end{flalign*}
and
\begin{flalign*}
P\left\{ \bar{R}_{m}\left(l\right)<r\mid\bar{R}_{m-1}\left(l\right)=r\right\}  & =0,
\end{flalign*}
and thereby obtain the bound
\begin{flalign}
\mathbb{E}\left[N\left(l\right)\right] & 
=\sum_{r=1}^{n}\mathbb{E}\left[N_{r}\left(l\right)-N_{r-1}\left(l\right)\right]\leq\sum_{r=1}^{n}\frac{gn^{2}}{\left(n-r+1\right)^{2}}=gn^{2}\sum_{r=1}^{n}\frac{1}{r^{2}}\leq2gn^{2}.\label{eq: expectation N(l) bound}
\end{flalign}
Hence, we find that (\ref{eq: kernel sum}) is upper bounded by 
\[
2gn^{2}\max_{\mathsf{s},\mathsf{s}^{\prime}\in\mathsf{S}_{n,b}}T\left(\mathsf{s},D\right)-T\left(\mathsf{s}^{\prime},D\right),
\]
as required.\hfill$\qed$

\subsubsection{Proof of \cref{lem: exp mv ineq}\label{subsec: proof of exp mv ineq}}

Observe that 
\begin{flalign*}
\frac{\text{d}}{\text{d}t}e^{tx}e^{(1-t)y} & =\left(x-y\right)e^{tx}e^{(1-t)y}.
\end{flalign*}
Thus, we find that 
\begin{flalign*}
c\left(e^{x}-e^{y}\right) & =c\int_{0}^{1}\left(\frac{\text{d}}{\text{d}t}e^{tx}e^{\left(1-t\right)y}\right)\text{dt}\tag{Fundamental Theorem of Calculus}\\
 & =c\left(x-y\right)\int_{0}^{1}e^{tx}e^{(1-t)y}\text{dt}\\
 & \leq c\left(x-y\right)\int_{0}^{1}\left(te^{x}+\left(1-t\right)e^{y}\right)\tag{Convexity}\\
 & =\frac{c}{2}\left(x-y\right)\left(e^{x}+e^{y}\right)\\
 & =\frac{1}{2}\left(\left(s^{-1}c^{2}\left(e^{x}+e^{y}\right)\right)\left(s\left(x-y\right)^{2}\left(e^{x}+e^{y}\right)\right)\right)^{1/2}\\
 & \leq\frac{1}{4}\left(\left(s^{-1}c^{2}\right)+s\left(x-y\right)^{2}\right)\left(e^{x}+e^{y}\right),\tag{AM-GM}
\end{flalign*}
as required. \hfill\qed

\subsubsection{Proof of \cref{lem: sequence bound}}

To begin, consider any collection of real numbers $x_{1},\ldots,x_{2^{c}}$,
for some positive integer $c$. For any integer $s>0$, we have
\begin{flalign*}
\left(x_{1}\cdots x_{2^{c}}\right)^{2s} & \leq\frac{1}{4}\left(\left(x_{i_{1}}\cdots x_{i_{2^{c-1}}}\right)^{2s}+\left(x_{i_{2^{c-1}+1}}\cdots x_{2^{c}}\right)^{2s}\right)^{2}\tag{Young's Inequality}\\
 & \leq\frac{1}{4}\left(\left(x_{i_{1}}\cdots x_{i_{2^{c-1}}}\right)^{2(s+1)}+2\left(x_{1}\cdots x_{2^{c}}\right)^{2s}+\left(x_{i_{2^{c-1}+1}}\cdots x_{2^{c}}\right)^{2(s+1)}\right)
\end{flalign*}
and so 
\begin{equation}
\left(x_{1}\cdots x_{2^{c}}\right)^{2s}\leq\frac{1}{2}\left(x_{i_{1}}\cdots x_{i_{2^{c-1}}}\right)^{2(s+1)}+\frac{1}{2}\left(x_{i_{2^{c-1}+1}}\cdots x_{2^{c}}\right)^{2(s+1)}.\label{eq: quadratic young}
\end{equation}
Consequently, we have that 
\begin{equation}
x_{1}\cdots x_{2^{c}}\leq\frac{1}{2^{c}}\sum_{i=1}^{2^{c}}x_{i}^{2c}\label{eq: generalized product sum ineq}
\end{equation}
through $2^{r}$ applications of (\ref{eq: quadratic young}). 

Now, to prove the Lemma, we may assume without loss that $h_{i}>0$
for all $i$ by continuity. Observe that
\begin{equation}
\mathbb{E}\left[\left(\sum_{m=0}^{\infty}X_{m}\right)^{2^{c}}\right]=\sum_{i_{1}=0}^{\infty}\cdots\sum_{i_{2^{c}}=0}^{\infty}\mathbb{E}\left[X_{i_{1}}\cdots X_{i_{2^{c}}}\right]\label{eq: apply dom conv}
\end{equation}
by dominated convergence. By writing 
\[
X_{1}\cdots X_{2^{c}}=\prod_{j=1}^{2^{c}}\left(\frac{\prod_{k\neq j}h_{k}}{h_{j}^{2^{c-1}}}\right)^{1/2^{c}}X_{j}
\]
we find that 
\[
\mathbb{E}\left[X_{1}\cdots X_{2^{c}}\right]\leq\frac{1}{2^{c}}\sum_{j=1}^{2^{c}}\frac{\prod_{k\neq j}h_{k}}{h_{j}^{2^{c-1}}}h_{j}^{2^{c}}=\prod_{j=1}^{2^{c}}h_{j}
\]
by (\ref{eq: generalized product sum ineq}). Hence, by (\ref{eq: apply dom conv}),
we have that 
\[
\mathbb{E}\left[\left(\sum_{m=0}^{\infty}X_{m}\right)^{2^{c}}\right]\leq\sum_{i_{1}=0}^{\infty}\cdots\sum_{i_{2^{c}}=0}^{\infty}\prod_{j=1}^{2^{c}}h_{i_{j}}=\left(\sum_{m=0}^{\infty}h_{j}\right)^{2^{c}},
\]
as required.\hfill\qed

\subsubsection{Proof of Lemma \ref{lem: a(Z) sq concentration}}

We begin by constructing Stein representers $\tilde{V}\left(Z,Z^{\prime}\right)$
and $\check{V}\left(Z,Z^{\prime}\right)$ for the statistic $\tilde{v}_{g,k}\left(Z\right)$
and $\check{v}_{g,k}\left(Z\right)$, respectively. We will use the same
exchangeable pair $\left(Z,Z^{\prime}\right)_{m\geq0}$ and Markov
chain $\left(Z_{m},Z_{m}^{\prime}\right)_{m\geq1}$ defined in Section
\ref{sec: construction}. The subsequent result follows from an
argument similar to Lemma \ref{lem: representer construction applied},
again based on an idea expressed by Lemma 4.1 of \citet{chatterjee2005concentration}. 
\begin{lemma}
\label{lem: var representer}$\text{ }$

\noindent \textbf{(i)} Let $\bm{\psi}$ and $\psi^{\prime}$ be two
sets, each containing $g$ elements of $\mathcal{P}_{n}$.
The inequalities
\begin{flalign*}
\sum_{m=0}^{\infty}\bigg\vert
\mathbb{E}\left[\left(\tilde{v}_{g,k}\left(Z_{m}\right)-\tilde{v}_{g,k}\left(Z_{m}^{\prime}\right)\right)\mid Z_{0}=Z,Z_{0}^{\prime}=Z^{\prime}\right]\bigg\vert & \leq n^{2}\max_{\mathsf{s},\mathsf{s}^{\prime}\in\mathsf{S}_{n,b}}\left(T(\mathsf{s},D)-T(\mathsf{s}^{\prime},D)\right)^{2}\quad\text{and}\\
\sum_{m=0}^{\infty}\bigg\vert\mathbb{E}\left[\left(\check{v}_{g,k}\left(Z_{m}\right)-\check{v}_{g,k}\left(Z_{m}^{\prime}\right)\right)\mid Z_{0}=Z,Z_{0}^{\prime}=Z^{\prime}\right]\bigg\vert 
& \leq2gn^{2}\max_{\mathsf{s}\in\mathsf{S}_{n,b}}\left(T(\mathsf{s},D)-T(\mathsf{s}^{\prime},D)\right)^{2}
\end{flalign*}
hold almost surely. 

\noindent \textbf{(ii)} The functions
\begin{flalign*}
\tilde{V}\left(Z,Z^{\prime}\right) & =\sum_{m=0}^{\infty}\mathbb{E}\left[\left(\tilde{v}_{g,k}\left(Z_{m}\right)-\tilde{v}_{g,k}\left(Z_{m}^{\prime}\right)\right)\mid Z_{0}=Z,Z_{0}^{\prime}=Z^{\prime}\right]\quad\text{and}\\
\check{V}\left(Z,Z^{\prime}\right) & =\sum_{m=0}^{\infty}\mathbb{E}\left[\left(\check{v}_{g,k}\left(Z_{m}\right)-\check{v}_{g,k}\left(Z_{m}^{\prime}\right)\right)\mid Z_{0}=Z,Z_{0}^{\prime}=Z^{\prime}\right]
\end{flalign*}
are finite and satisfy the equalities 
\begin{flalign*}
\mathbb{E}\left[\tilde{V}\left(Z,Z^{\prime}\right)\mid Z\right] & =\tilde{v}_{g,k}\left(Z\right)-v_{g,k}\left(D\right)\quad\text{and}\\
\mathbb{E}\left[\check{V}\left(Z,Z^{\prime}\right)\mid Z\right] & =\check{v}_{g,k}\left(Z\right)-v_{g,k}\left(D\right)
\end{flalign*}
almost surely. 
\end{lemma}
\noindent
Next, we apply the concentration inequality stated in Theorem \ref{thm: chatterjee concentration},
due to \citet{chatterjee2005concentration,chatterjee2007stein}. That
is, to establish part (i) of the Lemma, it will suffice to characterize
constants $s$ and $u$ that satisfy 
\[
U_{\tilde{v}}\left(Z\right)\leq s^{-1}u\quad\text{and}\quad U_{\tilde{V}}\left(Z\right)\leq su,
\]
with probability $1-\delta$, where 
\[
U_{\tilde{v}}\left(Z\right)=\frac{1}{2}\mathbb{E}\left[\left(\tilde{v}\left(Z\right)-\tilde{v}\left(Z^{\prime}\right)\right)\mid Z\right]\quad\text{and}\quad U_{\tilde{V}}\left(Z\right)=\frac{1}{2}\mathbb{E}\left[\tilde{V}\left(Z,Z^{\prime}\right)^{2}\mid Z\right].
\]
The same statement holds for part (ii) of the Lemma, for the objects
$U_{\check{v}}\left(Z\right)$ and $U_{\check{V}}\left(Z\right)$
defined analogously. 

We obtain such a characterization through the application of Lemma
\ref{lem: variance bounds}. To this end, observe that, by Lemma \ref{lem: var representer},
part (i), the bound 
\begin{align}
& \left(\sum_{m=0}^{\infty}\mathbb{E}\left[\left(\tilde{v}_{g,k}\left(Z_{m}\right)-\tilde{v}_{g,k}\left(Z_{m}^{\prime}\right)\right)\mid Z_{0}=Z,Z_{0}^{\prime}=Z^{\prime}\right]\right)^{2} \\
& \leq n^{4}\max_{\mathsf{s},\mathsf{s}^{\prime}\in\mathsf{S}_{n,b}}\left(T(\mathsf{s},D)-T(\mathsf{s}^{\prime},D)\right)^{4}\label{eq: var integrability bound}
\end{align}
holds almost surely. The right hand side of (\ref{eq: var integrability bound})
is square integrable, as the eighth-order split-stability $\zeta^{(8)}$
is finite almost surely. An analogous statement holds for the statistic
$\check{v}_{g,k}\left(Z\right)$. Thus, deterministic bounds of the form
(\ref{eq: c_i general bound}) can be obtained through the application
of the following Lemma. 
\begin{lemma}
\label{lem: var deterministic bound}Suppose that \cref{assu: invariance,assu: linearity} hold
 and that the data $D$ are independent and identically distributed. If the eighth-order split stability $\zeta^{(8)}$ is finite,
then 

\noindent \textbf{(i)} The inequality 
\begin{flalign*}
 & \mathbb{E}\left[\mathbb{E}\left[\tilde{v}_{g,k}(Z_{m})-\tilde{v}_{g,k}(Z_{m}^{\prime})\mid Z,Z^{\prime}\right]^{2}\mid\bm{\pi}\right]
 \lesssim
\left(1-\frac{2}{gn^{2}}\right)^{2m}\left(\frac{\left(2-\varphi k-\varphi\right)^4}{n^2 g^4}\right) \Gamma_{k,\varphi,b}^{(2)}~,
\end{flalign*}
holds almost surely for all integers $m\geq0$, and 

\noindent \textbf{(ii)} The inequality 
\begin{flalign*}
 & \mathbb{E}\left[\mathbb{E}\left[\check{v}_{g,k}(Z_{m})-\check{v}_{g,k}(Z_{m}^{\prime})\mid Z,Z^{\prime}\right]^{2}\mid\bm{\pi}\right]
 \lesssim
\left(1-\frac{2}{gn^{2}}\right)^{2m}\left(\frac{\left(2-\varphi k-\varphi\right)^4}{n^2 g^3}\right) \Gamma_{k,\varphi,b}^{(2)}~,
\end{flalign*}
holds almost surely for all integers $m\geq0$.
\end{lemma}
Thus, by Lemma \ref{lem: variance bounds}, part (i), the inequalities
\begin{flalign*}
U_{\tilde{V}}\left(Z\right) &  \lesssim
\frac{1}{\delta}
\left(\frac{gn^{2}}{2}\right)^{2} 
\left(\frac{\left(2-\varphi k-\varphi\right)^4}{ n^2 g^4}\right) \Gamma_{k,\varphi,b}^{(2)} \\
& \lesssim
\left(\frac{gn^{2}}{2}\right)        \left( \frac{1}{\delta} \frac{(2-\varphi k-\varphi)^4}{g^3} \right) \Gamma_{k,\varphi,b}^{(2)}
\end{flalign*}
and 
\begin{flalign*}
U_{\tilde{v}}\left(Z\right) &  \lesssim
\frac{1}{\delta} 
\left(\frac{\left(2-\varphi k-\varphi\right)^4}{ n^2 g^4}\right) \Gamma_{k,\varphi,b}^{(2)} \\
& \lesssim 
\left(\frac{2}{gn^{2}}\right) 
 \left(\frac{1}{\delta} \frac{\left(2-\varphi k-\varphi\right)^4}{g^3}\right) \Gamma_{k,\varphi,b}^{(2)}
\end{flalign*}
both hold with probability greater than $1-\delta$. Hence, we obtain
the bound (\ref{eq: var tilde concentration}) by applying Theorem
\ref{thm: chatterjee concentration} and choosing $s=gn^{2}/2$ and
\[
u=\frac{1}{\delta} \frac{\left(2-\varphi k-\varphi\right)^4}{g^3} \Gamma_{k,\varphi,b}^{(2)}~.
\]
Analogous inequalities for the objects $U_{\check{v}}\left(Z\right)$
and $U_{\check{V}}\left(Z\right)$ hold by Lemma \ref{lem: variance bounds},
part (ii), where in that case 
\[
u=\frac{1}{\delta} \frac{\left(2-\varphi k-\varphi\right)^4}{g^2} n^4\Gamma_{k,\varphi,b}^{(2)}~.
\]
which similarly implies the bound (\ref{eq: a(Z) sq concentration}) by Theorem \ref{thm: chatterjee concentration}.\hfill\qed

\subsubsection{Proof of Lemma \ref{lem: var representer}}
Reinstate the notation of the proof of Lemma \ref{lem: finiteness}. We begin by noting that 
\begin{flalign*}
\max_{\mathsf{s},\mathsf{s}^{\prime}\in\mathsf{S}_{n,b}}\left(T(\mathsf{s},D)-T(\mathsf{s}^{\prime},D)\right)^{2} & =\max_{\mathsf{s},\mathsf{s}^{\prime}\in\mathsf{S}_{n,b}}\left(T(\mathsf{s},D)-\bar{a}\left(D\right)\right)^{2}+\left(T(\mathsf{s}^{\prime},D)-\bar{a}\left(D\right)\right)^{2}\\
 & \quad\quad\quad\quad\quad+2\left(T(\mathsf{s},D)-\bar{a}\left(D\right)\right)\left(T(\mathsf{s}^{\prime},D)-\bar{a}\left(D\right)\right)\\
 & \geq4\max_{\mathsf{s}\in\mathsf{S}_{n,b}}\left(T(\mathsf{s},D)-\bar{a}\left(D\right)\right)^{2}.
\end{flalign*}
Observe that 
\begin{flalign*}
 & \sum_{i,i^{\prime}=1}^{k}\left(T_{i}\left(\pi_{m,l}\right)-\bar{a}\left(D\right)\right)\left(T_{i}\left(\pi_{m,l}\right)-\bar{a}\left(D\right)\right)-
 \sum_{i,i^{\prime}=1}^{k}\left(T_{i}\left(\pi_{m,l}^{\prime}\right)-\bar{a}\left(D\right)\right)\left(T_{i^{\prime}}\left(\pi_{m,l}^{\prime}\right)-\bar{a}\left(D\right)\right)=0
\end{flalign*}
for all for all $m^{\prime}\geq m$ if and only if $\bar{R}_{m}\left(l\right)=n$.
Thus, we have that 
\begin{flalign*}
 & \sum_{m=0}^{\infty}\bigg\vert\mathbb{E}\left[\tilde{v}_{g,k}(Z_{m})-\tilde{v}_{g,k}(Z_{m}^{\prime})\mid Z_{0}=\left(\mathsf{r}\left(\bm{\psi}\right),D\right),Z_{0}^{\prime}=\left(\mathsf{r}\left(\bm{\psi}^{\prime}\right),D\right)\right]\bigg\vert\\
 & \leq\sum_{m=0}^{\infty}\frac{1}{g^{2}k^{2}}\sum_{l=1}^{g}\sum_{i,i^{\prime}=1}^{k}\mathbb{E}\bigg[\bigg\vert\left(T_{i}\left(\pi_{m,l}\right)-\bar{a}\left(D\right)\right)\left(T_{i^{\prime}}\left(\pi_{m,l}\right)-\bar{a}\left(D\right)\right)\\
 & \quad\quad\quad\quad\quad-\left(T_{i}\left(\pi_{m,l}^{\prime}\right)-\bar{a}\left(D\right)\right)\left(T_{i^{\prime}}\left(\pi_{m,l}^{\prime}\right)-\bar{a}\left(D\right)\right)\bigg\vert\mid Z_{0}=\left(\mathsf{r}\left(\bm{\psi}\right),D\right),Z_{0}^{\prime}=\left(\mathsf{r}\left(\bm{\psi}^{\prime}\right),D\right)\bigg]\\
 & \leq\frac{2}{g^{2}k^{2}}\sum_{l=1}^{g}\sum_{i,i^{\prime}=1}^{k}\max_{\mathsf{s},\mathsf{s}^{\prime}\in\mathsf{S}_{n,b}}\big\vert\left(T\left(\mathsf{s},D\right)-\bar{a}\left(D\right)\right)\left(T\left(\mathsf{s}^{\prime},D\right)-\bar{a}\left(D\right)\right)\big\vert\mathbb{E}\left[N\left(l\right)\right].\\
 & \leq\frac{2}{g}\max_{\mathsf{s}\in\mathsf{S}_{n,b}}\left(T(\mathsf{s},D)-\bar{a}\left(D\right)\right)^{2}\mathbb{E}\left[N\left(l\right)\right].\\
 & \leq\frac{1}{2g}\max_{\mathsf{s},\mathsf{s}^{\prime}\in\mathsf{S}_{n,b}}\left(T(\mathsf{s},D)-T(\mathsf{s}^{\prime},D)\right)^{2}\mathbb{E}\left[N\left(l\right)\right].
\end{flalign*}
Similarly, we have that 
\begin{flalign*}
 & \sum_{m=0}^{\infty}\bigg\vert\mathbb{E}\left[\check{v}_{g,k}(Z_{m})-\check{v}_{g,k}(Z_{m}^{\prime})\mid Z_{0}=\left(\mathsf{r}\left(\bm{\psi}\right),D\right),Z_{0}^{\prime}=\left(\mathsf{r}\left(\bm{\psi}^{\prime}\right),D\right)\right]\bigg\vert\\
 & \leq\sum_{m=0}^{\infty}\frac{1}{g^{2}k^{2}}\sum_{l,l^{\prime}=1}^{g}\sum_{i,i^{\prime}=1}^{k}\mathbb{E}\bigg[\bigg\vert\left(T_{i}\left(\pi_{m,l}\right)-\bar{a}\left(D\right)\right)\left(T_{i^{\prime}}\left(\pi_{m,l^{\prime}}\right)-\bar{a}\left(D\right)\right)\\
 & \quad\quad\quad\quad\quad-\left(T_{i}\left(\pi_{m,l}^{\prime}\right)-\bar{a}\left(D\right)\right)\left(T_{i^{\prime}}\left(\pi_{m,l^{\prime}}^{\prime}\right)-\bar{a}\left(D\right)\right)\bigg\vert\mid Z_{0}=\left(\mathsf{r}\left(\bm{\psi}\right),D\right),Z_{0}^{\prime}=\left(\mathsf{r}\left(\bm{\psi}^{\prime}\right),D\right)\bigg]\\
 & \leq\frac{2}{g^{2}k^{2}}\sum_{l,l^{\prime}=1}^{g}\sum_{i,i^{\prime}=1}^{k}\max_{\mathsf{s},\mathsf{s}^{\prime}\in\mathsf{S}_{n,b}}\big\vert\left(T\left(\mathsf{s},D\right)-\bar{a}\left(D\right)\right)\left(T\left(\mathsf{s}^{\prime},D\right)-\bar{a}\left(D\right)\right)\big\vert\mathbb{E}\left[\max\left\{ N\left(l\right),N\left(l^{\prime}\right)\right\} \right]\\
 & \leq\frac{2}{g^{2}k^{2}}\sum_{l,l^{\prime}=1}^{g}\sum_{i,i^{\prime}=1}^{k}\max_{\mathsf{s},\mathsf{s}^{\prime}\in\mathsf{S}_{n,b}}\big\vert\left(T\left(\mathsf{s},D\right)-\bar{a}\left(D\right)\right)\left(T\left(\mathsf{s}^{\prime},D\right)-\bar{a}\left(D\right)\right)\big\vert\mathbb{E}\left[N\left(l\right)+N\left(l^{\prime}\right)\right]\\
 & \leq\frac{1}{2}\max_{\mathsf{s},\mathsf{s}^{\prime}\in\mathsf{S}_{n,b}}\left(T(\mathsf{s},D)-T(\mathsf{s}^{\prime},D)\right)^{2}\big\vert\mathbb{E}\left[N\left(l\right)+N\left(l^{\prime}\right)\right]
\end{flalign*}
Thus, part (i) of the Lemma follows by the bound (\ref{eq: expectation N(l) bound}).
Part (ii) of the Lemma then follows by an argument analogous to the
proof of Lemma \ref{lem: representer construction applied}. \hfill\qed

\subsubsection{Proof of Lemma \ref{lem: var deterministic bound}}
We reinstate the notation introduced in the proof of Lemma \ref{lem: deterministic bound}
from Section \ref{subsec: proof of deterministic bound}. First, observe
that 
\begin{flalign}
 & \mathbb{E}\left[\mathbb{E}\left[\tilde{v}_{g,k}(Z_{m})-\tilde{v}_{g,k}(Z_{m}^{\prime})\mid Z,Z^{\prime}\right]^{2}\mid\bm{\pi}\right]\nonumber \\
 & =\mathbb{E}\left[\left(P\left\{ \mathcal{H}_{m}\right\} \mathbb{E}\left[\tilde{v}_{g,k}(Z_{m})-\tilde{v}_{g,k}(Z_{m}^{\prime})\mid\mathcal{H}_{m},Z,Z^{\prime}\right]\right)^{2}\mid\bm{\pi}\right]\nonumber \\
 & \leq\left(1-\frac{2}{gn^{2}}\right)^{2m}\left(\frac{2kb\left(2n-bk-b\right)}{n^{2}}\right)^{2}
 \mathbb{E}\left[\left(\tilde{v}_{g,k}(Z_{m})-\tilde{v}_{g,k}(Z_{m}^{\prime})\right)^{2}\mid\mathcal{H}_{m},\bm{\pi}\right]\label{eq: var conc cond bound}
\end{flalign}
as before. Recall the notation
\[
\bar{a}(\mathsf{r}_l,D) = \frac{1}{k} \sum_{i=1}^k (T(\mathsf{s}_{l,i},D) - \bar{a}(D))
\]
and observe that
\begin{flalign}
 & \mathbb{E}\left[\left(\tilde{v}_{g,k}(Z_{m})-\tilde{v}_{g,k}(Z_{m}^{\prime})\right)^{2}
 \mid\mathcal{H}_{m},\bm{\pi}\right]\nonumber\\
 & = \frac{1}{g^4} \mathbb{E}\left[
\left(\bar{a}(\mathsf{r}(\pi_{m,L}),D)^2 - 
       \bar{a}(\mathsf{r}(\pi^\prime_{m,L}),D)^2 \right)^2 \mid \mathcal{H}_m \right] \nonumber\\
& = \frac{1}{g^4} \mathbb{E}\left[ 
\left(\bar{a}(\mathsf{r}(\pi_{m,L}),D) + \bar{a}(\mathsf{r}(\pi^\prime_{m,L}),D)\right)^2
\left(\bar{a}(\mathsf{r}(\pi_{m,L}),D) - \bar{a}(\mathsf{r}(\pi^\prime_{m,L}),D)\right)^2 \mid \mathcal{H}_m \right] \nonumber\\
& \leq \frac{1}{g^4} \left( \mathbb{E}\left[ \left(\bar{a}(\mathsf{r}(\pi_{m,L}),D) + \bar{a}(\mathsf{r}(\pi^\prime_{m,L}),D)\right)^4 \mid \mathcal{H}_m \right] \right)^{1/2} \nonumber\tag{Cauchy-Schwarz}\\
& \quad\quad\cdot \left( \mathbb{E}\left[ \left(\bar{a}(\mathsf{r}(\pi_{m,L}),D) - \bar{a}(\mathsf{r}(\pi^\prime_{m,L}),D)\right)^4 \mid \mathcal{H}_m \right] \right)^{1/2} \nonumber\\
& \leq \frac{2^4}{g^4} \left( \mathbb{E}\left[ \left(\bar{a}(\mathsf{r}(\pi_{m,L}),D) - \bar{a}(D)\right)^4 \right] \right)^{1/2} \nonumber\tag{Hölder}\\
& \quad\quad\cdot \left( \mathbb{E}\left[ \left(\bar{a}(\mathsf{r}(\pi_{m,L}),D) - \bar{a}(\mathsf{r}(\pi^\prime_{m,L}),D)\right)^4 \mid \mathcal{H}_m \right] \right)^{1/2}~.\label{eq: diffof squares}
\end{flalign}
Observe that
\[
\mathbb{E}\left[ \left(\bar{a}(\mathsf{r}(\pi_{m,L}),D) - \bar{a}(D)\right)^4 \right]
\leq 3^2 \left(2^4 (2-\varphi k - \varphi)^2\right)^2 \Gamma_{k,\varphi,b}^{(2)}
\]
by \cref{cor: Moment bound}, as Assumption \ref{assu: invariance} is maintained and the eighth-order sample-split stability $\zeta^{(8)}$ is finite. In turn, we have that
\begin{equation}
\label{eq: v tilde no H bound}
\mathbb{E}\left[ \left(\bar{a}(\mathsf{r}(\pi_{m,L}),D) - \bar{a}(\mathsf{r}(\pi^\prime_{m,L}),D)\right)^4 \mid \mathcal{H}_m \right]
\leq \frac{4}{k^4b^4} \Gamma_{k,\varphi,b}^{(2)}
\end{equation}
by \eqref{eq: break up square}, \eqref{eq: contibute sigma}, and \eqref{eq: contribute sigma train}. Hence, we have that
\begin{equation}
\mathbb{E}\left[\left(\tilde{v}_{g,k}(Z_{m})-\tilde{v}_{g,k}(Z_{m}^{\prime})\right)^{2}
 \mid\mathcal{H}_{m},\bm{\pi}\right]
 \leq
 \frac	{3\cdot 2^9}{k^2 b^2g^4} (2-\varphi k - \varphi)^2 \Gamma_{k,\varphi,b}^{(2)}~. \label{eq: v tilde H bound}
\end{equation}
Combining \eqref{eq: var conc cond bound},  \eqref{eq: v tilde no H bound}, \eqref{eq: v tilde H bound}, we find that
\[
\mathbb{E}\left[\mathbb{E}\left[\tilde{v}_{g,k}(Z_{m})-\tilde{v}_{g,k}(Z_{m}^{\prime})\mid Z,Z^{\prime}\right]^{2}\mid\bm{\pi}\right]
\leq 
\left(1-\frac{2}{gn^{2}}\right)^{2m}\left(\frac{\left(3\cdot 2^{11}\right)\left(2-\varphi k-\varphi\right)^4}{n^2 g^4}\right) \Gamma_{k,\varphi,b}^{(2)}~,
\]
which completes the proof of the first part of the Lemma. 

Second, following the same argument, we again have that
\begin{flalign}
 & \mathbb{E}\left[\mathbb{E}\left[\check{v}_{g,k}(Z_{m})-\check{v}_{g,k}(Z_{m}^{\prime})\mid Z,Z^{\prime}\right]^{2}\mid\bm{\pi}\right]\nonumber \\
 & \leq\left(1-\frac{2}{gn^{2}}\right)^{2m}\left(\frac{2kb\left(2n-bk-b\right)}{n^{2}}\right)^{2}\mathbb{E}\left[\left(\check{v}_{g,k}(Z_{m})-\check{v}_{g,k}(Z_{m}^{\prime})\right)^{2}\mid\mathcal{H}_{m},\bm{\pi}\right]\label{eq: a2 conc cond bound check}
\end{flalign}
In this case, we can compute
\begin{flalign}
 & \mathbb{E}\left[\left(\check{v}_{g,k}(Z_{m})-\check{v}_{g,k}(Z_{m}^{\prime})\right)^{2}
 \mid\mathcal{H}_{m},\bm{\pi}\right]\nonumber\\
 & = \mathbb{E}\left[
\left(\left(a(Z_m) - \bar{a}(D) \right)^2 - 
       \left(a(Z^\prime_m) - \bar{a}(D) \right)^2\right)^2 \mid \mathcal{H}_m \right] \nonumber\\
& = \mathbb{E}\left[ 
\left(a(Z_m) - \bar{a}(D) + a(Z^\prime_m) - \bar{a}(D)\right)^2
\left(a(Z_m) - a(Z^\prime_m)\right)^2 \mid \mathcal{H}_m \right] \nonumber\\
& \leq \left( \mathbb{E}\left[ \left(a(Z_m) - \bar{a}(D) + a(Z^\prime_m) - \bar{a}(D)\right)^4 \mid \mathcal{H}_m \right] \right)^{1/2} \nonumber\tag{Cauchy-Schwarz}\\
& \quad\quad\cdot \left( \mathbb{E}\left[ \left(a(Z_m) - a(Z^\prime_m)\right)^4 \mid \mathcal{H}_m \right] \right)^{1/2} \nonumber\\
& \leq 2^4 \left( \mathbb{E}\left[ \left(a(Z_m) - \bar{a}(D)\right)^4 \right] \right)^{1/2} \left( \mathbb{E}\left[ \left(a(Z_m) - a(Z^\prime_m)\right)^4 \mid \mathcal{H}_m \right] \right)^{1/2}\label{eq: diffof squares check}~,
\end{flalign}
where the last inequality follows from Hölder's inequality. Again we have that
\begin{equation}
\label{eq: v tilde no H bound check}
\mathbb{E}\left[ \left(a(Z_m) - a(Z^\prime_m)\right)^4 \right]
\leq 3^2 \left(\frac{2^4 \varphi^2 (2-\varphi k - \varphi)^2}{g}\right)^2 \Gamma_{k,\varphi,b}^{(2)}
\end{equation}
by \cref{cor: Moment bound}, as Assumption \ref{assu: invariance} is maintained and the eighth-order sample-split stability $\zeta^{(8)}$ is finite. Similarly. we have that 
\begin{equation}
\label{eq: v tilde H bound check}
\mathbb{E}\left[ \left(a(Z_m) - \bar{a}(D)\right)^4 \mid \mathcal{H}_m \right]
\leq \frac{4}{g^4 k^4 b^4} \Gamma_{k,\varphi,b}^{(2)}
\end{equation}
by \eqref{eq: break up square}, \eqref{eq: contibute sigma}, and \eqref{eq: contribute sigma train}. Combining \eqref{eq: a2 conc cond bound check},  \eqref{eq: v tilde no H bound check}, \eqref{eq: v tilde H bound check}, we find that
\[
\mathbb{E}\left[\left(\check{v}_{g,k}(Z_{m})-\check{v}_{g,k}(Z_{m}^{\prime})\right)^{2}
 \mid\mathcal{H}_{m},\bm{\pi}\right]
\leq 
\left(1-\frac{2}{gn^{2}}\right)^{2m}\left(\frac{\left(3\cdot 2^9\right)\varphi^4 \left(2-\varphi k-\varphi\right)^4}{n^2 g^3}\right) \Gamma_{k,\varphi,b}^{(2)}~,
\]
which completes the proof of the second part of the Lemma. \hfill\qed
 
\subsubsection{Proof of Lemma \ref{lem: star hat diff}, Part (i)}

Observe that 
\begin{flalign}
 & P\left\{ \vert U(\mathsf{R}_{\hat{g},k},D)\vert\geq t\sqrt{v_{g^{\star},k}\left(D\right)},\vert\hat{g}-g^{\star}\vert\leq cg^{\star}\mid D\right\} \nonumber \\
 & =P\left\{ \sum_{i=1}^{g}\bar{a}\left(\mathsf{r}_{i,k},D\right)-\sum_{i=1}^{g^{\star}}\bar{a}\left(\mathsf{r}_{i,k},D\right)\geq tg^{\star}\sqrt{v_{g^{\star},k}\left(D\right)},\vert\hat{g}-g^{\star}\vert\leq cg^{\star}\mid D\right\} \nonumber \\
 & \leq P\left\{ \max_{g^{\star}(1-c)\leq g\leq g^{\star}}\vert\sum_{i=1}^{g}\bar{a}\left(\mathsf{r}_{i,k},D\right)-\sum_{i=1}^{g^{\star}}\bar{a}\left(\mathsf{r}_{i,k},D\right)\vert\geq tg^{\star}\sqrt{v_{g^{\star},k}\left(D\right)}\mid D\right\} \nonumber \\
 & +P\left\{ \max_{g^{\star}\leq g\leq g^{\star}(1+c)}\vert\sum_{i=1}^{g}\bar{a}\left(\mathsf{r}_{i,k},D\right)-\sum_{i=1}^{g^{\star}}\bar{a}\left(\mathsf{r}_{i,k},D\right)\vert\geq tg^{\star}\sqrt{v_{g^{\star},k}\left(D\right)}\mid D\right\} \nonumber \\
 & =P\left\{ \max_{g^{\star}(1-c)\leq g\leq g^{\star}}\vert\sum_{i=1}^{g^{\star}-g}\bar{a}\left(\mathsf{r}_{i,k},D\right)\vert\geq tg^{\star}\sqrt{v_{g^{\star},k}\left(D\right)}\mid D\right\} \nonumber \\
 & +P\left\{ \max_{g^{\star}\leq g\leq g^{\star}(1+c)}\vert\sum_{i=1}^{g-g^{\star}}\bar{a}\left(\mathsf{r}_{i,k},D\right)\vert\geq tg^{\star}\sqrt{v_{g^{\star},k}\left(D\right)}\mid D\right\} .\label{eq: maximal probabilities}
\end{flalign}
To bound the two probabilities (\ref{eq: maximal probabilities}),
we apply the following Chernoff-type variation to the Kolmogorov maximal
inequality, due to \citet{steiger1970bernstein}. 
\begin{theorem}[\citealp{steiger1970bernstein}]
\label{thm: Steiger}Let $S_{i}$, $i=1,2,\ldots$, be a real-valued
martingale sequence. If the moment generating function
\[
m_{n}\left(\theta\right)=\mathbb{E}\left[\exp\left(\theta S_{n}\right)\right]
\]
is finite for all positive $\theta$, then the inequality
\[
\log\text{ }P\left\{ \max_{1\leq n^{\prime}\leq n}S_{n^{\prime}}>t\right\} \leq\inf_{\theta>0}\left(\log m_{n}\left(\theta\right)-\theta t\right),
\]
holds. 
\end{theorem}
Conditional on $D$, the random variables 
\[
\bar{a}\left(\mathsf{r}_{i,k},D\right),\quad i=1,2,\ldots,
\]
are mean-zero, independent, and identically distributed. Thus, the
partial sums 
\[
S_{m}=\sum_{i=1}^{m}\bar{a}\left(\mathsf{r}_{i,k},D\right)
\]
are a martingale sequence. Moreover, though inspection of the proof
of \cref{thm: chatterjee concentration}, we find that Theorem
\cref{thm: cross split concentration} implies that 
\begin{flalign*}
 & \inf_{\theta>0}\left(\log\mathbb{E}\left[\exp\left(\theta S_{\lfloor cg^{\star}\rfloor}\right)\mid D\right]-\theta\tau\right)\\
 & =\inf_{\theta>0}\left(\log\mathbb{E}\left[\exp\left(\frac{\theta}{\lfloor cg^{\star}\rfloor}S_{\lfloor cg^{\star}\rfloor}\right)\mid D\right]-\theta\frac{\tau}{\lfloor cg^{\star}\rfloor}\right)\\
 & \leq - \frac{\lfloor cg^{\star}\rfloor}{2^4 (2-\varphi k - \varphi)^2 \Gamma^{(1)}_{k,\varphi,b} } \frac{\delta \tau^2}{\left(\lfloor cg^{\star}\rfloor\right)^2}\\
 & \leq - \frac{\delta}{2^4 (2-\varphi k - \varphi)^2 \Gamma^{(1)}_{k,\varphi,b}} \frac{\tau^2}{ cg^{\star}}
 \end{flalign*}
with probability greater than $1-\delta$, as \cref{assu: invariance}
holds and the fourth-order split stability $\zeta^{(4)}$ is finite.
Thus, by setting 
\[
\tau=tg^{\star}\sqrt{v_{g^{\star},k}\left(D\right)}=\sqrt{g^{\star}}t\left(v_{1,k}(D)\right)^{1/2}
\]
we find that 
\begin{flalign*}
 & \log \frac{1}{2} P\left\{ \vert U(\mathsf{R}_{\hat{g},k},D)\vert\geq t\sqrt{v_{g^{\star},k}\left(D\right)},\vert\hat{g}-g^{\star}\vert\leq cg^{\star}\mid D\right\} \\
 & \leq - \frac{\delta v_{1,k}(D)}{2^4 (2-\varphi k - \varphi)^2 \Gamma^{(1)}_{k,\varphi,b}} \frac{t^2}{c}
\end{flalign*}
with probability greater than $1-\delta$, by \cref{thm: Steiger}
and the inequality \eqref{eq: maximal probabilities}.

\subsubsection{Proof of \cref{lem: star hat diff}, Part (ii)}

Observe that 
\begin{equation}
P\left\{ \vert\hat{g}-g^{\star}\vert>cg^{\star}\mid D\right\} =P\left\{ g^{\star}\left(1+c\right)<\hat{g}\mid D\right\} +P\left\{ \hat{g}<g^{\star}\left(1-c\right)\mid D\right\} .\label{eq: break up k}
\end{equation}
We begin by handling the first term. We have that
\begin{flalign*}
& P\left\{ g^{\star}\left(1+c\right)<\hat{g}\mid D\right\}  \\
& \leq P\left\{ \hat{v}_{g^{\star}\left(1+c\right),k}\left(Z\right)-v_{g^{\star}\left(1+c\right),k}\left(D\right)+v_{g^{\star}\left(1+c\right),k}\left(D\right)>\frac{1}{2}\left(\frac{\xi}{z_{1-\beta/2}}\right)^{2}\mid D\right\} \\
 & \leq P\left\{ \hat{v}_{g^{\star}\left(1+c\right),k}\left(Z\right)-v_{g^{\star}\left(1+c\right),k}\left(D\right)>\frac{v_{1,k}(D)}{g^{\star}}-\frac{v_{1,k}(D)}{g^{\star}\left(1+c\right)}\mid D\right\} \\
 & =P\left\{ \hat{v}_{g^{\star}\left(1+c\right),k}\left(Z\right)-v_{g^{\star}\left(1+c\right),k}\left(D\right)>\frac{v_{1,k}(D)}{g^{\star}}\frac{c}{1+c}\mid D\right\} ,
\end{flalign*}
where the first inequality follows the definition of $g^{\star}$.
Thus, as \cref{assu: invariance} holds and the eighth-order
split stability $\zeta^{(8)}$ is finite, we have that 
\begin{flalign*}
\log \frac{1}{4} P\left\{ g^{\star}\left(1+c\right)<\hat{g}\mid D\right\} 
  &  \lesssim -\frac{\delta (v_{1,k}(D))^2 }{(2-\varphi k - \varphi)^4} \frac{g^\star(1+c) c^2}{\Gamma_{k,\varphi,b}^{(2)}} \\
 & \lesssim -\frac{\delta (v_{1,k}(D))^3 }{(2-\varphi k - \varphi)^4} \frac{z^2_{1-\beta/2} c^2}{\xi^2 \Gamma_{k,\varphi,b}^{(2)}}
\end{flalign*}
with probability greater than $1-\delta$, by \cref{thm: variance estimator concentration},
the definition of $g^{\star}$ and the assumption that $0<c<1/2$. 

Next, we bound the second term in \eqref{eq: break up k}. Define
the partial sums
\begin{flalign*}
A_{g} & =\sum_{l=1}^{g}\left(\frac{1}{k^{2}}\sum_{i,i^{\prime}=1}^{k}\left(T(\mathsf{s}_{l,i},D)-\bar{a}\left(D\right)\right)\left(T(\mathsf{s}_{l,i^{\prime}},D)-\bar{a}\left(D\right)\right)-v_{1,k}(D)\right)\quad\text{and}\\
B_{g} & =\sum_{l=1}^{g}\Bigg(\frac{1}{k^{2}}\sum_{i,i^{\prime}=1}^{k}\left(T(\mathsf{s}_{l,i},D)-\bar{a}\left(D\right)\right)\left(T(\mathsf{s}_{l,i^{\prime}},D)-\bar{a}\left(D\right)\right)\\
 & \quad\quad\quad\quad\quad+2\sum_{l^{\prime}=1}^{l-1}\left(T(\mathsf{s}_{l,i},D)-\bar{a}\left(D\right)\right)\left(T(\mathsf{s}_{l^{\prime},i^{\prime}},D)-\bar{a}\left(D\right)\right)-v_{1,k}(D)\Bigg).
\end{flalign*}
Observe that the equality 
\[
\hat{v}_{g,k}\left(Z\right)-v_{g,k}\left(D\right)=\frac{g}{g-1}\left(\tilde{v}_{g,k}\left(Z\right)-v_{g,k}\left(D\right)\right)-\frac{1}{g-1}\left(\check{v}_{g,k}\left(Z\right)-v_{g,k}\left(D\right)\right),
\]
from the proof of \cref{thm: variance estimator concentration},
implies that
\begin{flalign*}
\hat{v}_{g,k}\left(Z\right)-v_{g,k}\left(D\right) & =\frac{1}{g\left(g-1\right)}\sum_{l=1}^{g}\frac{1}{k^{2}}\sum_{i,i^{\prime}=1}^{k}\left(T(\mathsf{s}_{l,i},D)-a\left(Z\right)\right)\left(T(\mathsf{s}_{l,i^{\prime}},D)-a\left(Z\right)\right)-\frac{v_{1,k}(D)}{g}\\
 & =\frac{1}{g-1}\frac{1}{g}A_{g}-\frac{1}{g-1}\frac{1}{g^{2}}B_{g}
\end{flalign*}
for all positive $g$. Thus, we can write
\begin{flalign*}
 & P\left\{ \hat{g}<g^{\star}\left(1-c\right)\mid D\right\} \\
 & =P\left\{ \min_{2\leq g^{\prime}\leq g^{\star}\left(1-c\right)}\hat{v}_{g^{\prime},k}\left(Z\right)\leq\left(\frac{\xi}{z_{1-\beta/2}}\right)^{2}\mid D\right\} \\
 & =P\left\{ \min_{2\leq g^{\prime}\leq g^{\star}\left(1-c\right)}\frac{\left(g^{\prime}A_{g^{\prime}}-B_{g^{\prime}}\right)}{\left(g^{\prime}-1\right)\left(g^{\prime}\right)^{2}}+\frac{1}{g^{\prime}}v_{1,k}(D)\leq\left(\frac{\xi}{z_{1-\beta/2}}\right)^{2}\mid D\right\} ,
\end{flalign*}
where the first equality follows from the definition of $\hat{g}$.
Now, in the event that 
\begin{equation}
\min_{2\leq g^{\prime}\leq g^{\star}\left(1-c\right)}\frac{\left(g^{\prime}A_{g^{\prime}}-B_{g^{\prime}}\right)}{\left(g^{\prime}-1\right)\left(g^{\prime}\right)^{2}}+\frac{1}{g^{\prime}}v_{1,k}(D)\leq\left(\frac{\xi}{z_{1-\beta/2}}\right)^{2},\label{eq: k star k hat as max}
\end{equation}
it must also be the case that
\begin{flalign}
\min_{2\leq g^{\prime}\leq g^{\star}\left(1-c\right)}\left(g^{\prime}-1\right)\left(g^{\prime}\right)^{2}\left(\frac{\left(g^{\prime}A_{g^{\prime}}-B_{g^{\prime}}\right)}{\left(g^{\prime}-1\right)\left(g^{\prime}\right)^{2}}+\frac{1}{g^{\prime}}v_{1,k}(D)\right)\leq & (\hat{g}-1)(\hat{g})^{2}\left(\frac{\xi}{z_{1-\beta/2}}\right)^{2}\label{eq: hat k fact}
\end{flalign}
as $\hat{g}\leq g^{\star}\left(1-c\right)$ necessarily. But then,
similarly, \eqref{eq: hat k fact} implies that
\begin{flalign*}
\min_{2\leq g^{\prime}\leq g^{\star}\left(1-c\right)}\left(g^{\prime}A_{g^{\prime}}-B_{g^{\prime}}\right)+\frac{(\hat{g}-1)(\hat{g})^{2}}{g^{\prime}}v_{1,k}(D) & \leq(\hat{g}-1)(\hat{g})^{2}\left(\frac{\xi}{z_{1-\beta/2}}\right)^{2},
\end{flalign*}
and in turn
\begin{flalign*}
 & \min_{2\leq g^{\prime}\leq g^{\star}\left(1-c\right)}\frac{\left(g^{\prime}A_{g^{\prime}}-B_{g^{\prime}}\right)}{\left(g^{\star}\left(1-c\right)-1\right)\left(g^{\star}\left(1-c\right)\right)^{2}}\\
 & \leq\frac{(\hat{g}-1)(\hat{g})^{2}}{\left(g^{\star}\left(1-c\right)-1\right)\left(g^{\star}\left(1-c\right)\right)^{2}}\left(\frac{1}{g^{\star}-1}-\frac{1}{g^{\star}\left(1-c\right)}\right)v_{1,k}(D)\\
 & =\frac{(\hat{g}-1)(\hat{g})^{2}}{\left(g^{\star}\left(1-c\right)-1\right)\left(g^{\star}\left(1-c\right)\right)^{2}}\left(\frac{1-g^{\star}c}{g^{\star}\left(g^{\star}-1\right)\left(1-c\right)}\right)v_{1,k}(D),
\end{flalign*}
are then also true. Finally, again as \eqref{eq: k star k hat as max}
is equivalent to $\hat{g}<g^{\star}\left(1-c\right)$, we have that
\begin{flalign*}
 & P\left\{ \min_{2\leq g^{\prime}\leq g^{\star}\left(1-c\right)}\frac{\left(g^{\prime}A_{g^{\prime}}-B_{g^{\prime}}\right)}{\left(g^{\prime}-1\right)\left(g^{\prime}\right)^{2}}+\frac{1}{g^{\prime}}v_{1,k}(D)\leq\left(\frac{\xi}{z_{1-\beta/2}}\right)^{2}\mid D\right\} \\
 & \leq P\left\{ \min_{2\leq g^{\prime}\leq g^{\star}\left(1-c\right)}\frac{g^{\prime}A_{g^{\prime}}-B_{g^{\prime}}}{\left(g^{\star}\left(1-c\right)-1\right)\left(g^{\star}\left(1-c\right)\right)^{2}}\leq\left(\frac{1-g^{\star}c}{g^{\star}\left(g^{\star}-1\right)\left(1-c\right)}\right)v_{1,k}(D)\mid D\right\} .
\end{flalign*}
Now, observe that we can write
\begin{flalign}
 & P\left\{ \min_{2\leq g^{\prime}\leq g^{\star}\left(1-c\right)}\frac{g^{\prime}A_{g^{\prime}}-B_{g^{\prime}}}{\left(g^{\star}\left(1-c\right)-1\right)\left(g^{\star}\left(1-c\right)\right)^{2}}\leq\left(\frac{1-g^{\star}c}{g^{\star}\left(g^{\star}-1\right)\left(1-c\right)}\right)v_{1,k}(D)\mid D\right\} \nonumber \\
 & = P\left\{ \max_{2\leq g^{\prime}\leq g^{\star}\left(1-c\right)}\frac{B_{g^{\prime}}-g^{\prime}A_{g^{\prime}}}{\left(g^{\star}\left(1-c\right)-1\right)\left(g^{\star}\left(1-c\right)\right)^{2}}\geq\left(\frac{g^{\star}c-1}{g^{\star}\left(g^{\star}-1\right)\left(1-c\right)}\right)v_{1,k}(D)\mid D\right\}~.\nonumber
\end{flalign}
We bound this term by combining the argument used
to establish \cref{thm: variance estimator concentration}
with an application of \cref{thm: Steiger}. To this end, observe
that 
\begin{flalign}
 & P\left\{ \max_{2\leq g^{\prime}\leq g}\frac{B_{g^{\prime}}-g^{\prime}A_{g^{\prime}}}{\left(g-1\right)g^{2}}\leq t\mid D\right\} \nonumber \\
 & \geq P\left\{ \max_{2\leq g^{\prime}\leq g}\frac{-1}{g-1}\frac{g^{\prime}}{g^{2}}A_{g^{\prime}}\leq\frac{1}{1+\sqrt{g}}t\mid D\right\} +P\left\{ \max_{2\leq g^{\prime}\leq g}\frac{1}{g-1}\frac{1}{g^{2}}B_{g^{\prime}}\leq\frac{\sqrt{g}}{1+\sqrt{g}}t\mid D\right\} -1\nonumber \\
 & \geq P\left\{ \max_{2\leq g^{\prime}\leq g}\frac{g}{g-1}\frac{-1}{g^{2}}A_{g^{\prime}}\leq\frac{\sqrt{g}}{1+\sqrt{g}}t\mid D\right\} +P\left\{ \max_{2\leq g^{\prime}\leq g}\frac{1}{g-1}\frac{1}{g^{2}}B_{g^{\prime}}\leq\frac{1}{1+\sqrt{g}}t\mid D\right\} -1\nonumber \\
 & \geq P\left\{ \max_{2\leq g^{\prime}\leq g}-A_{g^{\prime}}\leq g^{2}\frac{1}{\sqrt{g}}\frac{g-1}{1+\sqrt{g}}t\mid D\right\} +P\left\{ \max_{2\leq g^{\prime}\leq g}\vert B_{g^{\prime}}\vert\leq g^{2}\frac{g-1}{1+\sqrt{g}}t\mid D\right\} -1\label{eq: break into A and B pieces}
\end{flalign}
for any $t>0$ and any positive integer $D$. Observe that the partial
sums $A_{g}$ and $B_{g}$ are both martingale sequences. As \cref{assu: invariance} holds and the eighth-order split stability
$\zeta^{(8)}$ is finite, though inspection of the proof of \cref{thm: chatterjee concentration}, \cref{lem: a(Z) sq concentration} implies that
implies that 
\begin{flalign}
 & \inf_{\theta>0}\left(\log\mathbb{E}\left[\exp\left(\theta A_{g}\right)\mid D\right]-\theta g^{2}\frac{1}{\sqrt{g}}\frac{g-1}{1+\sqrt{g}}t\right)\nonumber \\
 & =\inf_{\theta>0}\left(\log\mathbb{E}\left[\exp\left(\frac{\theta}{g^{2}}A_{g}\right)\mid D\right]-\theta\frac{1}{\sqrt{g}}\frac{g-1}{1+\sqrt{g}}t\right)\nonumber \\
 & \lesssim - \frac{\delta}{(2 - \varphi k - \varphi)^4} \frac{g^3 t^2}{\Gamma_{k,\varphi,b}^{(2)}}\label{eq: A_k bound}
\end{flalign}
and 
\begin{flalign}
 & \inf_{\theta>0}\left(\log\mathbb{E}\left[\exp\left(\theta B_{g}\right)\mid D\right]-\theta g^{2}\frac{g-1}{1+\sqrt{g}}t\right)\nonumber \\
 & =\inf_{\theta>0}\left(\log\mathbb{E}\left[\exp\left(\frac{\theta}{g^{2}}B_{g}\right)\mid D\right]-\theta\frac{g-1}{1+\sqrt{g}}t\right)\nonumber \\
 & \lesssim - \frac{\delta}{(2 - \varphi k - \varphi)^4} \frac{g^3 t^2}{\Gamma_{k,\varphi,b}^{(2)}}\label{eq: B_k bound}
\end{flalign}
each with probability greater than $1-\delta$, where we have used
the facts that $4g\leq\left(1+\sqrt{g}\right)^{2}$ for all $g\geq1$
and $\left(1/4\right)g^{2}\geq\left(g-1\right)^{2}$ for all $g\geq2$.
Hence, by \cref{thm: Steiger}, and plugging \eqref{eq: A_k bound}
and \eqref{eq: B_k bound} into \eqref{eq: break into A and B pieces},
we find that 
\begin{flalign*}
\log \frac{1}{4} P\left\{ \max_{2\leq g^{\prime}\leq g}\frac{B_{g^{\prime}}-g^{\prime}A_{g^{\prime}}}{\left(g^{\prime}-1\right)\left(g^{\prime}\right)^{2}}\geq t\mid D\right\}  
& \lesssim - \frac{\delta}{(2 - \varphi k - \varphi)^4} \frac{g^3 t^2}{\Gamma_{k,\varphi,b}^{(2)}}
\end{flalign*}
with probability greater than $1-\delta$. Consequently, by the definition of $g^{\star}$, we find that 
\begin{flalign*}
 & \log \frac{1}{4} P\left\{ g^{\star}-\hat{g}\geq c\mid D\right\} \\
 & =\log \frac{1}{4}  P\left\{ \max_{2\leq g^{\prime}\leq g^{\star}\left(1-c\right)}\frac{B_{g^{\prime}}-g^{\prime}A_{g^{\prime}}}{\left(g^{\star}\left(1-c\right)-1\right)\left(g^{\star}\left(1-c\right)\right)^{2}}\geq\left(\frac{g^{\star}c-1}{g^{\star}\left(g^{\star}-1\right)\left(1-c\right)}\right)v_{1,k}(D)\mid D\right\} \\
 & \lesssim - \frac{\delta (v_{1,k}(D))^2}{(2 - \varphi k - \varphi)^4} \frac{g^\star(1-c) (g^\star c - 1)^2}{(g^{\star}-1)^2\Gamma_{k,\varphi,b}^{(2)}}\\
 & \lesssim - \frac{\delta (v_{1,k}(D))^2}{(2 - \varphi k - \varphi)^4} \frac{g^\star c^2 }{\Gamma_{k,\varphi,b}^{(2)}}\\
 & \lesssim - \frac{\delta (v_{1,k}(D))^3}{(2 - \varphi k - \varphi)^4} \frac{z^2_{1-\beta/2} c^2 }{\xi^2\Gamma_{k,\varphi,b}^{(2)}}
\end{flalign*}
where in the second to last inequality we have used the facts that
$\frac{1}{2}x\le x-1$ for $x\geq2$ and $g^{\star}c\geq2$. Thus,
putting the pieces together, we find that 
\begin{flalign*}
 & \log \frac{1}{8} P\left\{ \vert\hat{g}-g^{\star}\vert>cg^{\star}\mid D\right\}  
 \lesssim - \frac{\delta (v_{1,k}(D))^3}{(2 - \varphi k - \varphi)^4} \frac{z^2_{1-\beta/2} c^2 }{\xi^2\Gamma_{k,\varphi,b}^{(2)}}
\end{flalign*}
with probability greater than $1-\delta$, as required.\hfill\qed

\end{spacing}
\end{appendix}
\end{document}